\documentclass[sigplan,screen,nonacm]{acmart}

\AtBeginDocument{%
  }

\setcopyright{acmcopyright}
\copyrightyear{2018}
\acmYear{2018}
\acmDOI{XXXXXXX.XXXXXXX}

\acmConference[Conference acronym 'XX]{Make sure to enter the correct
  conference title from your rights confirmation emai}{June 03--05,
  2018}{Woodstock, NY}
\acmPrice{15.00}
\acmISBN{978-1-4503-XXXX-X/18/06}

\usepackage{mdframed}
\usepackage{multicol}
\usepackage{enumerate}
\usepackage[frozencache,cachedir=.,outputdir=.]{minted}
\setminted{fontsize=\small,breaklines=true}
\usepackage{bcprules}
\usepackage{commands}
\usepackage{subfig}
\usepackage{amsthm}
\newtheorem{theorem}{Theorem}[section]
\newtheorem{definition}{Definition}[section]

\newtheorem{lemma}[theorem]{Lemma}
\newtheorem{fact}[theorem]{Fact}
\newtheorem{corollary}[theorem]{Corollary}
\usepackage{thmtools} 
\usepackage{thm-restate}

\definecolor{light-gray}{gray}{0.92}

\begin{document}

\title{Degrees of Separation: A Flexible Type System \mbox{for Data Race Prevention}}

\author{Yichen Xu}
\email{yichen.xu@epfl.ch}
\affiliation{%
  \institution{EPFL}
  \country{}
}

\author{Aleksander Boruch-Gruszecki}
\email{aleksander.boruch-gruszecki@epfl.ch}
\affiliation{%
  \institution{EPFL}
  \country{}
}

\author{Martin Odersky}
\email{martin.odersky@epfl.ch}
\affiliation{%
  \institution{EPFL}
  \country{}
}

\renewcommand{\shortauthors}{Xu et al.}

\newcounter{todocounter}
\newcommand{\TODO}[1]{%
    \refstepcounter{todocounter}%
    \textcolor{orange}{\textbf{[TODO \#\thetodocounter]:} #1}%
}
\newcommand{\PTODO}[1]{%
    \refstepcounter{todocounter}%
}

\newcommand{\theCalculus}{\textsf{CSC}}

\newcommand{\varemptyset}{\varnothing}

\newcommand{\note}[1]{}

\begin{abstract}
  Data races are a notorious problem in parallel programming.
  There has been great research interest in type systems
  that statically prevent data races.
  Despite the progress in the safety and usability of these systems,
  lots of existing approaches enforce strict alias prevention principles to prevent data races.
  The adoption of these principles is often intrusive,
  in the sense that it invalidates common programming patterns and requires paradigm shifts.
  We propose Capture Separation Calculus (System \theCalculus{}),
  a calculus based on Capture Calculus (System \ccformal{}),
  that achieves static data race freedom while being non-intrusive.
  It allows aliasing in general to permit common programming patterns,
  but tracks aliasing and controls them when that is necessary to prevent data races.
  We study the formal properties of System \theCalculus{} by establishing its type safety and data race freedom.
  Notably, we establish the data race freedom property by proving the confluence of its reduction semantics.
  To validate the usability of the calculus,
  we implement it as an extension to the Scala 3 compiler,
  and use it to type-check the examples in the paper.
\end{abstract}





\maketitle

\setlength\abovedisplayskip{0pt}
\setlength\belowdisplayskip{0pt}
\setlength\abovedisplayshortskip{0pt}
\setlength\belowdisplayshortskip{0pt}

\section{Introduction}

Data races arise when mutable state is shared among multiple parallel procedures,
with at least one of them mutating the state.
They are notorious in parallel programming,
because they cause non-deterministic behaviours in parallel programs,
resulting in bugs that are hard to trace and resolve.

Past decades have witnessed extensive research efforts in data race prevention.
Many of these endeavors have focused on developing type systems that statically eliminate data races.
As the root cause of data races is sharing of mutable state,
these type systems employ aliasing control mechanisms to track and regulate aliases to mutable state.
Examples of such mechanisms
include
reference capabilities \cite{pony1,refcap1},
ownership types \cite{ownership1,ownership2},
balloon types \cite{almeida1997balloon,servetto2013balloon},
and Rust's ownership and borrowing system \cite{rustbook}.

While existing approaches
have made significant progress in terms of both safety and usability,
their alias-prevention principles often invalidate common programming patterns.
As a result,
the adoption of these systems is intrusive, 
requiring a shift of programming paradigms.
This can be a deterrent to migrating existing codebases to these systems,
despite the potential benefits they offer in terms of data race safety.
Take Rust, which offers fearless concurrency thanks to its ownership model \cite{rustbook}, as an example.
Consider the following Rust code:
\begin{minted}{rust}
struct Vec2 { x: i32, y: i32 }
fn update<F, G>(p: &mut Vec2, mut f: F, mut g: G)
  where F: FnMut(&i32) -> i32, G: FnMut(&i32) -> i32 {
    p.x = f(&p.x); p.y = g(&p.y);
}
fn main() {
    let mut p = Vec2 { x: 1, y: 2 };
    let mut sum = 0;
    update(&mut p, 
           |&x| { sum += x; x + 1 }, 
           |&y| { sum += y; y + 1 });
}
\end{minted}
The \mintinline{rust}{update} function updates the two fields of \mintinline{rust}{Vec2} in place.
The main function
increments the two fields of a \mintinline{rust}{Vec2} by one and simultaneously calculates their sum.
It cannot pass borrow checking as \mintinline{rust}{sum} is mutably borrowed by both closures.
Arguably, this is a reasonable programming pattern.
Aliasing of mutable state in this example is innocuous in terms of data races,
since the two closures are executed sequentially.

While we acknowledge that the benefits of Rust's alias prevention principle extend beyond data race prevention,
playing a vital role in Rust's garbage-collection-free safe memory management,
we argue that data race freedom alone is a significant safety guarantee.
Therefore, 
there is value in exploring a more flexible and non-intrusive type system
that focuses primarily on data race prevention
while being permissive enough to accommodate common programming patterns.
With a better balance between safety and usability,
such type systems facilitates the migration of existing codebases,
thereby promoting broader adoption.

The intrusiveness of the existing approaches can be attributed to the global alias prevention principles they enforce.
Most of the systems maintain an anti-aliasing invariant:
Rust's single ownership rule and the uniqueness of mutable borrows \cite{rustbook},
and the modes on object references in reference capabilities \cite{refcap1}.
A more permissive alternative to such a paradigm
is to stray from enforcing an anti-aliasing invariant globally.
Instead, aliases to mutable state can be generally allowed,
but are tracked so that they can be regulated when necessary.
This constitutes a \emph{control-as-you-need} paradigm.

Following this paradigm,
we propose Capture Separation Calculus (\theCalculus{})
as a non-intrusive approach for data race prevention.
It statically prevents data races
while being flexible and permissive enough to accommodate common programming patterns.
Our approach builds upon Capture Calculus (\ccformal{}) \cite{cc1},
which is originally proposed as an approach to effect checking.
As a lightweight extension to System F$_{<:}$,
\ccformal{} introduces a minimal set of constructs,
while supporting expressive alias tracking with its \emph{capturing types}.
These characteristics make it an ideal foundation for our work.
\theCalculus{} introduces
aliasing control mechanisms on top of \ccformal{},
which regulate mutable state aliases when data races are possible.
This aligns with the \emph{control-as-you-need} paradigm.
In addition, mutable variables and parallelly-executed let bindings are introduced
to model parallel programs with shared mutable states.

Now we demonstrate the \emph{non-intrusiveness} of the calculus with examples.
As \theCalculus{} has been implemented as an extension in the Scala 3 compiler,
we show the examples in Scala.
Consider the following Scala code, which is equivalent to the Rust example provided earlier:
\begin{minted}{scala}
case class Vec2(var x: Int, var y: Int)
def update(p: Vec2, 
           f: (a: Int) => Int, 
           g: (a: Int) => Int): Unit =
  p.x = f(p.x); p.y = g(p.y)
def main(): Unit =
  val p = new Vec2(1, 2)
  val sum = new Ref(0)
  update(p, x => { sum.update(_ + 1); x + 1 }, 
            y => { sum.update(_ + 1); y + 1 })
\end{minted}
This program is well-typed.
When \theCalculus{} is introduced,
the program remains well-typed without any modification.
The two aliases to \mintinline{scala}{sum} are allowed
since the two closures are executed sequentially,
and no data races can occur.
In contrast,
the following program \emph{should} be rejected
as it \emph{does} incur a data race:
\begin{minted}{scala}
def parupdate(p: Vec2, 
              f: (a: Int) => Int, 
              g: (a: Int) => Int): Unit =
  p.x = f(p.x) || p.y = g(p.y)
def main(): Unit =
  // ...
  parupdate(p, x => { sum.update(_ + 1); x + 1 }, 
               y => { sum.update(_ + 1); y + 1 })
\end{minted}
The only difference between \mintinline{scala}{parupdate}
and \mintinline{scala}{update} is that
the two closures are executed in parallel in \mintinline{scala}{parupdate}.
In fact,
\mintinline{scala}{p.x = f(p.x) || p.y = g(p.y)}
is ill-typed in \theCalculus{}
because, to execute two operations in parallel,
the \emph{separation} between them must be established.
To fix this error, the \mintinline{scala}{parupdate} function should
declare its parameters as \texttt{\textbf{sep}}arated:
\begin{minted}{scala}
def parupdate(p: Vec2, 
       sep|\texttt{\color{gray}\{p\}}| f: (a: Int) => Int, 
     sep|\texttt{\color{gray}\{p,f\}}| g: (a: Int) => Int): Unit = ...
\end{minted}
Here, \mintinline{scala}{sep{p} f: ...} indicates that
\mintinline{scala}{f} should be separated from \mintinline{scala}{p}.
The set \mintinline{scala}{{p}} is called the \emph{separation degree}.
In \theCalculus{}, \emph{separation} does not imply that
the variables referred to by two parties are non-overlapping.
Instead, it denotes \emph{non-interference}:
if \texttt{f} is separated from \texttt{p},
there is no mutable state
referred to by both \texttt{f} and \texttt{p}
and mutated by either of them.
Similarly, the annotation \mintinline{scala}{sep{p,f}} signifies that
\mintinline{scala}{g} is separated from both \mintinline{scala}{p} and \mintinline{scala}{f}.
After annotating the parameters of \mintinline{scala}{parupdate} with \texttt{\textbf{sep}},
the type error shifts to the call site: 
it is required that \mintinline{scala}{g} is separated from \mintinline{scala}{f},
but both closures mutate \mintinline{scala}{sum}.

The two examples highlight the {non-intrusiveness} of \theCalculus{}:
for sequential programs where data races cannot occur,
they remain unchanged and well-typed;
for parallel programs with shared mutable states,
the alias controlling mechanism ensures data race freedom,
and the additional annotations are concise and informative.
Moreover, the experimental implementation of \theCalculus{} in the Scala 3 compiler
supports \emph{separation degree inference}.
In practice,
this means that users are not required to explicitly specify the separation degrees
of the function arguments
(colored in gray in the example above),
as the compiler is capable of inferring them automatically.
{
Internally,
during the typing process of a function,
the compiler gathers the constraints on the omitted separation degrees and solves the constraints incrementally.
Importantly, the inference is localized.
Once a function is typed, the inferred separation degrees of its arguments are frozen.
}

\note{\color{orange}
You may notice that separation checking bears a resemblance to \emph{framing rules} \cite{refcap1}.
Indeed,
when type-checking a parallel expression \texttt{e1 || e2},
the variable referenced\footnote{More precisely, it should be the \textsf{cv} of a term.} by the two sides are required to be separated.
As the variables referenced by a term delineate a sub-domain of the typing context that is used when type-checking the term,
it is like slicing the typing context into two parts
and type-check each side with the its own portion of the context,
which is parallel to the framing rules.
However, this analogy is imprecise when taking boxes and read-only references into account.
}

In the metatheory,
we establish type safety of the calculus 
by the standard progress and preservation theorems.
Furthermore,
we formally prove data race freedom
by establishing the confluence of reductions.
Specifically, 
despite the possible arbitrary interleaving of the evaluation of the binding and the body in a parallel let binding, 
the program reduces to a deterministic result.

\section{Key Ideas}
\label{sec:informal}

Now we develop the ideas of System \theCalculus{} informally.

\subsection{Capturing Types As an Alias Tracking Device}


\ccshort{} tracks effects by tracking in types the capabilities captured by values.
Capabilities are just program variables in the calculus.
\ccshort{} introduces \emph{capturing types},
which are in the form of $T\capt \set{x_1, \cdots, x_n}$,
where $T$ is a \emph{shape} type that describes the shape of the value
(e.g. an integer, a function, etc.),
and $\set{x_1, \cdots, x_n}$ is the \emph{capture set}
which statically predicts an upper bound of the variables captured by the value.
For instance, the following closure performs an I/O effect via the \texttt{console} capability:
\begin{minted}{scala}
() => console.println("Hello")
\end{minted}
Its type is \mintinline{scala}{() ->{console} Unit},
\footnote{Here \mintinline{scala}{A ->{x1,...,xn} B.} is a shorthand for writing a capturing function type, being translated to \mintinline{scala}{(A -> B)^{x1,...,xn}}.}
indicating that the closure at most accesses the capability \texttt{console}.
This essentially tracks the effect of the closure.

In \ccshort{}, capabilities form a hierarchy,
wherein each capability is derived from some existing and more permissive capabilities.
All capabilities are ultimately derived from the special root capability \univ{}.
For instance,
given \mintinline{scala}{f : File^{cap}}
and \mintinline{scala}{logger : Logger^{f}},
we say \mintinline{scala}{logger} is derived from \mintinline{scala}{f}.
It obtains the permission to access a file from \mintinline{scala}{f}
and provides the capability for logging.

Capturing types essentially provide a way to track aliases to capabilities.
To capture a capability is to retain a reference to it, or to \emph{alias} it.
For instance,
by viewing a mutable variable reference
as a capability of accessing the mutable state,
capturing types track their aliases.
In the \texttt{Vec2} example,
the closure 
\mintinline{scala}{x => { sum += x; x + 1 }}
has the capturing type \mintinline{scala}{Int ->{sum} Int},
indicating that it \emph{aliases} the mutable state \texttt{sum}.
The aliasing information can then be used to 
regulate shared mutable state access
in order to prevent data races.
We treat the mutable variable \texttt{sum} as a capability for \emph{full} access to itself.
\theCalculus{} introduces \emph{reader capabilities}
for readonly access to mutable states.
A reader capability of a mutable state $x$ 
can be naturally expressed as a capability derived from $x$,
the full-access capability to the mutable state,
in the capability hierarchy.

\subsection{Separation Degrees and Separation Checking}

Separation degree is attached as an additional set to the bindings,
which is the set of variables that a binding is \emph{separated} from.
Intuitively, it describes the \emph{freshness} of a variable.
For instance, when we are allocating new mutable states:\footnote{The \texttt{new Ref} construct used here corresponds to the $\tLetVarM{D}{x}{y}{t}$ term in the formalism.}
\begin{minted}{scala}
val a = new Ref(0)
//  a :{} Ref[Int]^{cap} 
val b = new Ref(0)
//  ... , b :{a} Ref[Int]^{cap} 
val c = new Ref(0)
//  ... , c :{a,b} Ref[Int]^{cap} 
\end{minted}
We can freely specify the separation degree of the newly-allocated variable,
since it is known to be \emph{fresh}, or \emph{unaliased}.

We have seen in the \texttt{parupdate} example that we can declare a separation degree for function arguments.
For instance, in the following function which resets both mutable integers to zero:
\begin{minted}{scala}
def resetBoth(a: Ref[Int]^{cap}, 
              sep{a} b: Ref[Int]^{cap}): Unit =
  a.set(0) || b.set(0)
\end{minted}
the second argument is declared to be separated from the first argument.
When applying the function,
\emph{separation checking}
ensures that the declared separation degree is respected.
Specifically, separation checking works on two capture sets to see whether they are \emph{separated} from each other,
in the sense that the mutation of a variable does not overlap with read/write accesses to that variable.
For instance, \mintinline{scala}{resetBoth(a, a)} is rejected,
as \texttt{\{a\}} possesses the write-access to \texttt{a} which overlaps with itself.

Separation degrees can be used to establish the separation between variables.
For instance, \mintinline{scala}{resetBoth(a, b)}
typechecks since \texttt{a} is in the separation degree of \texttt{b}.
Additionally, 
as the capture set is an upper bound of what a value may alias,
a variable \mintinline{scala}{x}
is separated from another variable \mintinline{scala}{y}
if what \mintinline{scala}{x} captures is separated from \mintinline{scala}{y}.
\begin{minted}{scala}
val d = a        // d :{} Ref[Int]^{a}
resetBoth(d, b)
\end{minted}
In this example, \texttt{d} is considered to be separated from \texttt{b}
by inspecting its capture set \texttt{\{a\}}.

\subsection{Reader Capabilities}

To differentiate between readonly and writeable aliases,
\theCalculus{} incorporates the concept of \emph{reader capabilities}.
This can be exemplified by:
\begin{minted}{scala}
val cr = c.reader  // cr :{} Rdr[Int]^{c}
a.set(cr.get) || b.set(cr.get)
\end{minted}
Here, \texttt{cr} is a reader capability associated with \texttt{c}
possessing the permission of reading the value of \texttt{c}.
The type \mintinline{scala}{Rdr[Int]} marks reader capabilities.
Any two reader capabilities are considered to be separated from each other.
This is because two readonly accesses to mutable states,
even if they access the same mutable state,
will not result in data races.

In \theCalculus{}, there is a special root capability called the \emph{reader root capability} \rdroot{}. It sits at the top of the capturing hierarchy but below the universal capability \univ{}. 
The introduction of \rdroot{} as a special root capability for readers enhances the expressiveness of the system in terms of polymorphism.
For example, considering adding a \texttt{parmap} method to lists
which maps the elements of the list using the function in parallel:
\begin{minted}{scala}
class List[T]:
  def map[U](f: T ->{cap} U): List[U] = ...
  def parmap[U](f: T ->{rdr} U): List[U] = ...
\end{minted}
Since in \mintinline{scala}{parmap} the function \texttt{f} is run in parallel with itself,
mutating any mutable state from it results in race conditions.
Therefore, it is only safe to read from the mutable states,
which is specified by the capture set \texttt{\{rdr\}}.
In \theCalculus{},
\mintinline{scala}{T ->{cap} U} characterizes the \emph{impure} functions
that perform arbitrary effects and write to any mutable states.
These functions can only be run sequentially.
On the other hand, \mintinline{scala}{T ->{rdr} U} characterizes the \emph{read-only} functions
that can only read from the mutable states,
thus can be run in parallel. 

\subsection{Parallel Semantics}

System \theCalculus{} introduces parallel let-bindings.
Its semantics are similar to \emph{futures}:
the binding and body terms are evaluated in parallel,
until the binding value is needed in the body.
At that point, the body waits for the binding value to be evaluated.
When type checking it,
the separation between the binding and the body is required
in order to prevent data races.
Note that the parallel operator we have been using in the above examples \texttt{t1 || t2}
is actually a syntactic sugar for:
\mintinline{scala}{letpar _ = t1 in t2}.

\section{The Calculus}
\label{sec:formal}

In this section, we present the formal definitions of \theCalculus{}.
Figure \ref{fig:all-typing} shows the typing and evaluation rules of \theCalculus{},
with changes from \ccformal{} highlighted in gray.
We assume readers' familarity with \ccformal{} \cite{cc1}
and focus on explaining the changes and additions made in \theCalculus{}.


\subsection{Syntax}
\newcommand{\bbox}[1]{\mbox{\textbf{\textsf{#1}}}}
\newcommand{\OR}{\mathop{\ \ \ |\ \ \ }}

\begin{figure}[htbp]
  \small
  \begin{center}\rule{1.0\linewidth}{0.4pt}\end{center}

  \resizebox{1.05\linewidth}{!}{%
  $\begin{array}[t]{lllll}
    \bbox{Type Variable} & \multicolumn{2}{l}{X, Y, Z} \\
    \bbox{Variable} & \multicolumn{2}{l}{x, y, z, \univ{}, \rdroot{}} \\[0.5em]
    \bbox{Let Kind} & m & ::= & \seqmode \OR \parmode \\


     \bbox{Value} & v, w & ::= & {\lambda(x \new{:_D} T)t}
        \OR \lambda[X <: S]t
        \OR {\tBox{x}} \OR \new{\reader x} \\

    \bbox{Answer} & a & ::= & v \OR x \\

    \bbox{Term} &s, t& ::= & a
          \OR  x\,y
          \OR  x\,[S]
          \OR  {\textsf{let}_{m}}\ {x} = s \ \textsf{in}\ t \OR \tUnbox{C}{x} \\ 
          & & \OR & \new{\tLetVarM{D}{x}{y}{t}}
          \OR \new{\tRead{x}}  \OR  \new{\tWrite{x}{y}} \\ 

    \bbox{Shape Type} & S & ::= & X
          \OR  \top
          \OR  \forall(x \new{:_D} U) T
          \OR  \forall[X <: S] T \OR \tBox{T} \\
          & & \OR &  \new{\tRef{S}} \OR \new{\tRdr{S}} \\

    \bbox{Type} & T, U & ::= & S \OR {S\capt C} \\




    \bbox{Capture Set} & {C} & ::= & \{x_1, \cdots, x_n\} \\
    \bbox{Separation Degree} & \new{D} & ::= & C \quad\text{if $\univ{}, \rdroot \notin C$} \\[0.5em]

	  \bbox{Typing Context} & \G & ::= & \emptyset \OR \G, X <: S \OR \G, x \new{:_{D}} T \quad \text{if $x \notin \set{\univ{}, \rdroot{}}$} \\

    \bbox{Store Context} &
	  \sta  & ::=  & \cdot \OR \sta, \tVal{x} \mapsto v
    \OR \sta, \tVar{x} := v
    \OR \sta, \tSet{x} := v \\
    \bbox{Evaluation Context} & e  & ::=  &  []  
    \OR  \tLetMode{m}{x}{e}{t}
    \OR \tLetMode{\parmode}{x}{t}{e}
    \end{array}$%
  }

\caption{\label{fig:syntax} Syntax of the system.}
\begin{center}\rule{1.0\linewidth}{0.4pt}\end{center}
\end{figure}

The syntax of \theCalculus{} is presented in Figure \ref{fig:syntax}.

\paragraph{Preliminaries}
The types and terms are mostly the same as in \ccformal{}.
The new type constructors and expressions introduced in \theCalculus{} will be explained later.
\theCalculus{} straightforwardly extend the \textsf{cv} function
to account for the new expressions,
as defined in Section \ref{sec:cv}.
Boxes are important formal devices in \ccformal{},
but they are orthogonal to the development of \theCalculus{},
so we refer interested readers to the \ccformal{} paper \cite{cc1} for more details.

\paragraph{Mutable variables}
In \theCalculus{}, mutable variables are bound by the $\tLetVarM{D}{x}{y}{t}$ expression.
Mutable variable references are assigned the type $\tRef{S}$.
Note that we restrict the type of mutable variable contents to be shape types
so that the user cannot leak a local capability by assigning it to a mutable variable.
$\reader x$ creates a reader capability for the mutable variable $x$,
and $\tRdr{S}$ is the type for readers.
$\tWrite{x}{y}$ writes to and $\tRead{x}$ reads from a mutable variable respectively.

\paragraph{Separation degrees}
A separation degree is a set of program variables,
Note that the two root capabilities $\univ$ and $\rdroot$ are excluded,
since including them will not be meaningful.
Separation degrees are attached to the bindings,
both in the typing contexts and in lambda abstractions.
\subsection{Subcapturing and Subtyping}

The extension of \theCalculus{} does not result in any change in the subtyping rules,
but we still show the rules for the completeness of the presentation.
The last two rules \ruleref{sc-rdr-cap} and \ruleref{sc-reader} are the new rules introduced by \theCalculus{}.
They deal with the subcapturing relations regarding the reader root capability \rdroot{}.
The rule \ruleref{sc-rdr-cap} encodes the subcapturing relation between \rdroot{} and the universal root capability \univ{}.
The rule \ruleref{sc-reader} establishes the subcapturing relation between
reader capabilities and the reader root capability \rdroot{}.

\subsection{Separation Checking}
\begin{figure}[htbp]
  \small
\begin{center}\rule{1.0\linewidth}{0.4pt}\end{center}
  \flushleft{\textbf{Separation \quad $\new{\ninter{\G}{C_1}{C_2}}$}}
  \begin{multicols}{3}
    \infrule[\ruledef{ni-symm}]
    {\ninter{\G}{C_1}{C_2}}
    {\ninter{\G}{C_2}{C_1}}

    \infrule[\ruledef{ni-set}]
    {\overline{\ninter{\G}{x_i}{C_2}}^{i=1,\cdots,n}}
    {\ninter{\G}{\set{x_1, \cdots, x_n}}{C_2}}

    \infrule[\ruledef{ni-degree}]
    {x :_{D} T \in \G \\ y \in D}
    {\ninter{\G}{\set{x}}{\set{y}}}

  \end{multicols}
  \begin{multicols}{2}
    \infrule[\ruledef{ni-var}]
    {x :_{D} S\capt C \in \G \andalso \ninter{\G}{C}{y}}
    {\ninter{\G}{\set{x}}{\set{y}}}
    \infrule[\ruledef{ni-reader}]
    {\sub{\G}{\set{x}}{\set{\rdroot}}\\ \sub{\G}{\set{y}}{\set{\rdroot}}}
    {\ninter{\G}{\set{x}}{\set{y}}}
  \end{multicols}
  \vspace{-0.5cm}
  \caption{Separation Checking Rules}
  \label{fig:ninter}
\begin{center}\rule{1.0\linewidth}{0.4pt}\end{center}
\end{figure}

Figure \ref{fig:ninter} shows the separation checking rules.
The \ruleref{ni-symm} and \ruleref{ni-set} in combination states that
the separation between two sets $\{x_1, \cdots, x_n\}$ and $\{y_1, \cdots, y_m\}$
is equivalent to the separation between each pair of $x_i$ and $y_j$,
i.e. $\ninter{\G}{\set{x_i}}{\set{y_j}}$ for all $i \in \set{1,\cdots, n}$ and $j \in \set{1,\cdots, m}$.
The \ruleref{ni-degree} rule states that $x$ and $y$ is separated if $y$ is an element in $x$'s separation degree.
The \ruleref{ni-var} rule makes use of the capture set,
to determine whether two variable are separated.
The last rule \ruleref{ni-reader} states that all readers are separated from each other.
Indeed, two readonly references cannot introduce data races no matter whether they aliases.

The separation between terms, used in \ruleref{let},
is a syntactic sugar of the separation between their \textsf{cv}:
\begin{equation*}
\ninter{\G}{t}{u} \quad \stackrel{\text{def}}{=} \quad \ninter{\G}{{\cv{t} \cap \dom{\G}}}{{\cv{u} \cap \dom{\G}}}
\end{equation*}

\begin{figure*}[htbp]
\small

\flushleft{\textbf{Subcapturing \quad $\subDft{C_1}{C_2}$}}

\vspace{0.3em}

  \begin{multicols}{6}
    \infrule[\ruledef{sc-trans}]
    {\subDft{C_1}{C_2} \\ \subDft{C_2}{C_3}}
    {\subDft{C_1}{C_3}}

    \infrule[\ruledef{sc-var}]
    {x : \tCap{C}{S} \in \G}
    {\subDft{\set{x}}{C}}

    \infrule[\ruledef{sc-elem}]
    {x \in C}
    {\subDft{\set{x}}{C}}

    \infrule[\ruledef{sc-set}]
    {\overline{\subDft{\set{x}}{C_2}}^{x \in C_1}}
    {\subDft{C_1}{C_2}}

    \newruletrue
    \infax[\ruledef{sc-rdr-cap}]
    {\subDft{\set{\rdroot}}{\set{\univ}}}

    \infrule[\ruledef{sc-reader}]
    {\isrdr{\G}{x}}
    {\subDft{\set{x}}{\set{\rdroot}}}
    \newrulefalse
  \end{multicols}

  where $\new{\isrdr{\G}{x}} \quad\triangleq\quad (x : T \in \G) \wedge (\sub{\G}{T}{\tRdr{S}\capt C})$ for some $T, C, S$

\begin{center}\rule{1.0\linewidth}{0.4pt}\end{center}

  \flushleft{\textbf{Subtyping \quad $\subDft{T}{U}$}}

  \begin{multicols}{3}

    \infax[\ruledef{refl}]
    {\subDft{T}{T}}

    \infrule[\ruledef{trans}]
    {\subDft{T_1}{T_2} \andalso \subDft{T_2}{T_3}}
    {\subDft{T_1}{T_3}}

    \infrule[\ruledef{tvar}]
    {X <: S \in \G}
    {\subDft{X}{S}}

    \infax[\ruledef{top}]
    {\subDft{T}{\top}}

    \infrule[\ruledef{fun}]
    {\subDft{U_2}{U_1} \\
      \sub{\G, {x}:_D {U_2}}{T_1}{T_2}}
    {\subDft{\forall(x :_{{D}} U_1) T_1}{\forall(x :_{{D}} U_2) T_2}}

    \infrule[\ruledef{tfun}]
    {\subDft{S_2}{S_1} \\
      \sub{\G, X <: S_2}{T_1}{T_2}}
    {\subDft{\tTForall{X}{S_1}{T_1}}{\tTForall{X}{S_2}{T_2}}}

    \infrule[\ruledef{capt}]
    {\subDft{C_1}{C_2} \\
      \subDft{S_1}{S_2}}
    {\subDft{\tCap{C_1}{S_1}}{\tCap{C_2}{S_2}}}

    \infrule[\ruledef{boxed}]
    {\subDft{T_1}{T_2}}
    {\subDft{\tBox{T_1}}{\tBox{T_2}}}

  \end{multicols}

\begin{center}\rule{1.0\linewidth}{0.4pt}\end{center}

\flushleft{\textbf{Typing \quad $\typDft{t}{T}$}}

\begin{multicols}{3}

  \infrule[\ruledef{var}]
  {x : S\capt C \in \G}
  {\typ{\G}{x}{S\capt \set{x}}}

  \infrule[\ruledef{sub}]
  {\typDft{t}{T} \\
  \subDft{T}{U}}
  {\typDft{t}{U}}

  \infrule[\ruledef{abs}]
  {\typ{\G, x:_D U}{t}{T} \\
    \wfTyp{\G}{U} \andalso
    \wfTyp{\G}{D}}
  {\typDft{\lambda(x:_D U). t}{\tCap{\cv{t}\!\setminus\!x}{\forall(x:_D U) T}}}

  \infrule[\ruledef{tabs}]
  {\typ{\G, {X} <: {S}}{t}{T} \andalso \wfTyp{\G}{S}}
  {\typDft{\tTLambda{X}{S}{t}}{\tCap{\cv{t}}{\tTForall{X}{S}{T}}}}

  \infrule[\ruledef{app}]
  {\typM{}{\G}{x}{\tCap{C}{\forall(z {:_D} T) U}} \\
    \typ{\G}{y}{T} \andalso \new{\ninter{\G}{y}{D}}}
  {\typDft{x\ y}{\fSubst{z}{y}{U}}}

  \infrule[\ruledef{tapp}]
  {\typDft{x}{\tCap{C}{\tTForall{X}{S}{T}}}}
  {\typDft{x[S]}{\fSubst{X}{S} {U}}}

  \infrule[\ruledef{box}]
  {\typDft{x}{\tCap C S} \\ C \subseteq \dom{\G}}
  {\typDft{\tBox{x}}{\tBox{\tCap C S}}}

  \infrule[\ruledef{unbox}]
  {\typDft{x}{\tBox{\tCap C S}} \\ C \subseteq \dom{\G}}
  {\typDft{\tUnbox{C}{x}}{\tCap{C}{S}}}

\infrule[\ruledef{let}]
{\typDft{s}{T}\\
  \typ{\G, {x :_{\set{}} T}}{t}{U}\andalso
  x \notin \fv{U} \\
  \new{\ninter{\G}{s}{t} \quad\text{if $m = \parmode$}}}
{\typDft{\tLetMode{m}{x}{s}{t}}{U}}

\end{multicols}



\begin{multicols}{4}

  \newruletrue
  \infrule[\ruledef{dvar}]
  {\typDft{y}{S} \\
    \typ{\G, x :_{D} \tCap{\set{\univ}}{\tRef{S}}}{t}{U}\\
    x \notin \fv{U} \andalso \wfTyp{\G}{D}}
  {\typDft{\tLetVarM{D}{x}{y}{t}}{U}}
  \newrulefalse

  \newruletrue
  \infrule[\ruledef{reader}]
  {\typDft{x}{\tRef{S}\capt C}}
  {\typDft{\reader x}{\tRdr{S}\captset{x}}}

  \infrule[\ruledef{read}]
  {\typDft{x}{\tCap{C}{\tRdr{S}}}}
  {\typDft{\tRead{x}}{S}}

  \infrule[\ruledef{write}]
  {\typDft{x}{\tCap{C}{\tRef{S}}}\\
    \typDft{y}{S}}
  {\typDft{\tWrite{x}{y}}{S}}
  \newrulefalse

\end{multicols}

\begin{center}\rule{1.0\linewidth}{0.4pt}\end{center}

\textbf{Reduction \quad $\reduction{\stDft{t}}{\st{\sta^\prime}{{t^\prime}}}$}

\begin{multicols}{4}

  \infrule[\ruledef{apply}]
  {\sta(\tVal{x}) = \lambda(z :_{D} T) t}
  {\reduction{\stDft{\evctx{x\,y}}}{\stDft{\evctx{[y/z]t}}}}

  \infrule[\ruledef{tapply}]
  {\sta(\tVal{x}) = \lambda[X <: S^\prime] t}
  {\reduction{\stDft{\evctx{x[S]}}}{\stDft{\evctx{[S/X]t}}}}

  \infrule[\ruledef{open}]
  {\sta(\tVal{x}) = \Box{} \, y}
  {\reduction{\stDft{\evctx{\tUnbox{C}{x}}}}{\stDft{\evctx{y}}}}

  \newruletrue
  \infrule[\ruledef{get}]
  {\sta(\tVal{x}) = \reader y\\ \sta(\tVar{y}) = v}
  {\reduction{\stDft{\evctx{\tRead{x}}}}{\stDft{\evctx{v}}}}
  \newrulefalse

\end{multicols}

\begin{multicols}{2}

  \infrule[\ruledef{lift-let}]
  {\fv{v} \subseteq \dom{\sta}}
  {\reduction{\stDft{\evctx{\tLetMode{m}{x}{v}{t}}}}{\st{\gamma, \tVal{x} \mapsto v}{\evctx{t}}}}

  \infax[\ruledef{rename}]
  {\reduction{\stDft{\evctx{\tLetMode{m}{x}{y}{t}}}}{\stDft{\evctx{[y/x]t}}}}

\end{multicols}

\begin{multicols}{2}

  \newruletrue
  \infrule[\ruledef{lift-var}]
  {\sta(\tVal{y}) = v}
  {\reduction{\stDft{\evctx{\tLetVarM{D}{x}{y}{t}}}}{\st{\sta, \tVar{x} = v}{\evctx{t}}}}

  \infrule[\ruledef{lift-set}]
  {\sta(\tVal{y}) = v}
  {\reduction{\stDft{\evctx{\tWrite{x}{y}}}}{\st{\sta,\tSet{x} = v}{\evctx{v}}}}
  \newrulefalse

\end{multicols}

\caption{Typing and Evaluation Rules of System \theCalculus{}}
  \label{fig:all-typing}

\end{figure*}

\subsection{Typing}




The majority of the rules are unchanged from \ccformal{}.


In addition to that in \ccformal{},
the \ruleref{app} rule now applies the separation checking
to ensure that the argument conforms to the required separation degree.
\ruleref{let} now accounts for parallel let-bindings,
the separation between the binding and the body terms is checked.
This is to prevent data races when the two terms are evaluated in parallel.
The \ruleref{dvar} rule introduces a mutable variable.
Since the variable is newly allocated and fresh,
we can declare an arbitrary separation degree for it.
\ruleref{reader} types a reader capability derived from a mutable variable $x$.
The \ruleref{read} and \ruleref{write} rules type the reading and writing of a mutable variable respectively.
The \ruleref{read} rule requires a reader while the \ruleref{write} rule requires a mutable variable reference,
which can be viewed as a full read-write capability for accessing the mutable state.

\subsection{Reduction}

The reduction judgement $\st{\sta}{t}\red \st{\sta'}{t'}$ reduces a evaluation configuration to another.
The evaluation configuration consists of a store context
and the term being reduced.
The reduction rules in \ccformal{} has a similar notion for store context 
but formulated in a different way \cite{cc1}.
The rules for the constructs in \ccformal{} remain unchanged.
The notations $\sta(\tVal{x})$ and $\sta(\tVar{x})$
are used to lookup up an immutable or mutable variable in the store context,
respectively.
The definition of $\sta(\tVal{x})$ follows the standard approach of
locating the corresponding binding and returning the associated value.
$\sta(\tVar{x})$ looks up the latest update to the mutable variable $x$,
i.e. the $\tSet{x} := v$ binding,
and if no update is found,
the initial value (in $\tVar{x} := v$) is returned.
The formal definition can be found in Definition \ref{def:lookup-mut}.

The evaluation context $e$
(defined in Figure \ref{fig:syntax})
focuses on a subterm that is to be reduced.
The parallelism of reduction semantics is reflected by the fact that
the evaluation context can focus on
either the binding term or the body term of a parallel let binding.
So that the evaluation of the two terms can be interleaved.
In Section \ref{sec:example_parallel} we illustrate an example of parallel evaluation.

The rules for mutable variables are introduced by \theCalculus{}.
The \ruleref{lift-var} rule creates a mutable variable binding in the store.
It first looks up the initial value $v$,
then lifts $\tVar{x} = v$ to the store.
The \ruleref{write} rule updates a mutable variable $x$ with a new value $v$
by lifting a $\tSet{x} = v$ binding.
The \ruleref{read} rule looks up the value of a mutable variable using the $\sta(\tVar{x})$ function.

\section{Metatheory}
\label{sec:metatheory}

In this section, we present the main formal results developed in the metatheory.
The complete proof
can be found in Appendix \ref{sec:proof}.

\subsection{Type Soundness}

We take the syntactic approach \cite{typesoundness} to establish type safety.
The following two theorems, namely progress and preservation, are proven.
\begin{restatable}[Preservation]{theorem}{preservation}
  \label{thm:preservation}
  If
  (i) $\new{\match{}{\sta}{\G}}$,
  (ii) $\typ{\G}{t}{T}$,
  and (iii) $\reduction{\stDft{t}}{\st{\sta^\prime}{t^\prime}}$,
  then $\exists \G^\prime$ such that
  (1) $\match{}{\sta^\prime}{\G^\prime}$
  and (2) $\typ{\G^\prime}{t^\prime}{T}$.
\end{restatable}
\begin{restatable}[Progress]{theorem}{progress}
  \label{thm:progress}
  If
  (i) $\new{\match{}{\sta}{\G}}$,
  (ii) $\typ{\G}{t}{T}$,
  then
  either $t$ is an answer $a$,
  or $\exists \sta^\prime, t^\prime$ such that
  $\reduction{\stDft{t}}{\st{\sta^\prime}{t^\prime}}$.
\end{restatable}

$\new{\match{}{\sta}{\G}}$ denotes the typing of store $\sta$:
the bindings in $\G$ consistently assigns types to the store bindings in $\sta$.
The lemmas necessary to prove the two theorems are mostly standard.
The major challenge arise when proving preservation.
When reducing a term $\evctx{s}$ to $\evctx{s'}$,
we must show that typing is preserved after replacing $s$ inside the evaluation context with $s'$.
To address this challenge, 
we introduce the concept of \emph{evaluation context inversion} to analyze evaluation contexts.
It reasons about what terms can be plugged into the evaluation context.
The details can be found in Section \ref{sec:evctx-inv}

\subsection{Data Race Freedom}

We demonstrate that \theCalculus{} is capable of statically eliminating data races
by proving the confluence of reductions,
as stated in the following theorems.
\begin{restatable}[Confluence]{theorem}{confluence}
  \label{thm:confluence}
  Given two equivalent configurations
  $\new{\st{\sta_1}{t} \sequiv \st{\sta_2}{t}}$,
  if
  (1) $\sta_1\vdash t$ and $\sta_2 \vdash t$;
  (2) $\st{\sta_1}{t}\new{\redt} \st{\sta'_1}{t_1}$;
  and (3) $\st{\sta_2}{t}\redt \st{\sta'_2}{t_2}$,
  then
  there exists $\sta''_1, t', \sta''_2$ such that
  (1) $\st{\sta'_1}{t_1}\redt \st{\sta''_1}{t'}$,
  (2) $\st{\sta'_2}{t_2}\redt \st{\sta''_2}{t'}$,
  and (3) $\st{\sta''_1}{t'} \sequiv \st{\sta''_2}{t'}$.
\end{restatable}
\begin{restatable}[Uniqueness of Answer]{theorem}{uniquenessofanswer}
	\label{thm:uniqueness-of-answer}
	For any $t$,
  if (i) $\sta\vdash t$,
  (ii) $\st{\sta}{t}\redt \st{\sta_1}{a_1}$
  and (iii) $\st{\sta}{t}\redt \st{\sta_2}{a_2}$
  then $a_1 = a_2$ and $\sta_1 \sequiv \sta_2$.
\end{restatable}
Theorem \ref{thm:uniqueness-of-answer} is a straightforward corollary of Theorem \ref{thm:confluence}.
Here, $\st{\sta}{t}\new{\redt} \st{\sta'}{t'}$ denotes the transitive and reflexive closure of the small step reduction $\st{\sta}{t}\red \st{\sta'}{t'}$.
Two configurations are considered equivalent ($\new{\st{\sta_1}{t} \sequiv \st{\sta_2}{t}}$)
if the two stores are equivalent (written $\new{\sta_1\sequiv \sta_2}$)
and the terms are equal.

The concept of store equivalence arises from the fact that
bindings can be lifted to the store in different orders
during difference reduction paths.
$\sta_1\sequiv\sta_2$ indicates that $\sta_1$ and $\sta_2$ are \emph{externally indistinguishable}
despite the fact that the bindings may be permuted.
\begin{definition}[Equivalent stores]
  We say two stores $\sta_1$
  and $\sta_2$ are equivalent,
  written $\sta_1 \sequiv \sta_2$,
  iff
  \begin{enumerate}[(1)]
  \item $\sta_2$ is permuted from $\sta_1$;
  \item $\forall x \in \bvar{\sta_1}$,
    $\sta_1(\tVar{x}) = \sta_2(\tVar{x})$.
  \end{enumerate}
\end{definition}


We take the standard approach to prove confluence by showing the diamond property \cite{churchrosser}.
The full proof can be found in Section \ref{sec:confluence-proof}.

\section{Encoding Compound Data Types with Separation}
\label{sec:discussion}


Given that there is no compound data types in the formal system,
how to represent them in the calculus?
It is already demonstrated that compound data types, like linked lists, can be church-encoded in \ccformal{} \cite{cc1}.
Such encodings applies to \theCalculus{} as well.
More interestingly,
in \theCalculus{},
we can extend the church-encoding to retain the separation information in data structures.

The following shows the Böhm-Berarducci representation of a linked list containing \emph{mutually-separated} mutable references:
\begin{minted}{scala}
type Op[E, R] = (x: Ref[E]^) ->{} R ->{x} R
type RefList[X] =
  [R] -> (op: Op[X, R]) ->{cap} (x0: R) ->{cap} R

def nil[X] =
  [R] => (op: Op[X, R]) => (x0: R) => x0
def cons[X](x: Ref[X]^, sep{x} xs: RefList[X]^): 
  RefList[X]^{x, xs} =
  [R] => op => x0 =>
    letpar acc = xs[R](op, x0) in
      op(x, acc)
\end{minted}
The \mintinline{scala}{type} definitions can be thought of as type synonyms.
When applied, they expand to their right-hand-sides with the variables being substituted by the provided arguments.
A \mintinline{scala}{RefList[X]} carries the separation information that
the references stored in it are mutually-disjoint:
there is no aliasing between any two of them.
When constructing a list,
the separation between the new element and the rest of the list
is ensured by the \mintinline{scala}{sep{x}} annotation on the second argument of \mintinline{scala}{cons}.
When eliminating a \mintinline{scala}{RefList[X]},
the separation between the head and the tail is utilized
so that the head and the tail can be processed in parallel
with static data race freedom guarantees.

The following example makes use of the \mintinline{scala}{RefList[X]} data structure
and increments a list of mutable integer references simontaneously.
\begin{minted}{scala}
def incrementAll(xs: RefList[Int]^): Unit =
  xs[Unit](x => acc => x.update(_ + 1), ())

val x = new Ref(0)
val y = new Ref(0)
val z = new Ref(0)
val xs: RefList[Int]^{x, y, z} = 
  cons(x, cons(y, cons(z, nil)))
incrementAll(xs)
\end{minted}
Thanks to the separation information encoded in the list,
it is assured that the parallel updates on the last line are free from data races,
a guarantee upheld by the type system.

It is worth noting that
similar guarentees can be achieved in other systems, 
such as ownership types \cite{ownership1,ownership2} and unique types \cite{refcap1}.
Compared to the existing approaches,
\theCalculus{} is more flexible.
First,
ownership systems impose restrictions on the structure of data by default,
e.g. enforcing a tree-like structure on the heap.
By contrast, \theCalculus{} only enforces the separation between components 
when constructing separation-aware data structures like \mintinline{scala}{RefList[X]}.
Arbitrary aliasing is allowed by default in ordinary compound data types.
For instance, the church-encoded list in \ccformal{} is still valid in \theCalculus{},
allowing arbitrary aliasing between the elements.
Second,
putting a reference into a data structure like \mintinline{scala}{RefList[X]}
does not forbid the usage of the reference itself.
Indeed, it is safe to use the reference after putting it into the list.
What should be prevented is to use the reference \emph{in parallel with} the list.
The following code, which is modified from the previous code snippet, exemplifies the idea:
\begin{minted}{scala}
...
val xs: RefList[Int]^{x, y, z} = 
  cons(x, cons(y, cons(z, nil)))
x.update(_ + 1) // OK
incrementAll(xs) // OK
x.update(_ + 1) || incrementAll(xs) // error
\end{minted}
In unique types,
inserting a reference into a container with unique elements immediately renders the reference unusable.
We acknowledge other systems with stricter limitations
can offer stronger safety guarantees and facilitate advanced compiler optimizations.
Yet, the flexibility of \theCalculus{} is key to ensuring that
the system remains non-intrusive to existing code bases
and retains compatibility with \ccformal{}.

\note{
Note that the encoding only works for \emph{monomorphic} data structures.
To see why,
consider the following examples,
which tries to church-encode a pair of disjoint elements:
\begin{minted}{scala}
type SPair[A, B] = 
  [R] -> (op: (a: A) => (sep{a} b: B) => R) -> R
def spair[A, B](a: A, sep{a} b: B): SPair[A, B] = 
  [R] => op => op(a)(b)
\end{minted}
At first glance, the encoding seems to work.
But in fact, the \mintinline{scala}{sep{a}} annotation,
which ostensibly asserts that the two components are separate,
falls short of doing its job.
The reason is simple but crucial:
type variables and their instances are always pure.
Therefore, any requirements of separation between \mintinline{scala}{a} and \mintinline{scala}{b} are trivially met 
because \mintinline{scala}{a}'s capture set is empty to begin with.
}

\note{\color{red}
On the surface, \theCalculus{} seems equipped to represent data structures with separation information that is \emph{polymorphic}.
Take the following example: 
it employs a straightforward Church encoding to create a pair of disjoint elements.
\begin{minted}{scala}
type SPair[A, B] = [R] -> (op: (a: A) => (sep{a} b: B) => R) -> R
def spair[A, B](a: A, sep{a} b: B): SPair[A, B] = [R] => op => op(a)(b)
\end{minted}
At first glance, the encoding seems to work.
But in fact, the \mintinline{scala}{sep{a}} annotation,
which ostensibly asserts that the two components are separate,
falls short of doing its job.
The reason is simple but crucial:
type variables always instantiate as pure types. 
Therefore, any requirements of separation between \mintinline{scala}{a} and \mintinline{scala}{b} are trivially met 
because \mintinline{scala}{a}'s capture set is empty to begin with.
Actually, it is entirely possible to create a pair of aliasing references:
\begin{minted}[escapeinside=||]{scala}
val x = new Ref(0)
val p: SPair[|$\Box$| Ref[Int]^{x}, |$\Box$| Ref[Int]^{x}] = 
  spair(|$\Box$| x, |$\Box$| x)
\end{minted}
Remarkably, this does not invalidate the soundness of the system.
Before a boxed reference can be used, it must be unboxed, 
and the unboxing step exposes any overlap between the references. 
The following code, which is a continuation of the previous example, illustrates the point.
It has a type error on the second line
because both parties reference \mintinline{scala}{x} and the separation check will not pass:
\begin{minted}[escapeinside=@@]{scala}
def f(a: @$\Box$@ Ref[Int]^{x}, sep{a} b: @$\Box$@ Ref[Int]^{x}) =
  ({x} @$\multimapinv$@ a).update(_ + 1) || ({x} @$\multimapinv$@ b).update(_ + 1)
p[Unit](f)
\end{minted}
}

\note{\color{orange}

Indeed, we can extend \theCalculus{} to directly handle complex data types like pairs and product types. 
This lifts the limitations of using type variables in church-encoded pairs by allowing for impure elements within a pair. 
Furthermore, we can introduce a special product type where both elements are separate from each other.

However, adding pairs to the system complicates the notion of capture sets and separation degrees.
Instead of just variables, 
we now have to consider \emph{paths} to accurately model how different parts of a pair are captured.
The exploration of such a system is beyond the scope of this work.
}

\section{Related Work}
\label{sec:related-work}

The underlying ideas of \theCalculus{} have individually appeared in the literature in different forms.
We now examine the related work sharing the ideas of \theCalculus{}
and discuss how our work differs.

\subsection{Flexibility}

Considerable research attention has been given to
achieving a harmonious balance between flexibility and safety
in type systems for safe concurrency or, more generally, alias control.

One of the earliest work is syntactic control of interference \cite{reynolds1978syntactic}.
Its design principle largely inspires \theCalculus{}:
aliases are possible but are syntactically detectable. 
Its proposed approach for syntactically detecting interference
is similar to separation checking in \theCalculus{}.
The notion of \emph{passive} expressions
characterizes expressions that do not write to global variables.
Two passive expressions therefore never interfere.
Being passive is similar to subcapturing \rdroot{} in \theCalculus{}.
Notably, their interference detection is stricter than separation checking:
when calling a function, the argument must be non-interfering with the function.
This constraint is not enforced in \theCalculus{}
but is expressible as requiring the argument being separated from the function's capture set, i.e.,
\mintinline[breaklines]{scala}{(sep{x1,...,xn} z: T) ->{x1,...,xn} U}.

Recently, \citeauthor{milano2022flexible} proposes a flexible type system for fearless concurrency.
It groups objects into regions.
Intra-region references freely link objects within the same region;
inter-region references are controlled.
Fearless concurrency is achieved
by ensuring that the \emph{reservation} of each thread,
the regions the thread accesses,
is disjoint from others'.
To control inter-region links,
it starts by enforcing a global heap invariant,
requiring the regions pointed to being isolated:
an inter-region reference should be the unique pointer to the reachable subgraph of the target region.
This is a restrictive invariant;
to relax it, they propose \emph{tempered domination},
a mechanism similar to \emph{focus} \cite{fhndrich2002focus},
which exempts the targets of an isolated reference
from the domination invariant
at the cost of tracking aliases to them in the type system.
\citeauthor{milano2022flexible}
ensure the global domination invariant by default,
and permits and tracks aliasing as a special case.
In comparison, \theCalculus{} tracks and allows aliasing by default,
and enforces separation conditions when needed.
\citeauthor{milano2022flexible} formalize message-passing concurrent programs,
where threads communicate via send/recv primitives,
and the reservations are always disjoint;
whereas \theCalculus{} formalizes data-sharing concurrency,
where mutable states are shared directly,
and a state can be shared by multiple threads if they only read the state.

\subsection{Immutability}

\theCalculus{} allows using several reader capabilities in multiple parallel procedures
to immutably share a mutable state.
Similar functionalities are also present in existing systems.

Fractional permissions \cite{boyland2003fractional,boyland2010fractional}
are linear keys guarding accesses to mutable states.
They are splittable:
a permission can be split into arbitrary fractions to perform shared read accesses.
However, a write access requires the full permission,
so fractions of a permission have to be merged before mutating the state.

Other systems offer similar functionalities.
Rust \cite{rustbook} allows creating multiple immutable borrows to a variable;
all immutable borrows must be dropped before the variable can be moved or mutated.
Capability Calculus \cite{crary1999capabilities} uses bounded quantification
to enable the temporary sharing of a capability and recover the full capability afterwards.
With reference capabilities \cite{refcap1},
one can use an isolated reference at mutable or immutable mode within a scope,
then recover the isolation.

In most existing systems,
the original variable (or permission, capability) is rendered unusable
once it is shared,
and becomes usable again when the shared references are relinquished.
In comparison,
\theCalculus{} allows the original variable to co-exist and be used together with its shared references
as long as the usages are not parallelized.

\subsection{Systems with Global Invariants}

A significant number of existing work incorporate certain global anti-aliasing invariants.

Linearity \cite{wadler1990linear,fhndrich2002focus,walker2001linear,rustbook} enforces each variable to be used exactly once.
This effectively disables aliasing, prevents mutable states from being shared,
and thus eliminates data races.
Reference capabilities \cite{refcap1,pony1} tag the references with \emph{modes} that indicate the aliasing status of the referenced memory.
Only isolated references can be mutated in parallel procedures,
and the isolated mode indicates that the reference to the state is globally unique.
Therefore, mutating an isolated reference will not cause data races.
\citeauthor{ownership1} employ ownership types \cite{ownership2} to control object aliasing and to statically enforce the correct synchronization mechanisms.
For instance, thread-local objects, which are objects that are \emph{owned} by the thread and thus cannot be aliased by other threads, can be accessed without the protection of locks.
Rust \cite{rustbook,weiss2019oxide,pearce2021rust} achieves data race freedom via the anti-aliasing principle enforced by its ownership system.
Specifically, when spawning a thread in Rust, the data captured by the thread must be owned by the function,
which disables the aliasing of the mutable states.

\section{Conclusion}
\label{sec:conclusion}

In this paper, we have presented \theCalculus{}, 
a calculus for modeling parallelism with shared mutable states while statically preventing data races. 
\theCalculus{} follows the paradigm of \emph{control-as-you-need} 
by allowing aliases to mutable states in general and regulating them when necessary to prevent data races. 
In the metatheory,
we demonstrate the type soundness of \theCalculus{} by showing the progress and preservation theorems.
Furthermore, we formally prove the data race freedom of \theCalculus{} by showing the confluence of reductions.
A prototype of \theCalculus{} is implemented as an extension to the Scala 3 compiler.


\bibliographystyle{ACM-Reference-Format}
\bibliography{references}


\begin{thebibliography}{20}


\ifx \showCODEN    \undefined \def \showCODEN     #1{\unskip}     \fi
\ifx \showDOI      \undefined \def \showDOI       #1{#1}\fi
\ifx \showISBNx    \undefined \def \showISBNx     #1{\unskip}     \fi
\ifx \showISBNxiii \undefined \def \showISBNxiii  #1{\unskip}     \fi
\ifx \showISSN     \undefined \def \showISSN      #1{\unskip}     \fi
\ifx \showLCCN     \undefined \def \showLCCN      #1{\unskip}     \fi
\ifx \shownote     \undefined \def \shownote      #1{#1}          \fi
\ifx \showarticletitle \undefined \def \showarticletitle #1{#1}   \fi
\ifx \showURL      \undefined \def \showURL       {\relax}        \fi
\providecommand\bibfield[2]{#2}
\providecommand\bibinfo[2]{#2}
\providecommand\natexlab[1]{#1}
\providecommand\showeprint[2][]{arXiv:#2}

\bibitem[Almeida(1997)]%
        {almeida1997balloon}
\bibfield{author}{\bibinfo{person}{Paulo~S{\'e}rgio Almeida}.}
  \bibinfo{year}{1997}\natexlab{}.
\newblock \showarticletitle{Balloon Types: Controlling Sharing of State in Data
  Types}. In \bibinfo{booktitle}{\emph{European Conference on Object-Oriented
  Programming}}.
\newblock


\bibitem[Boyapati and Rinard(2001)]%
        {ownership1}
\bibfield{author}{\bibinfo{person}{Chandrasekhar Boyapati} {and}
  \bibinfo{person}{Martin~C. Rinard}.} \bibinfo{year}{2001}\natexlab{}.
\newblock \showarticletitle{A Parameterized Type System for Race-Free Java
  Programs}. In \bibinfo{booktitle}{\emph{OOPSLA}}. \bibinfo{pages}{56--69}.
\newblock


\bibitem[Boyland(2003)]%
        {boyland2003fractional}
\bibfield{author}{\bibinfo{person}{John~Tang Boyland}.}
  \bibinfo{year}{2003}\natexlab{}.
\newblock \showarticletitle{Checking Interference with Fractional Permissions}.
  In \bibinfo{booktitle}{\emph{Sensors Applications Symposium}}.
\newblock


\bibitem[Boyland(2010)]%
        {boyland2010fractional}
\bibfield{author}{\bibinfo{person}{John~Tang Boyland}.}
  \bibinfo{year}{2010}\natexlab{}.
\newblock \showarticletitle{Semantics of fractional permissions with nesting}.
\newblock \bibinfo{journal}{\emph{ACM Trans. Program. Lang. Syst.}}
  \bibinfo{volume}{32} (\bibinfo{year}{2010}), \bibinfo{pages}{22:1--22:33}.
\newblock


\bibitem[Church and Rosser(1936)]%
        {churchrosser}
\bibfield{author}{\bibinfo{person}{Alonzo Church} {and}
  \bibinfo{person}{J.~Barkley Rosser}.} \bibinfo{year}{1936}\natexlab{}.
\newblock \showarticletitle{Some properties of conversion}.
\newblock \bibinfo{journal}{\emph{Trans. Amer. Math. Soc.}}
  \bibinfo{volume}{39} (\bibinfo{year}{1936}), \bibinfo{pages}{472--482}.
\newblock


\bibitem[Clarke et~al\mbox{.}(1998)]%
        {ownership2}
\bibfield{author}{\bibinfo{person}{David~G. Clarke}, \bibinfo{person}{John
  Potter}, {and} \bibinfo{person}{James Noble}.}
  \bibinfo{year}{1998}\natexlab{}.
\newblock \showarticletitle{Ownership Types for Flexible Alias Protection}. In
  \bibinfo{booktitle}{\emph{OOPSLA}}. \bibinfo{pages}{48--64}.
\newblock


\bibitem[Clebsch et~al\mbox{.}(2015)]%
        {pony1}
\bibfield{author}{\bibinfo{person}{Sylvan Clebsch}, \bibinfo{person}{Sophia
  Drossopoulou}, \bibinfo{person}{Sebastian Blessing}, {and}
  \bibinfo{person}{Andy McNeil}.} \bibinfo{year}{2015}\natexlab{}.
\newblock \showarticletitle{Deny capabilities for safe, fast actors}.
\newblock \bibinfo{journal}{\emph{Proceedings of the 5th International Workshop
  on Programming Based on Actors, Agents, and Decentralized Control}}
  (\bibinfo{year}{2015}).
\newblock


\bibitem[Crary et~al\mbox{.}(1999)]%
        {crary1999capabilities}
\bibfield{author}{\bibinfo{person}{Karl Crary}, \bibinfo{person}{David Walker},
  {and} \bibinfo{person}{J.~Gregory Morrisett}.}
  \bibinfo{year}{1999}\natexlab{}.
\newblock \showarticletitle{Typed memory management in a calculus of
  capabilities}. In \bibinfo{booktitle}{\emph{ACM-SIGACT Symposium on
  Principles of Programming Languages}}.
\newblock


\bibitem[F{\"a}hndrich and DeLine(2002)]%
        {fhndrich2002focus}
\bibfield{author}{\bibinfo{person}{Manuel F{\"a}hndrich} {and}
  \bibinfo{person}{Robert DeLine}.} \bibinfo{year}{2002}\natexlab{}.
\newblock \showarticletitle{Adoption and focus: practical linear types for
  imperative programming}. In \bibinfo{booktitle}{\emph{ACM-SIGPLAN Symposium
  on Programming Language Design and Implementation}}.
\newblock


\bibitem[Gordon et~al\mbox{.}(2012)]%
        {refcap1}
\bibfield{author}{\bibinfo{person}{Colin~S. Gordon},
  \bibinfo{person}{Matthew~J. Parkinson}, \bibinfo{person}{Jared Parsons},
  \bibinfo{person}{Aleks Bromfield}, {and} \bibinfo{person}{Joe Duffy}.}
  \bibinfo{year}{2012}\natexlab{}.
\newblock \showarticletitle{Uniqueness and Reference Immutability for Safe
  Parallelism}. In \bibinfo{booktitle}{\emph{OOPSLA}}. \bibinfo{pages}{21--40}.
\newblock


\bibitem[Klabnik and Nichols(2018)]%
        {rustbook}
\bibfield{author}{\bibinfo{person}{Steve Klabnik} {and} \bibinfo{person}{Carol
  Nichols}.} \bibinfo{year}{2018}\natexlab{}.
\newblock \bibinfo{booktitle}{\emph{The Rust Programming Language}}.
\newblock \bibinfo{publisher}{No Starch Press}, \bibinfo{address}{USA}.
\newblock
\showISBNx{1593278284}


\bibitem[Milano et~al\mbox{.}(2022)]%
        {milano2022flexible}
\bibfield{author}{\bibinfo{person}{Mae Milano}, \bibinfo{person}{Joshua
  Turcotti}, {and} \bibinfo{person}{Andrew~C. Myers}.}
  \bibinfo{year}{2022}\natexlab{}.
\newblock \showarticletitle{A flexible type system for fearless concurrency}.
\newblock \bibinfo{journal}{\emph{Proceedings of the 43rd ACM SIGPLAN
  International Conference on Programming Language Design and Implementation}}
  (\bibinfo{year}{2022}).
\newblock


\bibitem[Odersky et~al\mbox{.}(2022)]%
        {cc1}
\bibfield{author}{\bibinfo{person}{Martin Odersky}, \bibinfo{person}{Aleksander
  Boruch-Gruszecki}, \bibinfo{person}{Edward Lee},
  \bibinfo{person}{Jonathan~Immanuel Brachth{\"a}user}, {and}
  \bibinfo{person}{Ondrej Lhot{\'a}k}.} \bibinfo{year}{2022}\natexlab{}.
\newblock \showarticletitle{Scoped Capabilities for Polymorphic Effects}.
\newblock \bibinfo{journal}{\emph{ArXiv}}  \bibinfo{volume}{abs/2207.03402}
  (\bibinfo{year}{2022}).
\newblock


\bibitem[Pearce(2021)]%
        {pearce2021rust}
\bibfield{author}{\bibinfo{person}{David~J. Pearce}.}
  \bibinfo{year}{2021}\natexlab{}.
\newblock \showarticletitle{A Lightweight Formalism for Reference Lifetimes and
  Borrowing in Rust}.
\newblock \bibinfo{journal}{\emph{ACM Transactions on Programming Languages and
  Systems (TOPLAS)}}  \bibinfo{volume}{43} (\bibinfo{year}{2021}),
  \bibinfo{pages}{1 -- 73}.
\newblock


\bibitem[Reynolds(1978)]%
        {reynolds1978syntactic}
\bibfield{author}{\bibinfo{person}{John~C. Reynolds}.}
  \bibinfo{year}{1978}\natexlab{}.
\newblock \showarticletitle{Syntactic control of interference}.
\newblock \bibinfo{journal}{\emph{Proceedings of the 5th ACM SIGACT-SIGPLAN
  symposium on Principles of programming languages}} (\bibinfo{year}{1978}).
\newblock


\bibitem[Servetto et~al\mbox{.}(2013)]%
        {servetto2013balloon}
\bibfield{author}{\bibinfo{person}{Marco Servetto}, \bibinfo{person}{David~J.
  Pearce}, \bibinfo{person}{Lindsay~J. Groves}, {and} \bibinfo{person}{Alex
  Potanin}.} \bibinfo{year}{2013}\natexlab{}.
\newblock \showarticletitle{Balloon Types for Safe Parallelisation over
  Arbitrary Object Graphs}.
\newblock


\bibitem[Wadler(1990)]%
        {wadler1990linear}
\bibfield{author}{\bibinfo{person}{Philip Wadler}.}
  \bibinfo{year}{1990}\natexlab{}.
\newblock \showarticletitle{Linear Types can Change the World!}. In
  \bibinfo{booktitle}{\emph{Programming Concepts and Methods}}.
\newblock


\bibitem[Walker and Watkins(2001)]%
        {walker2001linear}
\bibfield{author}{\bibinfo{person}{David Walker} {and} \bibinfo{person}{Kevin
  Watkins}.} \bibinfo{year}{2001}\natexlab{}.
\newblock \showarticletitle{On Regions and Linear Types}. In
  \bibinfo{booktitle}{\emph{ACM SIGPLAN International Conference on Functional
  Programming}}.
\newblock


\bibitem[Weiss et~al\mbox{.}(2019)]%
        {weiss2019oxide}
\bibfield{author}{\bibinfo{person}{Aaron Weiss}, \bibinfo{person}{Daniel
  Patterson}, \bibinfo{person}{Nicholas~D. Matsakis}, {and}
  \bibinfo{person}{Amal~J. Ahmed}.} \bibinfo{year}{2019}\natexlab{}.
\newblock \showarticletitle{Oxide: The Essence of Rust}.
\newblock \bibinfo{journal}{\emph{ArXiv}}  \bibinfo{volume}{abs/1903.00982}
  (\bibinfo{year}{2019}).
\newblock


\bibitem[Wright and Felleisen(1994)]%
        {typesoundness}
\bibfield{author}{\bibinfo{person}{Andrew~K. Wright} {and}
  \bibinfo{person}{Matthias Felleisen}.} \bibinfo{year}{1994}\natexlab{}.
\newblock \showarticletitle{A Syntactic Approach to Type Soundness}.
\newblock \bibinfo{journal}{\emph{Inf. Comput.}}  \bibinfo{volume}{115}
  (\bibinfo{year}{1994}), \bibinfo{pages}{38--94}.
\newblock


\end{thebibliography}

\appendix







\section{Additional Details}

\subsection{Definition of \textsf{cv}}
\label{sec:cv}

\begin{align*}
\small
&\cv{\lambda(x\colon T)t} \quad &= \quad &\cv{t} \setminus {x}, & \\
&\cv{\lambda[X<:S]t} \quad &= \quad &\cv{t}, & \\
&\cv{x} \quad &= \quad &\set{x}, & \\
&\cv{\tLetMode{m}{x}{v}{u}}  \quad &= \quad &\cv{u}, &\text{if $x \notin \cv{u}$} \\
&\cv{\tLetMode{m}{x}{t}{u}}  \quad &= \quad &\cv{t} \cup \cv{u} \setminus {x}, & \\
&\cv{x\,y} \quad &= \quad &\set{x, y}, & \\
&\cv{x[S]} \quad &= \quad &\set{x}, & \\
&\cv{\Box\ x} \quad &= \quad &\set{}, & \\
&\cv{\tUnbox{C}{x}} \quad &= \quad &C \cup \set{x}, & \\
&\new{\cv{\tLetVarM{D}{x}{y}{u}}} \quad &= \quad &\set{y} \cup \cv{u} \setminus {x}, & \\
&\new{\cv{\tWrite{x}{y}}} \quad &= \quad &\set{x, y}, & \\
&\new{\cv{\reader{x}}} \quad &= \quad &\set{x}, & \\
&\new{\cv{\tRead{x}}} \quad &= \quad &\set{x}. & \\
\end{align*}

\subsection{Example of Parallel Evaluation}
\label{sec:example_parallel}

Now we illustrate a concrete example for the reduction of parallel let bindings.
Let the following be the initial evaluation configuration:
\begin{align*}
  \sta \ | \ &\textsf{let}_{\parmode}\ z_1 =  \\
  &\quad \textsf{let}_{\epsilon}\ z_2 = \tRead{x}\ \textsf{in} \\
  &\quad\quad \lambda(z: \textsf{Nat}).\ z_2 + z \\
  &\textsf{in}\ \textsf{let}_{\epsilon}\ z_3 = \tRead{y}\ \textsf{in}\ z_1\ z_3
\end{align*}
It reads two mutable variables $x$ and $y$ (or the reader capabilities or the two mutable states)
and adds them together.
Note that we assume a standard natural number semantics as part of the language,
which is just for the sake of the example and not essential to the parallel reduction.
Additionally, we assume the initial store to be $\sta$
with the value resulted from reading $x$ and $y$ being $1$ and $2$ respectively.
We reduce the reading operation from $x$ in the first step (which is in the binding term of the parallel let binding),
resulting in the following term:
\begin{align*}
  \sta \ | \ &\textsf{let}_{\parmode}\ z_1 =  \\
  &\quad \textsf{let}_{\epsilon}\ z_2 = \new{1}\ \textsf{in} \\
  &\quad\quad \lambda(z: \textsf{Nat}).\ z_2 + z \\
  &\textsf{in}\ \textsf{let}_{\epsilon}\ z_3 = \tRead{y}\ \textsf{in}\ z_1\ z_3
\end{align*}
The second step focuses on the body term of the parallel let binding,
reducing it into:
\begin{align*}
  \sta \ | \ &\textsf{let}_{\parmode}\ z_1 =  \\
  &\quad \textsf{let}_{\epsilon}\ z_2 = 1\ \textsf{in} \\
  &\quad\quad \lambda(z: \textsf{Nat}).\ z_2 + z \\
  &\textsf{in}\ \textsf{let}_{\epsilon}\ z_3 = \new{2}\ \textsf{in}\ z_1\ z_3
\end{align*}
This illustrates the \emph{parallelism} of \theCalculus{}'s reduction, 
the binding term and the body term of the parallel binding are reduced interleavingly.
We continue to evaluate the body term, reducing it into:
\begin{align*}
  (\sta, \new{\tVal{z_3}\mapsto 2}) \ | \ &\textsf{let}_{\parmode}\ z_1 =  \\
  &\quad \textsf{let}_{\epsilon}\ z_2 = 1\ \textsf{in} \\
  &\quad\quad \lambda(z: \textsf{Nat}).\ z_2 + z \\
  &\textsf{in}\ \new{z_1\ z_3}
\end{align*}
This step uses the \ruleref{lift-let} rule to lift the value of $z_3$ to the store context.
At this point, we cannot evaluate the body term any further,
since it requires the value of $z_1$ which is not available yet.
The reduction of body term is blocked until the binding term is fully reduced,
which is similar to awaiting a future.
We continue to reduce the binding term, finishing the reduction of the binding:
\begin{align*}
  (\sta, {\tVal{z_3}\mapsto 2}, \new{\tVal{z_2}\mapsto 1}, \new{\tVal{z_1}\mapsto \lambda(z: \textsf{Nat}).\ z_2 + z}) \ | \ \new{z_1\, z_3}
\end{align*}
At this point,
the binding is fully evaluated and lifted to the store,
which unblocks the reduction of the body term.
This is similar to the resolution of a future.
Finally, we use the \ruleref{apply} rule and reduce the term into:
\begin{align*}
  (\sta, {\tVal{z_3}\mapsto 2}, {\tVal{z_2}\mapsto 1}, {\tVal{z_1}\mapsto \lambda(z: \textsf{Nat}).\ z_2 + z}) \ | \ \new{3}
\end{align*}

\subsection{Definition of Mutable Variable Lookup}

\begin{definition}[Mutable variable lookup]
  \label{def:lookup-mut}
  $\sta(\tVar{x})$ denotes the result of looking up the mutable variable $x$ in the store $\sta$.
  It is defined as follows:
  \begin{align*}
  &\sta(\tVar{x}) \quad &= \quad &v, \quad&\text{if $\sta = \sta^\prime, \tSet{x} = v$}, \\
  &\sta(\tVar{x}) \quad &= \quad &v, \quad&\text{if $\sta = \sta^\prime, \tVar{x} = v$}, \\
  &\sta(\tVar{x}) \quad &= \quad &\sta^\prime(\tVar{x}), \quad&\text{if $\sta = \sta^\prime, \tVar{y} = v$}, \\
  &\sta(\tVar{x}) \quad &= \quad &\sta^\prime(\tVar{x}), \quad&\text{if $\sta = \sta^\prime, \tSet{y} = v$}, \\
  &\sta(\tVar{x}) \quad &= \quad &\sta^\prime(\tVar{x}), \quad&\text{if $\sta = \sta^\prime, \tVal{y} \mapsto v$}. \\
  \end{align*}
  In the last three lines, $x$ and $y$ are distinct variables, i.e. $x \neq y$.
\end{definition}

\section{Proofs}
\label{sec:proof}

In the proofs, we follow the Barendregt convention where all the bound variables are distinct.

\subsection{Proof Devices}

We first introduce the supporting proof devices.
We begin by presenting the typing of stores and evaluation contexts.
Figure \ref{fig:matching-env} defines these judgments.

\begin{wide-rules}
  \textbf{Store Typing \quad $\match{}{\sta}{\G}$}

  \begin{multicols}{2}

    \infax[\ruledef{st-empty}]
    {\match{}{\cdot}{\emptyset}}

    \infrule[\ruledef{st-val}]
    {\match{}{\sta}{\G} \andalso
      \typ{\G}{v}{\tCap{\cv{v}}{S}}}
    {\match{}{(\sta, \tVal{x} \mapsto v)}{\G, x :_{\set{}} \tCap{\cv{v}}{S}}}

    \infrule[\ruledef{st-var}]
    {\match{}{\sta}{\G} \andalso \typ{\G}{v}{S}}
    {\match{}{(\sta, \tVar{x} := v)}{\G, x :_{\dom{\G}} \tCap{\set{\univ}}{\tRef{S}}}}

    \infrule[\ruledef{st-set}]
    {\match{}{\sta}{\G} \andalso x :_{D} \tCap{C}{\tRef{S}} \in \G \\
      \typ{\G}{v}{S}}
    {\match{}{(\sta, \tSet{x} := v)}{\G}}

  \end{multicols}

  \textbf{Evaluation Context Typing \quad $\match{\G}{e}{\Delta}$}

    \infax[\ruledef{ev-empty}]
    {\match{\G}{[]}{\emptyset}}

  \begin{multicols}{2}

    \infrule[\ruledef{ev-let-1}]
    {\match{\G}{e}{\Delta}}
    {\match{\G}{\tLetMode{m}{x}{e}{s}}{\Delta}}

    \infrule[\ruledef{ev-let-2}]
    {\typ{\G}{s}{T}\andalso 
      \match{\G, x :_{\set{}} T}{e}{\Delta}}
    {\match{\G}{\tLetMode{\parmode}{x}{s}{e}}{x :_{\set{}} T, \Delta}}

  \end{multicols}

  \caption{Store and Evaluation Context Typing}
  \label{fig:matching-env}

\end{wide-rules}

\subsubsection{Store Typing}
$\match{}{\sta}{\G}$ states that the store $\sta$ can be typed as a typing context $\G$.
\ruleref{st-val} and \ruleref{st-var} types the immutable and mutable bindings in the store.
\ruleref{st-val} types the value $v$ as capturing a precise capture set $\cv{v}$,
and \ruleref{st-var} introduces the mutable variable binding with the separation degree spanning over the entire context (i.e. being $\dom{\G}$).
These treatments type the store as a precise and strong typing context, which eases the proof.
The \ruleref{st-set} rule types an update to the variable $x$
by verifying that the mutable variable $x$ is defined in the context
and the type of the new value $v$ matches the type of $x$.

\subsubsection{Evaluation Context Typing}
\theCalculus{} additionally introduces the typing of evaluation contexts.
$\match{\G}{e}{\Delta}$ states that the evaluation context $e$ can be typed as $\Delta$ under an existing typing context $\G$,
where in the metatheory the $\G$ is always obtained from the typing of a store $\sta$ ($\match{}{\sta}{\G}$).
The need for this judgment arises from the parallel semantics of the calculus.
When reasoning about a term $\evctx{s}$ under a store $\sta$ (typed as $\match{}{\sta}{\G}$),
there may be not-yet reduced parallel bindings in $e$ and they can be referred to in $s$.
Therefore, to reason about the focused term $s$
we have to extract and type the bindings from the evaluation context $e$ as well (using the judgment $\match{\G}{e}{\Delta}$),
so that $s$ is well-typed under $\G, \Delta$.

In the proof we use $\match{}{\sta; e}{\G; \Delta}$ as a shorthand for
$\match{}{\sta}{\G} \wedge \match{\G}{e}{\Delta}$.
Also, $\sta\vdash t$ denotes
$\exists \G, T. \match{}{\sta}{\G} \wedge \typ{\G}{t}{T}$.

\subsubsection{Evaluation Context Inversion}
\label{sec:evctx-inv}

\begin{definition}[Evaluation context inversion]
  \label{defn:evctx-inversion}
  We say $\typec{\G}{e}{\Delta}{U}{s}{T}$ iff
  (i) $\match{\G}{e}{\Delta}$,
  (ii) $\forall s^\prime.$
  $\typ{\G, \Delta}{s'}{U}$
  and $\greyed{\subi{\G; \Delta}{s^\prime}{s}}$
  imply $\typ{\G}{\evctx{s^\prime}}{T}$.
\end{definition}
The inversion of an evaluation context 
characterizes the terms that can be plugged into the context
while preserving the typing of the entire term.
The notion $\greyed{\subi{\G;\Delta}{s'}{s}}$ denotes that $s'$ is \emph{fresher} than $s$.
This means that replacing $s$ with $s'$ in the evaluation context 
preserves the separation checks,
as stated formally in the following definition:
\begin{definition}[Fresher Term]
	The term $s$ is considered fresher than another term $t$ under typing context $\G, \Delta$
  (written $\subi{\G; \Delta}{s}{t}$),
  iff
  given any $\Delta_1, \Delta_2, e$,
  such that $\Delta = \Delta_1, \Delta_2$,
  $\match{\G, \Delta_1}{e}{\Delta_2}$
  we have
  $\forall C.$
  $\ninter{\G, \Delta_1}{\evctx{t}}{C}$
  implies
  $\ninter{\G, \Delta_1}{\evctx{s}}{C}$.
\end{definition}
This notion is necessary to ensure that the separation checks involved in the typing of the entire term
are preserved when replacing $s$ with $s'$.

\subsubsection{Well-formed Environment}

In the metatheory we assume that all the environments (or typing contexts) we deal with are well-formed.
An environment $\G$ is well-formed if all bindings it contains are well-formed in the defining environment, i.e. given $\G_0, x : T$, we have $\wfTyp{\G_0}{T}$.
Since well-formedness is implicitly assumed,
all the transformations on the environments should preserve well-formedness.

\subsubsection{Inertness}

We introduce the notion of inertness,
which is a property for typing contexts,
to reflect the idea that the separation degrees of the bindings are well-grounded.
Specifically, the separation degree should either
\begin{itemize}
  \item be introduced by a mutable variable binding,
    as a fresh mutable variable can specify an arbitrary separation degree;
  \item or be derivable from the capture set of the type.
    In other words, given a binding $x :_D \tCap{S}{C}$,
    $D$ is consistent with the capture set $C$ if $C\bowtie D$.
\end{itemize}

\begin{definition}[Inert environment]
	We say $\G$ is inert iff
  $\forall x :_{D} T \in \G$,
  either
  (1) $\univ{} \in \new{\cs{T}}$,
  then $\G = \G_1, x :_{D} T, \G_2$ for some $\G_1, \G_2$
  implies that
  $D = \dom{\G_1}$;
  or (2) $\univ{} \notin \cs{T}$,
  then $\G = \G_1, x :_{D} T, \G_2$ for some $\G_1, \G_2$
  implies that
  $\ninter{\G_1}{D}{{\cs{T}}}$.
\end{definition}

$\new{\cs{T}}$ denotes the capture set part of $T$, i.e. $\cs{S\capt C} = C$.

The following facts can be straightforwardly verified by inspecting the definition of store and evaluation context typing.

\begin{fact}
	$\match{}{\sta}{\G}$ implies that $\G$ is inert.
\end{fact}

\begin{fact}
	$\match{\G}{e}{\Delta}$ implies that $\Delta$ is inert.
\end{fact}

\subsubsection{Binding Depth}

Binding depth is the ``index'' of a binding in the context.
It helps us to define the induction schemes on capture sets.

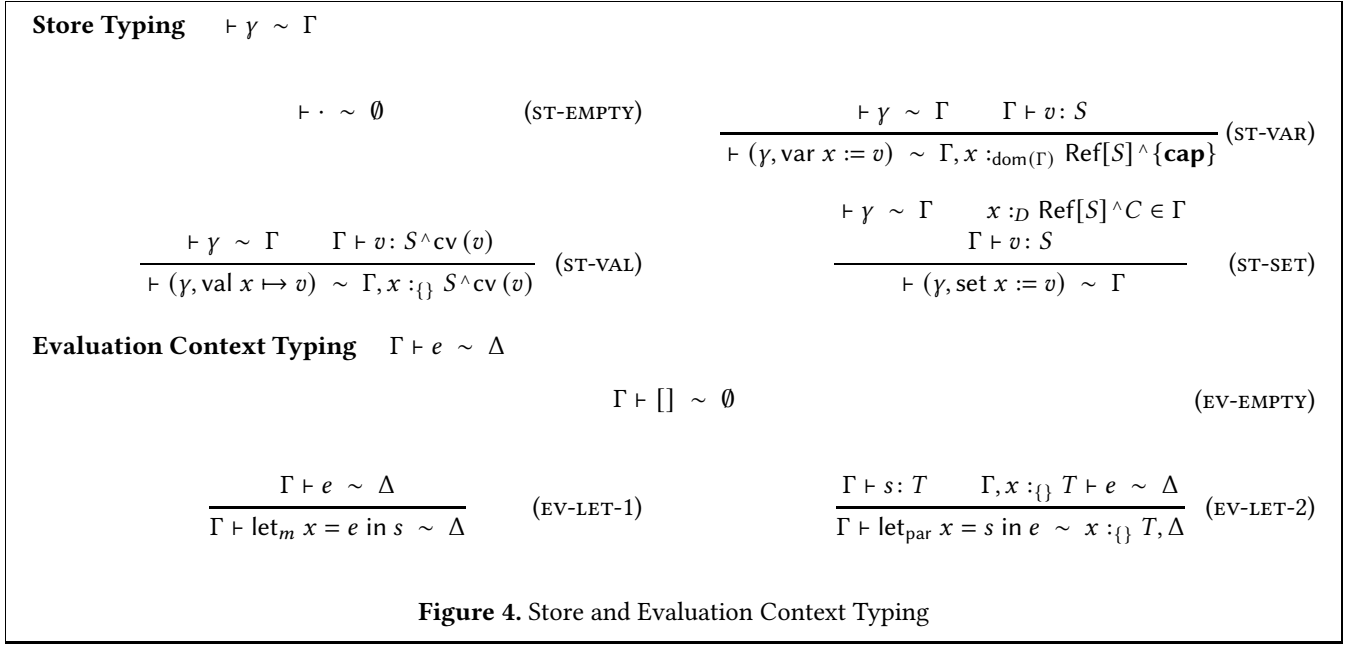
\begin{definition}[Depth]
	We define the depth of the variable $x$ in an environment $\G$
  as the index of $x$ in $\G$, i.e.
  \begin{align*}
    &\fDepth{\G, x : T}{x} \quad &= \quad &| \G | \\
    &\fDepth{\G, x : T}{y} \quad &= \quad &\fDepth{\G}{x} \\
  \end{align*}
\end{definition}

Note that $\fDepth{\G}{x}$ is a partial function:
it is only defined on the domain of $\G$.
Notably, it is not defined on the special root capabilities
$\univ$ and $\rdroot$.

\begin{definition}[Depth of Capture Set]
  We define the depth of a capture set $C$, written $\fEmbed{C}$, as the maximal depth of the variables in $C$,
  i.e.
  \begin{equation*}
    \fEmbed{C} = \max_{x \in C\setminus\set{\univ{}, \rdroot{}}} \fDepth{\G}{x}.
  \end{equation*}
  Specially, we let $\fEmbed{C} = -1$
  if $C\setminus\set{\univ{}, \rdroot{}}$ is empty.
\end{definition}

In Lemma \ref{lemma:widening-pres-ninter}, the induction is carried out on a lexical order of 
the depth of the capture set and the height of the derivation tree.

\subsubsection{Auxilliary Judgment for \textsf{is-reader}}

The \ruleref{sc-reader} subcapturing rule uses the $\isrdr{\G}{x}$ judgment,
which itself depends on subtyping.
This results in the subcapturing and the subtyping rules being mutually dependent,
which complicates the induction scheme on this two judgments.

To disentangle subcapturing and subtyping rules,
we define the following judgment which is equivalent to the check implemented by \textsf{is-reader},
but eliminates the dependency on subtyping.


  \infax[\ruledef{rd-reader}]
  {\isrdrtype{\G}{\tRdr{S}\capt C}}

  \infrule[\ruledef{rd-tvar}]
  {X<:S \in \G\andalso \isrdrtype{\G}{S\capt C}}
  {\isrdrtype{\G}{X\capt C}}


Later, we will prove the equivalence between $\isrdrtype{\G}{T}$ and $\exists C, S. \sub{\G}{T}{\tRdr{S}\capt C}$.

\subsection{Soundness}

\subsubsection{Properties of Subcapturing}

\begin{lemma}[Decomposition of subcapturing]
  \label{lemma:subcapt-decompose}
	$\subDft{C_1}{C_2}$
  implies
  $\forall x \in C_1. \subDft{\set{x}}{C_2}$.
\end{lemma}

\begin{proof}
	By induction on the subcapture derivation.

  \emph{Case \ruleref{sc-var}, \ruleref{sc-elem}, \ruleref{sc-rdr-cap} and \ruleref{sc-reader}}.
  The proof is concluded trivially
  since in these cases $C_1 = \set{x}$.

  \emph{Case \ruleref{sc-set}}.
  This case follows directly from the preconditions.

  \emph{Case \ruleref{sc-trans}}.
  Then $\subDft{C_1}{C}$ and $\subDft{C}{C_2}$.
  By using the IH, we have
  $\overline{\subDft{\set{x}}{C}}^{x \in C_1}$,
  from which we can conclude by using the \ruleref{sc-trans} rule repeatedly.
\end{proof}







\begin{lemma}[Reflexivity of subcapturing]
	$\subDft{C}{C}$.
\end{lemma}

\begin{proof}
	We begin by showing that $\overline{\subDft{\set{x}}{C}}^{x \in C}$ using the \ruleref{sc-elem} rule.
  Afterwards,
  we may conclude this case by the \ruleref{sc-set} rule.
\end{proof}

\begin{lemma}[Set inclusion implies subcapturing]
	$C_1 \subseteq C_2$ implies $\subDft{C_1}{C_2}$.
\end{lemma}

\begin{proof}
	By repeated \ruleref{sc-elem} and \ruleref{sc-set}.
\end{proof}

\begin{lemma}
  \label{lemma:subcapt-join-left}
	$\sub{\G}{C_1}{C}$
  and $\sub{\G}{C_2}{C}$
  implies $\sub{\G}{C_1 \cup C_2}{C}$.
\end{lemma}

\begin{proof}
  We begin by showing that
  $\overline{\subDft{\set{x}}{C}}^{x \in C_1}$\\
  and $\overline{\subDft{\set{x}}{C}}^{x \in C_2}$
  using Lemma \ref{lemma:subcapt-decompose}.
  Now we can conclude this case by applying the \ruleref{sc-set} rule.
\end{proof}

\begin{lemma}
  \label{lemma:subcapt-join-right}
	$\sub{\G}{C}{C_1}$,
  implies $\sub{\G}{C}{C_1 \cup C_2}$.
\end{lemma}

\begin{proof}
	By induction on the subcapture derivation.

  \emph{Case \ruleref{sc-set}}.
  Then $C = \set{x}$ and $x \in C_1$.
  We have $x \in C_1 \cup C_2$
  and conclude this case by applying the \ruleref{sc-set} rule again.

  \emph{Other cases.}
  By IH and the same rule.
\end{proof}

\begin{corollary}
  \label{coro:subcapt-join-both}
	$\sub{\G}{C_1}{C_2}$
  and $\sub{\G}{D_1}{D_2}$
  implies $\sub{\G}{C_1 \cup D_1}{C_2 \cup D_2}$.
\end{corollary}





\begin{lemma}[Capture set is irrelevant in reader checking]
  \label{lemma:captset-irrelevant-rdr-checking}
	If $\isrdrtype{\G}{S\capt C}$
  then $\isrdrtype{\G}{S\capt C'}$.
\end{lemma}

\begin{proof}
  By straightforward induction on the derivation.
  In the \ruleref{rd-reader} case we conclude from the premise immediately.
  In the \ruleref{rd-tvar} case we conclude by the IH and the same rule.
\end{proof}

\begin{lemma}[Subtyping preserves reader checking]
  \label{lemma:subtype-pres-rdr-checking}
	If $\isrdrtype{\G}{T}$ and $\sub{\G}{T'}{T}$
  then $\isrdrtype{\G}{T'}$.
\end{lemma}
\begin{proof}
  By induction on the subtyping derivation.

  \emph{Case \ruleref{refl}}. Immediate.

  \emph{Case \ruleref{capt}}.
  By the IH and Lemma \ref{lemma:captset-irrelevant-rdr-checking}.

  \emph{Case \ruleref{trans}}. By repeated application of the IH.

  \emph{Case \ruleref{tvar}}.
  Then $T' = X$ and $X <: T \in \G$.
  We conclude by using the \ruleref{rd-tvar} rule.
\end{proof}

\begin{lemma}[Equivalence between reader checking]
  \label{lemma:is-rdr-eqv}
  $\isrdrtype{\G}{T}$ iff $\exists C, S.\ \sub{\G}{T}{\tRdr{S}\capt C}$.
\end{lemma}

\begin{proof}
	We prove the two directions in the equivalence respectively.

  \textbf{($\Rightarrow$)}:
  Proceed the proof by induction on the derivation.
  In the \ruleref{rd-reader} case, we conclude immediately by the reflexivity of subtyping.
  In the \ruleref{rd-tvar} case, we conclude by the IH and the \ruleref{tvar} rule.

  \textbf{($\Leftarrow$)}.
  By induction on the subtyping derivation.
  In the \ruleref{refl} case we conclude immediately by the \ruleref{rd-reader} rule.
  In the \ruleref{trans} case, 
  we have $\sub{\G}{T}{T'}$ and $\sub{\G}{T'}{C\,\tRdr{S}}$ for some $T'$.
  We first use the IH to show that
  $\isrdrtype{\G}{T'}$.
  Then we invoke Lemma \ref{lemma:subtype-pres-rdr-checking} to conclude this case.
  Finally, the \ruleref{tvar} case can be concluded immediately using the \ruleref{rd-tvar} rule.
\end{proof}





\begin{lemma}[Reader checking strengthening]
	Given $\G = \G_1, \Delta, \G_2$,
  if $\G' = \G_1, \G_2$ is still well-formed,
  $\isrdrtype{\G}{T}$,
  and $T$ is well-formed in $\G'$,
  then $\isrdrtype{\G'}{T}$.
\end{lemma}

\begin{proof}
	By straightforward induction on the derivation.
  In the \ruleref{rd-reader} case we conclude immediately using the same rule.
  In the \ruleref{rd-tvar} case, we have $T = X\capt C$,
  $X<:R \in \G$,
  and $\isrdrtype{\G}{R}$.
  By the well-formedness of $\G'$ we can show that
  $R$ is well-formed in $\G'$.
  Then we conclude by the IH and the \ruleref{rd-tvar} rule.
\end{proof}

\begin{corollary}[\textsf{is-reader} strengthening]
  \label{coro:isrdr-strengthening}
	Given $\G = \G_1, \Delta, \G_2$,
  if $\G' = \G_1, \G_2$ is still well-formed,
  $\isrdrtype{\G}{T}$,
  and $x \in \dom{\G'}$,
  then $\isrdr{\G'}{x}$.
\end{corollary}

\begin{lemma}[Subcapture strengthening]
  \label{lemma:subcapt-strengthening}
	Given $\G = \G_1, \Delta, \G_2$,
  if $\G^\prime = \G_1, \G_2$ is still well-formed,
  and $\sub{\G}{C_1}{C_2}$,
  then $\sub{\G^\prime}{C_1\setminus \dom{\Delta}}{C_2\setminus \dom{\Delta}}$.
\end{lemma}

\begin{proof}
	By induction on the subcapture derivation.

  \emph{Case \ruleref{sc-trans}}.
  Then $\sub{\G}{C_1}{C}$ and $\sub{\G}{C}{C_2}$ for some $C$.
  By the IH we can show that
  $\sub{\G^\prime}{C_1\setminus \dom{\Delta}}{C\setminus \dom{\Delta}}$
  and $\sub{\G'}{C\setminus \dom{\Delta}}{C_2\setminus \dom{\Delta}}$.
  Hence we conclude using the \ruleref{sc-trans} rule.

  \emph{Case \ruleref{sc-var}}.
  Then $C_1 = \set{x}$,
  and $x :_D \tCap{C}{S} \in \G$.
  Proceed by a case analysis on whether $x$ is bound in $\Delta$.
  \begin{itemize}
  \item If $x \in \dom{\Delta}$,
    we have $C_1\setminus \dom{\Delta} = \varemptyset$
    and can conclude this case by (\textsc{sc-set}).

  \item If $x \notin \dom{\Delta}$,
    by the well-formedness of $\G^\prime$ we have $C \cap \dom{\Delta} = \varemptyset$,
    which implies that $C\setminus \dom{\Delta} = C$.
    Note that $\set{x}\setminus \dom{\Delta} = \set{x}$.
    We can therefore apply the \ruleref{sc-var} rule to conclude.
  \end{itemize}

  \emph{Case \ruleref{sc-elem}}.
  Then $C_1 = \set{x}$
  and $x \in C_2$.
  Again we proceed by a case analysis on whether $x$ is bound in $\Delta$.
  \begin{itemize}
  \item If $x \in \dom{\Delta}$,
    we have $C_1\setminus \dom{\Delta} = \varemptyset$ and
    thus conclude this case using the \ruleref{sc-set} rule.

  \item Otherwise if $x \notin \dom{\Delta}$,
    we have $x \in C_2\setminus \dom{\Delta}$
    and can conclude this case by the \ruleref{sc-elem} rule.
  \end{itemize}

  \emph{Case \ruleref{sc-set}}.
  We conclude by repeated IH and the same rule.

  \emph{Case \ruleref{sc-rdr-cap}}.
  We conclude immediately using the same rule
  since $\set{\rdroot}\setminus \dom{\Delta} = \set{\rdroot}$
  and $\set{\univ}\setminus \dom{\Delta} = \set{\univ}$.

  \emph{Case \ruleref{sc-reader}}.
  Then $C_1 = \set{x}$,
  $\isrdr{\G, x :_D P, \Delta}{x}$,
  and $C_2 = \set{\rdroot}$.
  If $x \in \dom{\Delta}$,
  then $C_1 \setminus \dom{\Delta} = \set{}$ and we can conclude directly.
  Otherwise, if $x \notin \dom{\Delta}$,
  we can show that $x\in \dom{\G'}$,
  and then use Corollary~\ref{coro:isrdr-strengthening} to conclude
  to show that $\isrdr{\G'}{x}$.
  Note that $\set{\rdroot}\setminus \dom{\Delta} = \set{\rdroot}$.
  Now we can conclude this case by the \ruleref{sc-reader} rule.
\end{proof}

\subsubsection{Properties of Typing and Subtyping}

\begin{lemma}[Subtype inversion: type variable]
  \label{lemma:subtype-inversion-tvar}
  If $\subDft{U}{\tCap{C}{X}}$,
  then $U = \tCap{C^\prime}{Y}$
  for some $C^\prime, Y$,
  such that
  $\subDft{C^\prime}{C}$,
  and $\subDft{Y}{X}$.
\end{lemma}

\begin{proof}
  By induction on the subtype derivation, wherein only the following cases are possible.

  \emph{Case \ruleref{refl}}. Immediate.

  \emph{Case \ruleref{tvar}}.
  Then $U = Y$,
  $Y <: X \in \G$,
  and $C = \set{}$.
  Now we conclude by applying the \ruleref{tvar} rule again.

  \emph{Case \ruleref{trans}}.
  Then $\subDft{U}{U^\prime}$
  and $\subDft{U^\prime}{\tCap{C}{X}}$
  for some $U^\prime$.
  By IH,
  we can first show that
  $U^\prime = \tCap{C^\prime}{Y}$ for some $C^\prime$ and $Y$.
  Now, we can invoke IH on the derivation $\subDft{U}{\tCap{C^\prime}{Y}}$
  to show that
  $U = \tCap{C^{\prime\prime}}{Z}$,
  $\subDft{C^{\prime\prime}}{C^\prime}$
  and $\subDft{Z}{Y}$.
  Finally we conclude by the transitivity of both subcapturing and subtyping.

  \emph{Case \ruleref{capt}}.
  Then $U = \tCap{C'}{S}$ for some $C, S$,
  $\subDft{C'}{C}$,
  and $\subDft{S}{X}$.
  Now we invoke the IH To show that
  $S = Y\capt C''$ for some $Y$, where $C'' = \set{}$
  (note that we consider $S$ to be equivalent to a capturing type with an empty capture set),
  and $\subDft{Y}{X}$.
  This case is therefore concluded.
\end{proof}

\begin{lemma}[Subtype inversion: mutable reference]
  \label{lemma:subtype-inversion-mut}
  If $\subDft{U}{\tRef{S}\capt C}$,
  then
  either (i) $U$ is of the form $\tCap{C^\prime}{X}$,
  $\subDft{C^\prime}{C}$
  and $\subDft{X}{\tRef{S}}$,
  or (ii) $U$ is of the form $\tCap{C^\prime}{\tRef{S}}$,
  and $\subDft{C^\prime}{C}$.
\end{lemma}

\begin{proof}
	By induction on the subtype derivation.
  wherein only the following cases apply.

  \emph{Case \ruleref{refl}}.
  Then $U = \tRef{S}\capt C$.
  This case is concluded immediately.

  \emph{Case \ruleref{tvar}}.
  Then $U = X$
  and $X <: \tRef{S} \in \G$.
  We conclude this case by the \ruleref{tvar} rule.

  \emph{Case \ruleref{capt}}.
  By IH.
\end{proof}

\begin{lemma}[Subtype inversion: reader]
  \label{lemma:subtype-inversion-reader}
  If $\subDft{U}{\tRdr{S}\capt C}$,
  then
  either (i) $U$ in the form of $X\capt C'$
  where $\subDft{C'}{C}$ and $\subDft{X}{\tRdr{S}}$,
  or (ii) $U$ is in the form of $C'\,\tRdr{S}$,
  where $\subDft{C'}{C}$.
\end{lemma}

\begin{proof}
	Analogous to the proof of Lemma \ref{lemma:subtype-inversion-mut}.
\end{proof}

\begin{lemma}[Subtype inversion: term abstraction]
  \label{lemma:subtype-inversion-fun}
  If $\subDft{P}{\forall(x :_D U) T\capt C}$,
  then
  either (i) $P$ is of the form $\tCap{C^\prime}{X}$,
  $\subDft{C^\prime}{C}$
  and $\subDft{X}{\forall(x :_D U) T}$,
  or (ii) $P$ is of the form $\tCap{C^\prime}{\forall(x :_D U^\prime) T^\prime}$ such that
  $\subDft{C^\prime}{C}$,
  $\subDft{U}{U^\prime}$,
  and $\sub{\G, x :_D U^\prime}{T^\prime}{T}$.
\end{lemma}

\begin{proof}
	By induction on the subtype derivation.

  \emph{Case \ruleref{refl}}.
  Then $U = \forall(x :_D U) T\capt C$.
  We conclude immediately
  by the reflexivity of subcapture and subtyping.

  \emph{Case \ruleref{tvar}}.
  Then $P = X$,
  $C = \set{}$,
  and $X <: \forall(x :_D U) T$.
  Then we conclude immediately.

  \emph{Case \ruleref{fun}}.
  Then $P = \forall(x :_D U^\prime) T^\prime\capt C'$
  and we conclude from the preconditions.

  \emph{Case \ruleref{trans}}.
  Then $\subDft{P}{P^\prime}$
  and $\subDft{P^\prime}{C\,\forall(x :_D U) T}$
  for some $P^\prime$.
  By IH we can show that
  $P^\prime$ is either of the form $\tCap{C^\prime}{X}$
  such that
  $\subDft{C^\prime}{C}$
  and $\subDft{X}{\forall(x :_D T) U}$,
  or $P^\prime = C^\prime\,\forall(x :_D T^\prime) U^\prime$,
  such that $\subDft{C^\prime}{C}$,
  $\subDft{T}{T^\prime}$,
  and $\sub{\G, x :_D T^\prime}{U^\prime}{U}$.
  In the first case,
  we invoke Lemma \ref{lemma:subtype-inversion-tvar}
  to show that
  $P = \tCap{C_1}{Y}$,
  $\subDft{C_1}{C^\prime}$,
  and $\subDft{Y}{X}$.
  Now we can conclude by the transitivity of subcapturing and subtyping.
  In the other case,
  we invoke IH again on the first subtype derivation
  and conclude by the transitivity of subcapturing and subtyping.

  \emph{Case \ruleref{capt}}.
  By IH.
\end{proof}

\begin{lemma}[Subtype inversion: type abstraction]
  \label{lemma:subtype-inversion-tfun}
  If $\subDft{P}{\forall[X <: S] T\capt C}$,
  then
  either (i) $P$ is of the form $\tCap{C^\prime}{Y}$,
  $\subDft{C^\prime}{C}$
  and $\subDft{Y}{\forall[X <: S] T}$,
  or (ii) $P$ is of the form $\tCap{C^\prime}{\forall[X <: S^\prime] T^\prime}$ such that
  $\subDft{C^\prime}{C}$,
  $\subDft{S}{S^\prime}$,
  and $\sub{\G, X <: S^\prime}{T^\prime}{T}$.
\end{lemma}

\begin{proof}
  Analogous to the proof of Lemma \ref{lemma:subtype-inversion-fun}.
\end{proof}

\begin{lemma}[Subtype inversion: boxed term]
  \label{lemma:subtype-inversion-boxed}
  If $\subDft{P}{(\Box\ T)\capt C}$,
  then
  either (i) $P$ is of the form $\tCap{C^\prime}{Y}$,
  $\subDft{C^\prime}{C}$
  and $\subDft{Y}{\Box\ T}$,
  or (ii) $P$ is of the form $\tCap{C^\prime}{(\Box\ T^\prime)}$ such that
  $\subDft{C^\prime}{C}$,
  and $\sub{\G}{T^\prime}{T}$.
\end{lemma}

\begin{proof}
  Analogous to the proof of Lemma \ref{lemma:subtype-inversion-fun}.
\end{proof}

\subsubsection{Properties of Separation}
Now we establish some properties of separation checking.

\begin{lemma}[Separation checking inversion: elements]
  \label{lemma:ninter-to-elem}
	If $\ninter{\G}{C_1}{C_2}$,
  then
  (1) $\forall x \in C_1$ we have $\ninter{\G}{x}{C_2}$;
  and (2) $\forall x \in C_2$ we have $\ninter{\G}{C_1}{x}$.
\end{lemma}

\begin{proof}
  By induction on the derivation.

  \emph{Case \ruleref{ni-symm}}.
  Concluding by swapping the two conclusions in the IH.

  \emph{Case \ruleref{ni-set}}.
  Then we have $\overline{\ninter{\G}{x}{C_2}}^{x \in C_1}$.
  The first part of the goal is immediate.
  Now we show the second part of the goal.
  By applying IH repeatedly we can deduce that
  $\forall x_1 \in C_1, \forall x_2 \in C_2$
  we have $\ninter{\G}{x_1}{x_2}$.
  We therefore show that
  $\forall x_2 \in C_2$
  we have $\ninter{\G}{C_1}{x_2}$
  by \ruleref{ni-symm} and \ruleref{ni-set},
  thus concluding this case.

  \emph{Case \ruleref{ni-degree}, \ruleref{ni-var} and \ruleref{ni-reader}}.
  These cases are immediate since both $C_1$ and $C_2$ are singletons.
\end{proof}

\begin{lemma}[Reader capability specialization]
  \label{lemma:reader-specialization}
  Given any $\G$ and $h$, we have:
  (1) $| \ninter{\G}{\set{\rdroot}}{C_2} | \le h$
  implies $\ninter{\G}{\set{x}}{C_2}$ for every $x$ such that
  $\isrdr{\G}{x}$;
  and (2) $| \ninter{\G}{C_1}{\set{\rdroot}} | \le h$
  implies $\ninter{\G}{C_1}{\set{x}}$ for every $x$ such that
  $\isrdr{\G}{x}$.
\end{lemma}

\begin{proof}
  By induction on the derivation depth $h$.
  We prove each of the conclusion respectively, starting by analyzing the cases of the first one.

  \emph{Case \ruleref{ni-symm} and \ruleref{ni-set}}.
  By applying the IH.

  \emph{Case \ruleref{ni-degree} and \ruleref{ni-var}}.
  Not applicable.

  \emph{Case \ruleref{ni-reader}}.
  Then $C_2 = \set{y}$ and $\subDft{\set{y}}{\set{\rdroot}}$.
  By the \ruleref{sc-reader} we can show that
  $\sub{\G}{\set{x}}{\set{\rdroot}}$.
  This case can therefore be concluded using the \ruleref{ni-reader} rule.

  Now we inspect the derivation in the second case.

  \emph{Case \ruleref{ni-symm}}. By the IH.

  \emph{Case \ruleref{ni-set}}.
  Then $\overline{\ninter{\G}{x}{\set{\rdroot}}}^{x \in C_1}$.
  By the IH we can show that given any $y$
  such that $\isrdr{\G}{y}$,
  we can show that
  $\overline{\ninter{\G}{x}{y}}^{x \in C_1}$.
  We can therefore conclude this case by the \ruleref{ni-set} rule.

  \emph{Case \ruleref{ni-degree}}. Not applicable.

  \emph{Case \ruleref{ni-var}}. By the IH and the same rule.

  \emph{Case \ruleref{ni-reader}}.
  This case can be concluded analogously to the one in the previous subgoal.
\end{proof}

\begin{lemma}[Universal capability specialization]
  \label{lemma:univ-specialization}
  Given any $\G$ and $h$, we have:
  (1) $| \ninter{\G}{\set{\univ}}{C_2} | \le h$
  implies $\ninter{\G}{C}{C_2}$ for any $C$,
  and 
  (2) $| \ninter{\G}{C_1}{\set{\univ}} | \le h$
  implies $\ninter{\G}{C_1}{C}$ for any $C$.
\end{lemma}

\begin{proof}
  By induction on the derivation depth $h$.
  We establish the two conclusions respectively.
  We start by showing that we can prove the first one in each case.

  \emph{Case \ruleref{ni-symm} and \ruleref{ni-set}}.
  By using the IH.

  \emph{Case \ruleref{ni-degree} and \ruleref{ni-var}}.
  Not applicable since $\univ{} \notin \dom{\G}$.

  \emph{Case \ruleref{ni-reader}}.
  Then $\subDft{\set{\univ{}}}{\set{\rdroot{}}}$.
  By induction on the subcapturing derivation
  we can derive a contradiction in each case,
  rendering this case impossible.

  Then we show that we can prove the second conclusion in each case.
  \emph{Case \ruleref{ni-symm} and \ruleref{ni-set}}.
  By the IH.

  \emph{Case \ruleref{ni-degree}}.
  Then $C_1 = \set{x}$, $x:_D T \in \G$ and $\univ{} \in D$.
  This is contradictory since the $D$ cannot contain $\univ{}$.

  \emph{Case \ruleref{ni-var}}.
  Then $C_1 = \set{x}$, $x: S\capt C' \in \G$ and $\ninter{\G}{\set{\univ{}}}{C'}$.
  By the IH we can show that $\ninter{\G}{C}{C'}$ for any $C$,
  and conclude by the \ruleref{ni-var} rule.

  \emph{Case \ruleref{ni-reader}}.
  Similarly we have $\sub{\G}{\set{\univ{}}}{\set{\rdroot{}}}$
  and derive a contradiction from it.
\end{proof}

\begin{lemma}[Subcapture preserves separation]
  \label{lemma:subcapt-pres-ninter}
	If $\ninter{\G}{C_1}{C_2}$
  and $\sub{\G}{C_0}{C_1}$,
  then $\ninter{\G}{C_0}{C_2}$.
\end{lemma}

\begin{proof}
	By induction on the subcapture derivation.

  \emph{Case \ruleref{sc-trans}}.
  Then $\sub{\G}{C_0}{C}$ and $\sub{\G}{C}{C_1}$ for some $C$.
  We conclude by applying the IH twice.

  \emph{Case \ruleref{sc-var}}.
  Then $C_0 = \set{x}$,
  $x : \tCap{C}{S} \in \G$,
  and $\sub{\G}{C}{C_1}$.
  By IH we can show that
  $\ninter{\G}{C}{C_2}$,
  and conclude by the \ruleref{ni-var} rule.

  \emph{Case \ruleref{sc-elem}}.
  Then $C_0 = \set{x}$ and $x \in C_1$.
  We conclude this case by Lemma \ref{lemma:ninter-to-elem}.

  \emph{Case \ruleref{sc-set}}.
  Then $\overline{\ninter{\G}{x}{C_2}}^{x \in C_1}$.
  This case can be concluded by applying IH repeatedly and using the \ruleref{ni-set} rule.

  \emph{Case \ruleref{sc-reader}}.
  Then $C_0 = \set{x}$,
  $\isrdr{\G}{x}$,
  and $C_1 = \set{\rdroot}$.
  Now we can conclude this case by invoking Lemma \ref{lemma:reader-specialization}.

  \emph{Case \ruleref{sc-rdr-cap}}.
  Then $C_0 = \set{\rdroot}$
  and $C_1 = \set{\univ}$.
  Now we conclude by invoking Lemma \ref{lemma:univ-specialization}.
\end{proof}

\begin{corollary}[Set inclusion preserves separation]
  \label{coro:subset-pres-ninter}
	If $\ninter{\G}{C_1}{C_2}$
  and ${C_0} \subseteq {C_1}$,
  then $\ninter{\G}{C_0}{C_2}$.
\end{corollary}

\begin{lemma}[Evaluation context reification over subcapture]
  \label{lemma:evctx-reification-subcapt}
  If
  (i) $\match{\G}{e}{\Delta}$,
  (ii) $\sub{\G, \Delta}{\cv{s^\prime}}{\cv{s}}$,
  and (iii) $s$ is not a value,
  then
  $\sub{\G}{\cv{\evctx{s^\prime}}}{\cv{\evctx{s}}}$.
\end{lemma}

\begin{proof}
	By induction on $e$.

  \emph{Case $e = []$}.
  Immediate.

  \emph{Case $e = \tLetMode{m}{x}{e^\prime}{u}$}.
  Then we have
  $\cv{\evctx{s}} = \cv{\evpctx{s}} \cup \cv{u}\setminus \set{x}$,
  and $\cv{\evctx{s^\prime}} \subseteq \cv{\evpctx{s^\prime}} \cup \cv{u}\setminus \set{x}$.
  By the reflexivity of subcapturing we have
  $\sub{\G}{\cv{u}\setminus \set{x}}{\cv{u}\setminus \set{x}}$.
  By IH, we have
  $\sub{\G}{\cv{\evpctx{s^\prime}}}{\cv{\evpctx{s^\prime}}}$.
  By Corollary \ref{coro:subcapt-join-both},
  we have
  $\sub{\G}{\cv{\evpctx{s^\prime}} \cup \cv{u}\setminus \set{x}}{\cv{\evpctx{s^\prime}} \cup \cv{u}\setminus \set{x}}$.
  We can therefore conclude this case.

  \emph{Case $e = \tLetMode{\parmode}{x}{t}{e^\prime}$}.
  By inspecting the derivation of $\match{\G}{e}{\Delta}$,
  we can show that
  $\Delta = x :_D T, \Delta^\prime$,
  $\typ{\G}{t}{T}$,
  and $\match{\G, x :_D T}{e^\prime}{\Delta^\prime}$.
  Also, we have
  $\cv{\evctx{s}} = \cv{t} \cup \cv{\evpctx{s}}\setminus \set{x}$
  and $\cv{\evctx{s^\prime}} = \cv{t} \cup \cv{\evpctx{s^\prime}}\setminus \set{x}$.
  By the reflexivity of subcapture,
  we have $\sub{\G}{\cv{t}}{\cv{t}}$.
  By IH,
  we have $\sub{\G, x :_D T}{\cv{\evpctx{s^\prime}}}{\cv{\evpctx{s}}}$.
  Now we invoke Lemma \ref{lemma:subcapt-strengthening}
  and show that
  $\sub{\G}{\cv{\evpctx{s^\prime}}\setminus \set{x}}{\cv{\evpctx{s}}\setminus \set{x}}$.
  We can conclude this case by Corollary \ref{coro:subcapt-join-both}.
\end{proof}

\begin{lemma}[Subcapturing implies fresher terms]
  \label{lemma:subcapt-to-subi}
	Given the environment $\G, \Delta$ and two terms $s, s^\prime$,
  if (i) $s$ is not a value,
  and (ii) $\sub{\G, \Delta}{\cv{s^\prime}}{\cv{s}}$,
  then $\subi{\G; \Delta}{s^\prime}{s}$.
\end{lemma}

\begin{proof}
  Given any $\Delta_1, \Delta_2$ and $e$ such that
  $\Delta = \Delta_1, \Delta_2$
  and $\match{\G, \Delta_1}{e}{\Delta_2}$,
  the goal is to show that
  $\forall C$
  $\ninter{\G, \Delta_1}{\cv{\evctx{s}}}{C}$
  implies $\ninter{\G, \Delta_1}{\cv{\evctx{s^\prime}}}{C}$.
  By Lemma \ref{lemma:evctx-reification-subcapt}
  we have $\sub{\G, \Delta_1}{\cv{\evctx{s^\prime}}}{\cv{\evctx{s}}}$.
  Now we conclude this case by Lemma \ref{lemma:subcapt-pres-ninter}.
\end{proof}

\begin{corollary}[Captured set inclusion implies fresher terms]
	\label{coro:subset-to-subi}
  If $\cv{s^\prime} \subseteq \cv{s}$
  and $s$ is not a value,
  then $\subi{\G; \Delta}{s^\prime}{s}$.
\end{corollary}







\begin{lemma}[Widening preserves separation]
  \label{lemma:widening-pres-ninter}
  Given an inert environment $\G$,
  a variable $x :_D T \in \G$ where $\univ{} \notin \cs{T}$ 
  and $\sub{\G}{\set{x}}{\set{\rdroot{}}}$ does not hold,
  and a natural number $h$,
  then
  (1) $| \ninter{\G}{x}{C} |\ \le h$ implies $\ninter{\G}{\cs{T}}{C}$,
  and (2) $| \ninter{\G}{C}{x} |\ \le h$ implies $\ninter{\G}{C}{\cs{T}}$.
  Here, $| \ninter{\G}{C_1}{C_2} |$ denotes the height of the derivation tree.
\end{lemma}

\begin{proof}
  By induction on the lexical order $(\fEmbed{C}, h)$.
  We prove the two conclusions respecitively.
  We start by establishing the first one,
  by a case analysis on the last rule applied in the derivation of
  $\ninter{\G}{x}{C}$.

  \emph{Case \ruleref{ni-symm}}.
  Then $\ninter{\G}{C}{x}$.
  This case is concluded by IH.

  \emph{Case \ruleref{ni-set}}. By IH.

  \emph{Case \ruleref{ni-degree}}.
  Then $C = \set{y}$ and $y \in D$.
  By the inertness of $\G$ and that $\univ{} \notin \cs{T}$,
  we have $\ninter{\G}{D}{\cs{T}}$.
  By Lemma \ref{lemma:ninter-to-elem},
  $\ninter{\G}{y}{\cs{T}}$.
  Now we conclude by \ruleref{ni-symm}.

  \emph{Case \ruleref{ni-var}}.
  Then $\ninter{\G}{\cs{T}}{y}$
  and we can conclude this case immediately.

  \emph{Case \ruleref{ni-reader}}.
  This case is not applicable.

  The first conclusion is therefore proven.
  Now we proceed to the second one. 

  \emph{Case \ruleref{ni-symm}}.
  By IH.

  \emph{Case \ruleref{ni-set}}.
  Then $C = \set{y_i}_{i=1,\cdots,n}$,
  and $\ninter{\G}{y_i}{x}$ for $i = 1, \cdots, n$.
  By repeated IH,
  we can demonstrate that $\overline{\ninter{\G}{y_i}{C}}^{i = 1, \cdots, n}$.
  Now we conclude this case by \ruleref{ni-set}.

  \emph{Case \ruleref{ni-degree}}.
  Then $C = \set{x^\prime}$,
  $x^\prime :_{D^\prime} T^\prime \in \G$,
  and $x \in D^\prime$.
  If $\univ{} \in \cs{T^\prime}$,
  then by the inertness of $\G$,
  we have $D' = \dom{\G_0}$
  if we decompose the environment into
  $\G = \G_0, x^\prime :_{D^\prime} T^\prime, \G_1$.
  Since $x \in D'$, by the well-formedness of the environment,
  $x$ is bound in $\G_0$.
  Again by the well-formedness,
  we have $\cs{T} \subseteq \dom{\G_0} = D'$.
  We can therefore prove the goal
  by repeated \ruleref{ni-degree} and \ruleref{ni-set}.
  Otherwise, if $\univ{} \notin \cs{T^\prime}$,
  we have $\ninter{\G}{D^\prime}{\cs{T^\prime}}$.
  By Lemma \ref{lemma:ninter-to-elem}
  we have $\ninter{\G}{x}{\cs{T^\prime}}$.
  Note that $\fEmbed{\cs{T^\prime}} < \fDepth{\G}{x^\prime}$
  by the well-formedness of $\G$.
  We can therefore invoke IH and show that
  $\ninter{\G}{\cs{T}}{\cs{T^\prime}}$
  and conclude the case by \ruleref{ni-symm}.

  \emph{Case \ruleref{ni-var}}.
  Then $C = \set{x^\prime}$,
  $x^\prime :_{D^\prime} T^\prime \in \G$,
  and $\ninter{\G}{\cs{T^\prime}}{x}$.
  We conclude this case by IH and the same rule.

  \emph{Case \ruleref{ni-reader}}.
  Not applicable.
\end{proof}

\begin{lemma}[Widening implies fresher terms]
  \label{lemma:widening-to-subi}
  Given $\G, \Delta, s, s^\prime$ where $s$ is not a value,
  if (i) $\G, \Delta$ is inert,
  (ii) $x :_D T \in \G$, and $x \in \cv{s}$,
  (iii) $\cv{s^\prime} \subseteq \cv{s} - \set{x, x^\prime} \cup \cs{T}$,
  then $\subi{\G; \Delta}{s^\prime}{s}$.
\end{lemma}

\begin{proof}
  Given $\Delta_1, \Delta_2, e$ such that
  $\Delta = \Delta_1, \Delta_2$
  and $\match{\G, \Delta_1}{e}{\Delta_2}$,
  and we have
  $\ninter{\G}{\cv{\evctx{s}}}{C}$.
  Since $s$ is not a value,
  we have $x \in \cv{s} \subseteq \cv{\evctx{s}}$.
  By Lemma \ref{lemma:ninter-to-elem}
  we show that $\ninter{\G}{x}{C}$.
  By Lemma \ref{lemma:widening-pres-ninter},
  we have $\ninter{\G}{\cs{T}}{C}$.
  Therefore,
  $\ninter{\G}{\cv{\evctx{s}} \cup \cs{T}}{C}$.
  Now we show that
  $\cv{\evctx{s^\prime}} \subseteq \cv{\evctx{s}} \cup \cs{T}$
  and conclude this case.
\end{proof}

\begin{lemma}[Fresher terms: boundary shifting]
  \label{lemma:subi-shift-boundary}
  Given $\G, \Delta_1, \Delta_2$,
  $\subi{\G; \Delta_1, \Delta_2}{s^\prime}{s}$
  implies
  $\subi{\G, \Delta_1; \Delta_2}{s^\prime}{s}$.
\end{lemma}

\begin{proof}
	Given any $\Delta_{3}, \Delta_4, e$ such that
  $\Delta_2 = \Delta_{3}, \Delta_4$
  and $\match{\G, \Delta_1, \Delta_3}{e}{\Delta_4}$,
  if $\ninter{\G, \Delta_1, \Delta_3}{\evctx{s}}{C}$,
  we can show that
  $\ninter{\G, \Delta_1, \Delta_3}{\evctx{s^\prime}}{C}$
  by the premise.
\end{proof}

\begin{lemma}[Subtyping preserves \textsf{is-reader}]
  \label{lemma:subtype-pres-isrdr}
	If $\isrdr{\G}{T}$ and $\sub{\G}{T'}{T}$,
  then $\isrdr{\G}{T'}$.
\end{lemma}

\begin{proof}
	By induction on the derivation of $\sub{\G}{T'}{T}$.
  Note that by inspecting the derivation of $\isrdr{\G}{T}$
  we know that $T$ is either a type variable or a reader,
  which implies that only the following cases are applicable.

  \emph{Case \ruleref{refl}}.
  Then $T' = T$,
  and we conclude immediately from the premise.

  \emph{Case \ruleref{tvar}}.
  Then $T' = X$, $T = S$, and $X <: S \in \G$.
  We conclude immediately by the \ruleref{rd-tvar} rule.
\end{proof}

\subsubsection{Structural Properties of Typing}

\begin{lemma}[Permutation]
  \label{lemma:permutation}
  Given $\G$,
  and $\Delta$ which is a well-formed environment permuted from $\G$:
  \begin{enumerate}[(i)]
  \item $\wfTyp{\G}{T}$ implies $\wfTyp{\Delta}{T}$;
  \item $\wfTyp{\G}{D}$ implies $\wfTyp{\Delta}{D}$;
  \item $\typ{\G}{t}{T}$ implies $\typ{\Delta}{t}{T}$;
  \item $\sub{\G}{T}{U}$ implies $\sub{\Delta}{T}{U}$;
  \item $\sub{\G}{C_1}{C_2}$ implies $\sub{\Delta}{C_1}{C_2}$;
  \item $\isrdr{\G}{x}$ implies $\isrdr{\G}{x}$;
  \item $\ninter{\G}{C_1}{C_2}$ implies $\ninter{\Delta}{C_1}{C_2}$.
  \end{enumerate}
\end{lemma}

\begin{proof}
  By straightforward induction on the derivations.
  No rule depends on the order of the bindings.
\end{proof}

\begin{lemma}[Weakening]
  Given $\G, \Delta$,
  \begin{enumerate}[(i)]
  \item $\wfTyp{\G}{T}$ implies $\wfTyp{\G, \Delta}{T}$;
  \item $\wfTyp{\G}{D}$ implies $\wfTyp{\G, \Delta}{D}$;
  \item $\typ{\G}{t}{T}$ implies $\typ{\G, \Delta}{t}{T}$;
  \item $\sub{\G}{T}{U}$ implies $\sub{\G, \Delta}{T}{U}$;
  \item $\sub{\G}{C_1}{C_2}$ implies $\sub{\G, \Delta}{C_1}{C_2}$;
  \item $\isrdr{\G}{x}$ implies $\isrdr{\G, \Delta}{x}$;
  \item $\ninter{\G}{C_1}{C_2}$ implies $\ninter{\G, \Delta}{C_1}{C_2}$.
  \end{enumerate}
\end{lemma}

\begin{proof}
  As usual, the rules only check if a variable is bound in the environment
  and all versions of the lemma are provable by straightforward induction.
  For rules which extend the environment, such as \ruleref{abs}, we need permutation.
  All cases are analogous, so we will illustrate only one.

  \emph{Case \ruleref{abs}}.
  In this case,
  $t = \lambda(z :_D U) s$,
  $T = \forall(z :_D U) T^\prime$,
  and $\typ{\G, z :_D U}{s}{T^\prime}$.
  By IH we have
  $\typ{\G, z :_D U, \Delta}{s}{T^\prime}$.
  Since $D$ and $U$ cannot mention variables in $\Delta$,
  $\G, \Delta, z :_D U$ is still well-formed.
  By permutation we have
  $\typ{\G, \Delta, z :_D U}{s}{T^\prime}$.
  This case is therefore concluded by \ruleref{abs}.
\end{proof}

\begin{lemma}[Bound narrowing]
  \label{lemma:bound-narrowing}
  Given an environment $\G, X <: S, \Delta$
  and the fact that $\sub{\G}{S^\prime}{S}$,
  the followings hold:
  \begin{enumerate}
  \item $\wfTyp{\G, X <: S, \Delta}{T}$ implies $\wfTyp{\G, X <: S^\prime, \Delta}{T}$;
  \item $\wfTyp{\G, X <: S, \Delta}{D}$ implies $\wfTyp{\G, X <: S^\prime, \Delta}{D}$;
  \item $\sub{\G, X <: S, \Delta}{C_1}{C_2}$ implies $\sub{\G, X <: S^\prime, \Delta}{C_1}{C_2}$;
  \item $\sub{\G, X <: S, \Delta}{T}{U}$ implies $\sub{\G, X <: S^\prime, \Delta}{T}{U}$;
  \item $\isrdr{\G, X <: S, \Delta}{x}$ implies $\isrdr{\G, X <: S^\prime, \Delta}{x}$;
  \item $\ninter{\G, X <: S, \Delta}{C_1}{C_2}$ implies $\ninter{\G, X <: S^\prime, \Delta}{C_1}{C_2}$;
  \item $\typ{\G, X <: S, \Delta}{t}{T}$ implies $\typ{\G, X <: S^\prime, \Delta}{t}{T}$.
  \end{enumerate}
\end{lemma}

\begin{proof}
	By straightforward induction on the derivations.
  The only two rules that make use of type variable bounds are \ruleref{tvar} in subtyping,
  and \ref{rd-widen} in \textsf{is-reader}.
  In all other cases, we can conclude it by utilizing the IH, other narrowing lemmas and the same rule.
  We will now give the proof of the \ruleref{tvar} case.

  \emph{Case \ruleref{tvar}}.
  In this case, $T_1 = Y$, $T_2 = R$,
  and $Y <: R \in \G, X <: S, \Delta$.
  If $Y \ne X$,
  then $Y <: R \in \G, X <: S^\prime, \Delta$.
  Applying the \ruleref{tvar} rule allows us to derive the conclusion.
  Otherwise, $X = Y$.
  The goal becomes $\sub{\G, X <: S^\prime, \Delta}{X}{S}$.
  By weakening we can show that
  $\sub{\G, X <: S^\prime, \Delta}{S^\prime}{S}$.
  We can therefore conclude this case by \ruleref{tvar} and \ruleref{trans}.
\end{proof}

\begin{lemma}[Type narrowing]
  \label{lemma:type-narrowing}
  Given
  $\G, x : P, \Delta$,
  and that $\sub{\G}{P^\prime}{P}$,
  the following propositions hold:
  \begin{enumerate}
  \item $\wfTyp{\G, x : P, \Delta}{T}$ implies $\wfTyp{\G, x : P^\prime, \Delta}{T}$;
  \item $\wfTyp{\G, x : P, \Delta}{D}$ implies $\wfTyp{\G, x : P^\prime, \Delta}{D}$;
  \item $\sub{\G, x : P, \Delta}{C_1}{C_2}$ implies $\sub{\G, x : P^\prime, \Delta}{C_1}{C_2}$;
  \item $\sub{\G, x : P, \Delta}{T_1}{T_2}$ implies $\sub{\G, x : P^\prime, \Delta}{T_1}{T_2}$;
  \item $\isrdr{\G, x : P, \Delta}{y}$ implies $\isrdr{\G, x : P', \Delta}{y}$;
  \item $\ninter{\G, x : P, \Delta}{C_1}{C_2}$ implies $\ninter{\G, x : P^\prime, \Delta}{C_1}{C_2}$;
  \item $\typ{\G, x : P, \Delta}{t}{T}$ implies $\sub{\G, x : P^\prime, \Delta}{t}{T}$.
  \end{enumerate}
\end{lemma}

\begin{proof}
  The proof follows through a straightforward induction on the derivations,
  wherein only the cases enumerated below are contigent to the variable bindings in the context.
  The other cases can be deduced by employing the IH, the same rule and other narrowing lemmas.

  \emph{Case \ruleref{var}}.
  In this case, $t = y$,
  $y : \tCap{C}{S} \in \G, x : P, \Delta$,
  and $T = \set{x}\,S$.
  If $x \ne y$, then the binding for $y$ is remains unaffected
  and we can directly conclude by applying the \ruleref{var} rule.
  Otherwise, we have $x = y$,
  implying that $P = \tCap{C}{S}$.
  We can demonstrate that $P^\prime = \tCap{C^\prime}{S^\prime}$ for some $C^\prime, S^\prime$.
  Inspecting the judgment $\sub{\G}{P^\prime}{P}$
  we can show that $\sub{\G}{S^\prime}{S}$.
  Consequently, We can establish that $\sub{\G, x : S^\prime, \Delta}{\set{x}\,S^\prime}{\set{x}\,S}$
  by applying the weakening lemma and the \ruleref{capt} rule.
  Finally, we can conclude by utilizing the \ruleref{var} and the \ruleref{sub} rules.

  \emph{Case \ruleref{sc-var}}.
  In this case, $C_1 = \set{y}$,
  $y : \tCap{C}{S} \in \G, x : P, \Delta$,
  and $\sub{\G, x : P, \Delta}{C}{C_2}$.
  Using the IH, we can demonstrate that $\sub{\G, x : P^\prime, \Delta}{C}{C_2}$.
  If $x \ne y$, the binding for $y$ remains unaffected
  and we can immediately apply the IH and the \ruleref{sc-var} rule to conclude.
  Otherwise, we have $x = y$,
  implying that $P = \tCap{C}{S}$.
  Inpsecting the subtype judgment $\sub{\G}{P^\prime}{P}$,
  we can show that $\sub{\G}{C^\prime}{C}$
  where $C^\prime$ denotes $\cs{P^\prime}$.
  Invoking the weakening lemma and the subcapturing transitivity lemma,
  we can derive that $\sub{\G, x : P^\prime, \Delta}{C^\prime}{C_2}$.
  Finally, we conclude by applying the \ruleref{sc-var} rule.


  \emph{Case \ruleref{ni-var}}.
  In this case, $C_1 = \set{y}$, $C_2 = \set{z}$
  $y : \tCap{C}{S} \in \G, x : P, \Delta$
  and $\ninter{\G, x : P, \Delta}{C}{z}$.
  Using IH, we can demonstrate that $\ninter{\G, x : P^\prime, \Delta}{C}{z}$.
  If $y \ne x$, then the binding for $y$ stays unaffected
  and we can immediately invoke the IH and the \ruleref{ni-var} rule to conclude.
  Otherwise we have $y = x$, implying that $P = \tCap{C}{S}$.
  From the subtype judgment $\sub{\G}{P^\prime}{P}$
  we can show that $\sub{\G}{C^\prime}{C}$
  where $C^\prime$ is $\cs{P^\prime}$.
  We invoke the weakning lemma and Lemma \ref{lemma:subcapt-pres-ninter} to prove that
  $\ninter{\G, x : P^\prime, \Delta}{C^\prime}{z}$.
  Finally, we can apply the \ruleref{ni-var} rule to conclude this case.
\end{proof}

\begin{lemma}[Separation degree expansion]
  \label{lemma:sepdegree-expansion}
  Given $\G, x :_D T, \Delta$,
  and that $D \subseteq D^\prime$,
  the following propositions hold:
  \begin{itemize}
  \item $\wfTyp{\G, x :_D T, \Delta}{T}$ implies $\wfTyp{\G, x :_{D^\prime} T, \Delta}{T}$;
  \item $\wfTyp{\G, x :_D T, \Delta}{D}$ implies $\wfTyp{\G, x :_{D^\prime} T, \Delta}{D}$;
  \item $\sub{\G, x :_D T, \Delta}{C_1}{C_2}$ implies $\sub{\G, x :_{D^\prime} T, \Delta}{C_1}{C_2}$;
  \item $\sub{\G, x :_D T, \Delta}{T_1}{T_2}$ implies $\sub{\G, x :_{D^\prime} T, \Delta}{T_1}{T_2}$;
  \item $\isrdr{\G, x :_D T, \Delta}{y}$ implies $\isrdr{\G, x :_{D^\prime} T, \Delta}{y}$;
  \item $\ninter{\G, x :_D T, \Delta}{C_1}{C_2}$ implies $\ninter{\G, x :_{D^\prime} T, \Delta}{C_1}{C_2}$;
  \item $\typ{\G, x :_D T, \Delta}{t}{U}$ implies $\typ{\G, x :_{D^\prime} T, \Delta}{t}{U}$.
  \end{itemize}
\end{lemma}

\begin{proof}
  The proof is carried out through straightforward induction on the derivations.
  Only the \ruleref{ni-degree} case depends on the separation degrees.
  Other cases can be concluded by IH, other expansion lemmas and the same rule.

  \emph{Case \ruleref{ni-degree}}.
  Then $C_1 = \set{z_1}$, $C_2 = \set{z_2}$,
  $z_1 :_{D_1} T_1 \in \G, x :_D T, \Delta$,
  and $z_2 \in D_1$.
  If $z_1 = x$, then $D_1 = D$.
  Considering $D \subseteq D^\prime$,
  we can show that $z_2 \in D^\prime$.
  Finally, we invoke \ruleref{ni-var} and conclude.
\end{proof}

\subsubsection{Properties of Evaluation Configurations}

\begin{lemma}[Value typing (I)]
  \label{lemma:val-typing}
	If $\typDft{v}{T}$
  then $T$ is not in the form of $X\capt C$.
\end{lemma}

\begin{proof}
	By straightforward induction on the derivation.

  \emph{Case \ruleref{sub}}.
  Then $\typDft{v}{T^\prime}$,
  and $\subDft{T^\prime}{T}$.
  By IH we know that $T^\prime$ is not in the form of $X\capt C$.
  If $T$ is in the form of $X\capt C$,
  then we invoke Lemma \ref{lemma:subtype-inversion-tvar}
  and derive a contradiction.
  We therefore conclude.

  \emph{Other cases}.
  Other cases for typing values are immediate.
\end{proof}

\begin{lemma}[Value typing (II)]
  \label{lemma:val-typing-alt}
	If $\typDft{v}{T}$
  then $T$ is not in the form of $\tRef{S}\capt C$.
\end{lemma}

\begin{proof}
	By straightforward induction on the typing derivation.
  No typing rule for values results in a $\tRef{S}\capt C$ type.
  In the \ruleref{sub} case,
  we have $\typDft{v}{T^\prime}$ and $\subDft{T^\prime}{T}$.
  By IH we have $T^\prime$ is not in the form of $\tRef{S}\capt C$.
  By induction on the $\subDft{T^\prime}{T}$ we can show that
  $T$ is also not in this form.
  We can therefore conclude this case.
\end{proof}

\begin{lemma}[Store lookup: mutable variables]
  \label{lemma:store-lookup-mut}
	If
  (i) $\match{}{\sta}{\G}$,
  (ii) $\typ{\G}{x}{\tRef{S}\capt C}$,
  then
  $\exists v. \sta(\tVar{x}) = v$.
\end{lemma}

\begin{proof}
  By induction on $\match{}{\sta}{\G}$.

  \emph{Case \ruleref{st-empty}}.
  Then $\sta = \G = \varemptyset$.
  It is contradictory to have $\typ{\varemptyset}{x}{\tRef{S}\capt C}$.

  \emph{Case \ruleref{st-val}}.
  Then $\sta = \sta_0, \tValM{D}{y} \mapsto v$
  $\G = \G_0, y :_D \cv{v}\, S^\prime$,
  and $\typ{\G_0}{v}{\cv{v}\, S^\prime}$.
  \begin{itemize}
  \item If $x = y$,
    by induction on the typing judgment and applying weakening,
    we can show that
    $\sub{\G}{S^\prime\capt \cv{v}}{\tRef{S}\capt C}$.
    By Lemma \ref{lemma:subtype-inversion-mut} and Lemma \ref{lemma:val-typing},
    we can show that $\sub{\G}{\cv{v}}{C}$
    and $S^\prime = \tRef{S}$.
    By Lemma \ref{lemma:val-typing-alt} we can derive the contradiction.

  \item If $x \ne y$,
    then we can show that $x \in \dom{\G_0}$,
    and therefore $\typ{\G_0}{x}{C\,\tRef{S}}$.
    We conclude this case by IH.
  \end{itemize}

  \emph{Case \ruleref{st-var}}.
  Then $\sta = \sta_0, \tVar{y} = v$,
  and $\G = \G_0, y :_{\dom{\G_0}} \tCap{\set{\univ}}{\tRef{S^\prime}}$.
  If $y = x$, we conclude this case by
  the fact that $\sta(\tVar{x}) = v$.
  Otherwise, we have $\typ{\G_0}{x}{C\,\tRef{S}}$
  and conclude again by IH.

  \emph{Case \ruleref{st-set}}.
  Analogous to the previous case.
\end{proof}

\begin{lemma}[Store lookup: pure values]
  \label{lemma:store-lookup-pure-val}
	If
  (i) $\match{}{\sta}{\G}$,
  (ii) $\typ{\G}{x}{S}$,
  then
  $\exists v. \sta(\tVal{x}) = v$.
\end{lemma}

\begin{proof}
	By induction on the derivation of $\match{}{\sta}{\G}$.

  \emph{Case \ruleref{st-empty}}.
  Contradictory.

  \emph{Case \ruleref{st-val}}.
  Then $\sta = \sta_0, \tVal{y} \mapsto v$,
  $\match{}{\sta_0}{\G_0}$,
  $\typ{\G_0}{v}{T}$,
  and $\G = \G_0, y :_{\set{}} T$.
  If $x \ne y$ we conclude by IH.
  Otherwise we have $v$ such that $\sta(\tVal{x}) = v$.

  \emph{Case \ruleref{st-var}}.
  Then $\sta = \sta_0, \tVar{y} = v$,
  $\match{}{\sta_0}{\G_0}$,
  $\typ{\G_0}{S^\prime}$,
  and $\G = \G_0, y :_{\dom{\G_0}} \tCap{\set{\univ}}{\tRef{S^\prime}}$.
  If $x = y$, we can show that $\sub{\G}{\set{\univ}}{\varemptyset}$
  which is contradictory.
  Otherwise we conclude the case by IH.

  \emph{Case \ruleref{st-set}}.
  By IH.
\end{proof}

\begin{lemma}[Store lookup: term abstractions]
  \label{lemma:store-lookup-fun}
	If
  (i) $\match{}{\sta}{\G}$,
  (ii) $\typ{\G}{x}{\forall(z :_D U) T\capt C}$,
  then
  $\exists v. \sta(\tVal{x}) = v$.
\end{lemma}

\begin{proof}
	By induction on the $\match{}{\sta}{\G}$ derivation.

  \emph{Case \ruleref{st-empty}}.
  Contradictory.

  \emph{Case \ruleref{st-val}}.
  Then $\sta = \sta_0, \tVal{y} \mapsto v$,
  $\match{}{\sta_0}{\G_0}$,
  $\typ{\G_0}{v}{T}$,
  and $\G = \G_0, y :_D T$.
  If $x = y$, then we conclude immediately.
  Otherwise we conclude the goal by IH.

  \emph{Case \ruleref{st-var}}.
  Then $\sta = \sta_0, \tVar{y} = v$,
  $\match{}{\sta_0}{\G_0}$,
  $\typ{\G_0}{v}{S^\prime}$,
  and $\G = \G_0, y :_{\dom{\G_0}} \tCap{\set{\univ}}{\tRef{S^\prime}}$.
  If $x = y$,
  we can show that $\subDft{\tCap{\set{\univ}}{\tRef{S^\prime}}}{\tCap{C}{\forall(z :_D T) U}}$.
  By Lemma \ref{lemma:subtype-inversion-fun},
  we can derive a contradiction.
  Otherwise, $x \ne y$.
  We conclude this case by IH.

  \emph{Case \ruleref{st-set}}.
  As above.
\end{proof}

\begin{lemma}[Store lookup: type abstractions]
  \label{lemma:store-lookup-tfun}
	If
  (i) $\match{}{\sta}{\G}$,
  (ii) $\typ{\G}{x}{\forall[X <: S] T\capt C}$,
  then
  $\exists v. \sta(\tVal{x}) = v$.
\end{lemma}

\begin{proof}
	Analogous to the proof of Lemma \ref{lemma:store-lookup-fun}.
\end{proof}

\begin{lemma}[Store lookup: reader]
  \label{lemma:store-lookup-reader}
  If
  (i) $\match{}{\sta}{\G}$,
  (ii) $\typ{\G}{x}{\tRdr{S}\capt C}$,
  then
  $\exists v. \sta(\tVal{x}) = v$.
\end{lemma}

\begin{proof}
  Analogous to the proof of Lemma \ref{lemma:store-lookup-fun}.
\end{proof}

\begin{lemma}[Inversion of evaluation context typing]
  \label{lemma:evctx-typing-inversion}
	If $\typ{\G}{\evctx{s}}{T}$,
  then $\exists \Delta, U$ such that
  (i) $\typec{\G}{e}{\Delta}{U}{s}{T}$,
  and (ii) $\typ{\G, \Delta}{s}{U}$.
\end{lemma}

\begin{proof}
	By induction on $e$.

  \emph{Case $e = []$}.
  Set $\Delta = \varemptyset$ and $U = T$.
  This case can be concluded immediately.

  \emph{Case $e = \tLetMode{m}{x}{e^\prime}{u}$}.
  Then $\typ{\G}{\tLetMode{m}{x}{\evpctx{s}}{u}}{T}$.
  By induction on this typing derivation.
  The sub-goal is that $\typ{\G}{\evpctx{s}}{T_0}$ for some $T_0$,
  $\typ{\G, x :_{\set{}} T_0}{u}{T}$,
  $x \notin \fv{T}$,
  and $\ninter{\G}{\evpctx{s}}{u}$ if $m = \parmode$.
  \begin{itemize}
  \item \emph{Case \ruleref{let}}.
    Conclude immediately from the premises.

  \item \emph{Case \ruleref{sub}}.
    In this case, $\typ{\G}{\tLetMode{\parmode}{x}{\evpctx{s}}{u}}{T^\prime}$ for some $T^\prime$,
    and $\sub{\G}{T^\prime}{T}$.
    We conclude this case by IH and \ruleref{sub}.

  \item \emph{Other cases.}
    Not applicable.
  \end{itemize}
  Now we invoke IH on $\typ{\G}{\evpctx{s}}{T_0}$
  and show that
  $\exists \Delta_0, U_0$ such that
  $\typec{\G}{e^\prime}{\Delta_0}{U_0}{s}{T_0}$.
  Now we set $\Delta = \Delta_0$ and $U = U_0$.
  First, we have $\match{\G}{e}{\Delta}$ by \ruleref{ev-let-1}.
  Given $s^\prime$ such that
  $\typ{\G, \Delta}{s^\prime}{U}$
  and $\subi{\G; \Delta}{s^\prime}{s}$,
  we have $\typ{\G}{\evpctx{s^\prime}}{T_0}$ by IH.
  Also, since we have $\match{\G}{e^\prime}{\Delta}$,
  from $\subi{\G; \Delta}{s^\prime}{s}$
  we can show that $\ninter{\G}{\evpctx{s^\prime}}{u}$.
  We can conclude this case by \ruleref{let}.

  \emph{Case $e = \tLetMode{\parmode}{x}{t}{e^\prime}$}.
  Then $\typ{\G}{\tLetMode{\parmode}{x}{t}{\evpctx{s}}}{T}$.
  By a similar induction on this typing derivation, we can show that
  $\typ{\G}{t}{T_0}$ for some $T_0$,
  $\typ{\G, x :_{\set{}} T_0}{\evpctx{s}}{T}$,
  $x \notin \fv{T}$,
  and $\ninter{\G}{t}{\evpctx{s}}$.
  Now we invoke IH to show that
  $\exists \Delta_0, U_0$ such that
  $\typec{\G, x :_{\set{}} T_0}{e^\prime}{\Delta_0}{U_0}{s}{T}$,
  and $\typ{\G, x :_{\set{}} T_0, \Delta_0}{s}{U_0}$.
  Set $\Delta = x :_{\set{}} T_0, \Delta_0$ and $U = U_0$.
  We have $\match{\G}{e}{x :_{\set{}} T_0, \Delta_0}$ by \ruleref{ev-let-2}.
  $\forall s^\prime$ such that
  $\typ{\G, \Delta}{s^\prime}{U}$
  and $\subi{\G; \Delta}{s^\prime}{s}$.
  the goal is to show that
  $\typ{\G}{\tLetMode{\parmode}{x}{t}{\evpctx{s^\prime}}}{T}$.
  By Lemma \ref{lemma:subi-shift-boundary},
  we have $\subi{\G, x :_{\set{}} T_0; \Delta_0}{s^\prime}{s}$.
  By IH we can show that
  $\typ{\G, x :_{\set{}} T_0}{\evpctx{s^\prime}}{T}$.
  From $\subi{\G; \Delta}{s^\prime}{s}$
  we can show that
  $\ninter{\G, x :_{\set{}} T_0}{\cv{t}}{\cv{\evpctx{s^\prime}}}$.
  By strengthening, we show that
  $\ninter{\G}{\cv{t}}{\cv{\evpctx{s^\prime}}\setminus \set{x}}$,
  which is the same as
  $\ninter{\G}{t}{\evpctx{s^\prime}}$.
  This case is therefore concluded by \ruleref{let}.
\end{proof}

\begin{lemma}[Weakening of evaluation context inversion]
  \label{lemma:evctx-typing-weakening}
  $\typec{\G}{e}{\Delta}{U}{s}{T}$
  implies
  $\typec{\G, \Delta}{e}{\Delta}{U}{s}{T}$.
\end{lemma}

\begin{proof}
	By weakening of environment matching, typing and separation.
\end{proof}

\begin{lemma}[Downgrading separation degree preserves subcapturing]
  \label{lemma:downgrading-sepdegree-pres-subcapt}
	Given $\G = \G_1, x :_{\dom{\G_1}} \tCap{\set{\univ{}}}{S}, \G_2$,
  and $y :_{D} T \in \G_1$ where $\univ{} \notin \cs{T}$,
  $\sub{\G}{T_1}{T_2}$
  implies $\sub{\G^\prime}{T_1}{T_2}$
  where $\G^\prime = \G_1, x :_{\dom{\G_1}/y} \tCap{\set{\univ{}}}{S}, \G_2$.
\end{lemma}

\begin{proof}
	By straightfoward induction on the derivation,
  wherein no rule makes use of the separation degrees in the context.
\end{proof}

\begin{lemma}[Downgrading separation degree preserves separation]
  \label{lemma:downgrading-sepdegree-pres-ninter}
	Given $\G = \G_1, x :_{\dom{\G_1}} \tCap{\set{\univ{}}}{S}, \G_2$,
  and $y :_{D} T \in \G_1$ where $\univ{} \notin \cs{T}$,
  $\ninter{\G}{C_1}{C_2}$
  implies $\ninter{\G^\prime}{C_1}{C_2}$
  where $\G^\prime = \G_1, x :_{\dom{\G_1}/y} \tCap{\set{\univ{}}}{S}, \G_2$.
\end{lemma}

\begin{proof}
  By induction on the separation derivation.

  \emph{Case \ruleref{ni-symm}}. By IH and \ruleref{ni-symm} again.

  \emph{Case \ruleref{ni-set}}.
  Then $C_1 = \set{x_i}_{i=1,\cdots,n}$,
  and $\overline{\ninter{\G}{x_i}{C_2}}^{i=1,\cdots,n}$.
  By repeated IH, we have
  $\overline{\ninter{\G^\prime}{x_i}{C_2}}^{i=1,\cdots,n}$.
  This case is therefore concluded by \ruleref{ni-set}.

  \emph{Case \ruleref{ni-degree}}.
  Then $C_1 = \set{z_1}$,
  $C_2 = \set{z_2}$,
  $z_1 :_{D_1} T_1 \in \G$,
  and $z_2 \in D_1$.
  If $z_1 = x$ and $z_2 = y$,
  Since by the well-formedness of the environment and that $\univ{} \notin \cs{T}$,
  we have $\cs{T} \subseteq \dom{\G}$.
  We can therefore show that $\ninter{\G}{x}{\cs{T}}$
  by \ruleref{ni-var} and \ruleref{ni-set}.
  Otherwise the goal follows directly from the preconditions.

  \emph{Case \ruleref{ni-var}}.
  By applying the IH and the same rule.

  \emph{Case \ruleref{ni-reader}}.
  By applying Lemma \ref{lemma:downgrading-sepdegree-pres-subcapt} and the same rule.
\end{proof}

\begin{lemma}[Downgrading separation degree preserves subtyping]
  \label{lemma:downgrading-sepdegree-pres-sub}
	Given $\G = \G_1, x :_{\dom{\G_1}} \tCap{\set{\univ{}}}{S}, \G_2$,
  and $y :_{D} T \in \G_1$ where $\univ{}\notin \cs{T}$,
  $\sub{\G}{T}{U}$
  implies $\sub{\G^\prime}{T}{U}$
  where $\G^\prime = \G_1, x :_{\dom{\G_1}/y} \tCap{\set{\univ{}}}{S}, \G_2$.
\end{lemma}

\begin{proof}
	By straightforward induction on the subtyping derivation.
  No rule makes use of the separation degree on the bindings.
\end{proof}

\begin{lemma}[Downgrading separation degree preserves typing]
  \label{lemma:downgrading-sepdegree-pres-typing}
	Given $\G = \G_1, x :_{\dom{\G_1}} \tCap{\set{\univ{}}}{S}, \G_2$,
  and $y :_{D} T \in \G_1$ where $\univ{}\notin \cs{T}$,
  $\typ{\G}{t}{T}$
  implies $\typ{\G^\prime}{t}{T}$
  where $\G^\prime = \G_1, x :_{\dom{\G_1}/y} \tCap{\set{\univ{}}}{S}, \G_2$.
\end{lemma}

\begin{proof}
	By induction on the typing derivation.

  \emph{Case \ruleref{var}}.
  By the precondition and the same rule.

  \emph{Case \ruleref{sub}}.
  By IH, Lemma \ref{lemma:downgrading-sepdegree-pres-sub}, and the same rule.

  \emph{Case \ruleref{app} and \ruleref{let}}.
  By IH, Lemma \ref{lemma:downgrading-sepdegree-pres-ninter} and the same rule.

  \emph{Other cases}.
  By IH and the same rule.
\end{proof}

\begin{lemma}[Evaluation context inversion and reification for mutable variables]
  \label{lemma:evctx-typing-inversion-mut}
  $\typ{\G}{\evctx{\tLetVar{x}{y}{s}}}{T}$
  implies
  $\exists \Delta, S, U$ such that
  (i) $\match{\G}{e}{\Delta}$,
  (ii) $\typ{\G, \Delta}{y}{S}$,
  (iii) $\typ{\G, x :_{\dom{\G}} \tRef{S}\capt \set{\univ{}}, \Delta}{s}{U}$,
  and (iv) $\typ{\G, x :_{\dom{\G}} \tRef{S}\capt \set{\univ{}}}{\evctx{s}}{T}$.
\end{lemma}

\begin{proof}
  By induction on $e$.

  \emph{Case $e = []$}.
  Then ${\G}\vdash \tLetVar{x}{y}{s} : {T}$.
  By induction on the typing derivation we can show that
  $\typ{\G}{y}{S}$ for some $S$,
  and $\typ{\G, x :_D \tRef{S}\capt\set{\univ{}}}{s}{T}$ for some $D$.
  We set $\Delta = \varemptyset$ and $U = T$.
  We can conclude this case by separation degree expansion.

  \emph{Case $e = \tLetMode{m}{z}{e^\prime}{u}$}.
  Then ${\G}\vdash\tLetMode{m}{z}{e'[\tLetVar{x}{y}{s}]}{u}: {T}$.
  By inspecting this derivation,
  we can show that
  $\typ{\G}{\evpctx{\tLetVar{x}{y}{s}}}{T_0}$ for some $T_0$,
  $\typ{\G, z :_{\set{}} T_0}{u}{T}$,
  and $\ninter{\G}{\evpctx{\tLetVar{x}{y}{s}}}{u}$.
  By IH, we show that $\exists \Delta_0, S_0, U_0$ such that
  $\match{\G}{e'}{\Delta_0}$,
  $\typ{\G, \Delta_0}{y}{S_0}$,
  $\typ{\G, x :_{\dom{\G}} \tRef{S}\capt\set{\univ{}}, \Delta_0}{s}{U_0}$,
  and $\typ{\G, x :_{\dom{\G}} \tRef{S}\capt\set{\univ{}}}{\evctx{s}}{T}$.
  Note that $\match{\G}{e}{\Delta_0}$.
  By weakening, we can show that
  $\ninter{\G, x :_{\dom{\G}}}{D^\prime}{\cs{T_0}}$.
  We can show that
  $\cv{\evpctx{s}} \subseteq \cv{\evpctx{\tLetVar x y s}}\setminus \set{y} \cup \set{x}$.
  Now, we can show that
  $\ninter{\G, x :_{\dom{\G}} \tRef{S}\capt\set{\univ{}}}{\cv{\evpctx{\tLetVar x y s}}\setminus \set{y}}{\cv{u}\setminus \set{z}}$
  from the precondition and weakening.
  Then, we can show that
  $\G, x :_{\dom{\G}} \tRef{S}\capt\set{\univ{}}\vdash {x}\bowtie \cv{u}\setminus \set{z}$
  by \ruleref{ni-degree} and \ruleref{ni-set}.
  We can therefore show that
  $\ninter{\G, x :_{\dom{\G}} \tRef{S}\capt\set{\univ{}}}{\evpctx{s}}{u}$.
  This case is therefore concluded by setting $\Delta = \Delta_0$ and $U = U_0$,
  and using the \ruleref{let} rule.

  \emph{Case $e = \tLetMode{\parmode}{z}{t}{e^\prime}$}.
  Then ${\G}\vdash\tLetMode{\parmode}{z}{t} e'[\tLetVar{x}{y}{s}]: {T}$.
  By inspecting the typing derivation we can show that
  $\typ{\G}{t}{T_0}$ for some $T_0$,
  $\typ{\G, z :_{\set{}} T_0}{\evpctx{\tLetVar{x}{y}{s}}}{T}$,
  and $\ninter{\G}{t}{\evpctx{\tLetVar{x}{y}{s}}}$.
  Now we invoke IH on the typing derivation $\typ{\G, z :_{\set{}} T_0}{\evpctx{\tLetVar{x}{y}{s}}}{T}$
  to show that $\exists \Delta_0, S_0, U_0$ such that
  $\match{\G, z :_{\set{}} T_0}{e^\prime}{\Delta_0}$,
  $\typ{\G, z :_{\set{}} T_0, \Delta_0}{y}{S_0}$,
  $\typ{\G, z :_{\set{}} T_0, x :_{\dom{\G} \cup \set{z}} \tRef{S_0}\capt\set{\univ{}}, \Delta_0}{s}{U_0}$
  and $\typ{\G, x :_{\dom{\G}} \tRef{S_0}\capt\set{\univ{}}, z :_{\set{}} T_0}{\evpctx{s}}{T}$.
  Set $\Delta = z :_{\set{}} T_0, \Delta_0$.
  We can first show that $\match{\G}{e}{z :_{\set{}} T_0, \Delta_0}$.
  Also, we have $\typ{\G, \Delta}{y}{S_0}$.
  Then we have to show that
  $\typ{\G, x :_{\dom{\G}} \tRef{S_0}\capt\set{\univ{}}, \Delta}{s}{U_0}$,
  which requires dropping $z$ from the separation degree of $x$.
  Now we invoke Lemma \ref{lemma:downgrading-sepdegree-pres-typing} to show that
  $\typ{\G, z :_{\set{}} T_0, x :_{\dom{\G}} \tRef{S_0}\capt\set{\univ{}}, \Delta_0}{s}{U_0}$.
  Then by permutation we have
  $\typ{\G, x :_{\dom{\G}} \tRef{S_0}\capt\set{\univ{}}, \Delta}{s}{U_0}$.
  We can similarly show that
  $\typ{\G, x :_{\dom{\G}} \tRef{S_0}\capt\set{\univ{}}, z :_{\set{}} T_0}{\evpctx{s}}{T}$.
  Similar to the previous case,
  we can show that $\ninter{\G, x :_{\dom{\G}} \tRef{S}\capt\set{\univ{}}}{t}{\evpctx{s}}$.
  We can therefore conclude this case by \ruleref{let}.
\end{proof}

\subsubsection{Substitution}

\paragraph{Term Substitution}

\begin{lemma}[Term substitution preserves readers]
	\label{lemma:term-subst-pres-rdr-checking}
  If
  (i) $\isrdrtype{\G, x : P, \Delta}{T}$,
  and (ii) $\typ{\G}{x'}{P}$,
  then
  $\isrdr{\G, \theta\Delta}{\theta T}$,
  where $\theta = \fSubst{x}{x'}$.
\end{lemma}

\begin{proof}
	By induction on the first derivation.

  \emph{Case \ruleref{rd-reader}}.
  We conclude by the same rule.

  \emph{Case \ruleref{rd-tvar}}.
  Then $\mathcal{X} = C\,X$,
  $X <: S \in \G$ for some $S$,
  and $\isrdr{\G, x :_D P, \Delta}{C\,S}$.
  Now we inspect where $X <: S$ is bound,
  wherein in each case we can show that $X <: \theta S \in \theta \G, \theta\Delta$.
  Now we conclude by using the IH and the same rule.


\end{proof}

\begin{lemma}[Term substitution preserves \textsf{is-reader}]
	\label{lemma:term-subst-pres-isrdr}
  If
  (i) $\isrdr{\G, x : P, \Delta}{z}$,
  and (ii) $\typ{\G}{x'}{P}$,
  then
  $\isrdr{\G, \theta\Delta}{\theta z}$,
  where $\theta = \fSubst{x}{x'}$.
\end{lemma}

\begin{proof}
  We have $z : T \in \G, x : P, \Delta$ for some $T$
  and $\sub{\G, x : P, \Delta}{T}{\tRdr{S}\capt C}$ for some $C, S$.
  We first invoke Lemma \ref{lemma:is-rdr-eqv} to show that
  $\isrdrtype{\G, x : P, \Delta}{T}$.
  Using Lemma \ref{lemma:term-subst-pres-rdr-checking} we can show that
  $\isrdr{\G, \theta\Delta}{\theta T}$.
  Then we perform a case analysis on where $z$ is bound.
  If $z = x$ then we have $\theta z = x'$ and $\isrdrtype{\G, \theta\Delta}{\theta P}$.
  By induction on the derivation of $\typ{\G}{x'}{P}$ we can show that
  there exists $S'\capt C'$ such that $x' : S'\capt C' \in \G$ and $\sub{\G}{S'\capt\set{x'}}{P}$.
  By the well-formedness of the environment,
  we can show that $\theta P = P$, implying that $\isrdrtype{\G, \theta\Delta}{P}$.
  Now, by weakening and Lemma \ref{lemma:subtype-pres-rdr-checking},
  we can show that $\isrdrtype{\G, \theta\Delta}{S'\capt\set{x'}}$.
  Finally, we use Lemma \ref{lemma:captset-irrelevant-rdr-checking} to show that
  $\isrdrtype{\G, \theta\Delta}{S'\capt C'}$,
  and then use Lemma \ref{lemma:is-rdr-eqv} again to conclude this case.
  If $z$ is bound in $\G$ or $\Delta$,
  in both cases we can show that $z : \theta T \in \G, \theta\Delta$
  and therefore conclude directly.
\end{proof}

\begin{lemma}[Term substitution preserves subcapturing]
  \label{lemma:term-subst-pres-subcapt}
	If
  (i) $\sub{\G, x :_D P, \Delta}{C_1}{C_2}$,
  and (ii) $\typ{\G}{x^\prime}{P}$
  then
  $\sub{\G, \theta\Delta}{\theta C_1}{\theta C_2}$
  where $\theta = \fSubst{x}{x^\prime}$.
\end{lemma}

\begin{proof}
	By induction on the subcapture derivation.

  \emph{Case \ruleref{sc-trans}}.
  By applying the IH twice and use the same rule.

  \emph{Case \ruleref{sc-elem}}.
  Then $C_1 = \set{y}$,
  and $y \in C_2$.
  If $y \ne x$,
  we can show that $x \in \theta C_2$,
  and therefore conclude the case by \ruleref{sc-elem} again.
  Otherwise, if $y = x$,
  we can show that $x' \in \theta C_2$.
  This case is therefore conclude by \ruleref{sc-elem}.

  \emph{Case \ruleref{sc-set}}.
  Then $\overline{\sub{\G, x :_D P, \Delta}{\set{x}}{C_2}}^{x \in C_1}$.
  By applying the IH repeatedly,
  we can show that
  $\overline{\sub{\G, \theta \Delta}{\theta \set{x}}{\theta C_2}}^{x \in C_1}$.
  Note that
  $\theta C_1 = \bigcup_{x \in C_1} \theta \set{x}$.
  This case is therefore concluded by \ruleref{sc-set}.

  \emph{Case \ruleref{sc-var}}.
  Then $C_1 = \set{z}$,
  and $z : \tCap{C_2}{S} \in \G, x :_D P, \Delta$.
  Now we inspect where $z$ is bound.
  \begin{itemize}
  \item \emph{When $x = z$}.
    The goal becomes $\sub{\G, \theta\Delta}{x^\prime}{\theta C_2}$.
    By inspecting the derivation of $\typ{\G}{x^\prime}{P}$,
    we can show that
    $x^\prime : P^\prime \in \G$
    and $\sub{\G}{\set{x^\prime}}{C_2}$.
    By weakening we have
    $\sub{\G, \theta\Delta}{\set{x^\prime}}{C_2}$.
    Also, by the well-formedness,
    we can show that $x \notin C_2$
    and therefore $\theta C_2 = C_2$.
    This case is therefore concluded.

  \item \emph{When $z$ is bound in $\G$}.
    By well-formedness we know that $x \notin \fv{C_2}$.
    Therefore, $\theta C_2 = C_2$.
    Also, $z : \tCap{C_2}{S} \in \G, \theta\Delta$,
    and we conclude this case by \ruleref{sc-var}.

  \item \emph{When $z$ is bound in $\Delta$}.
    Then $z : \tCap{\theta S}{\theta C_2} \in \G, \theta \Delta$.
    Therefore this case is concluded by IH and \ruleref{sc-var}.
  \end{itemize}

  \emph{Case \ruleref{sc-rdr-cap}}.
  By the same rule.

  \emph{Case \ruleref{sc-reader}}.
  We conclude using the IH, Lemma \ref{lemma:term-subst-pres-isrdr}, and the same rule.
\end{proof}

\begin{lemma}[Term substitution preserves subtyping]
  \label{lemma:term-subst-pres-subtyping}
	If
  (i) $\sub{\G, x :_D P, \Delta}{T}{U}$,
  and (ii) $\typ{\G}{x^\prime}{P}$
  then
  $\sub{\G, \theta\Delta}{\theta T}{\theta U}$
  where $\theta = \fSubst{x}{x^\prime}$.
\end{lemma}

\begin{proof}
	By induction on the subtype derivation.

  \emph{Case \ruleref{refl} and \ruleref{top}}.
  By the same rule.

  \emph{Case \ruleref{capt}}.
  By IH, Lemma \ref{lemma:term-subst-pres-subcapt}
  and the same rule.

  \emph{Case \ruleref{trans}, \ruleref{boxed}, \ruleref{fun} and \ruleref{tfun}}.
  By IH and the same rule.

  \emph{Case \ruleref{tvar}}.
  Then $T = X$, $U = {S}$, and $X <: {S} \in \G, x :_D P, \Delta$.
  Our goal is to show that
  $\sub{\G, \theta\Delta}{\theta X}{\theta {S}}$.
  We inspect where $X$ is bound.
  First, we show that $\theta X = X$ since $X \ne x$.
  If $X \in \dom{\G}$,
  we can show that $x \notin \fv{{S}}$ by the well-formedness of the environment.
  Therefore $\theta S = S$.
  Since $X <: S \in \G, \theta\Delta$, we can conclude this case by \ruleref{tvar}.
  Otherwise if $X \in \dom{\Delta}$,
  we have $X <: \theta S \in \G, \theta\Delta$.
  This case is therefore concluded by \ruleref{tvar} too.
\end{proof}

\begin{lemma}[Term substitution preserves separation]
  \label{lemma:term-subst-pres-ninter}
	If
  (i) $\ninter{\G, x :_D P, \Delta}{C_1}{C_2}$,
  (ii) $\typ{\G}{x^\prime}{P}$
  and (iii) $\ninter{\G}{D}{x^\prime}$,
  then
  $\ninter{\G, \theta\Delta}{\theta C_1}{\theta C_2}$
  where $\theta = \fSubst{x}{x^\prime}$.
\end{lemma}

\begin{proof}
	By induction on the separation derivation.

  \emph{Case \ruleref{ni-symm}}.
  Then $\ninter{\G, x :_D P, \Delta}{C_2}{C_1}$.
  We conclude this case by IH and the same rule.

  \emph{Case \ruleref{ni-set}}.
  Then $C_1 = \set{\tau_i}_{i=1,\cdots,n}$.
  By repeated IH we have
  $\overline{\ninter{\G, \theta\Delta}{\theta \tau_i}{\theta C_2}}^{i=1,\cdots,n}$.
  Note that $\theta C_1 = \bigcup_{i=1,\cdots,n} \theta \tau_i$.
  we can therefore conclude this case by \ruleref{ni-set}.

  \emph{Case \ruleref{ni-degree}}.
  Then $C_1 = \set{z_1}$,
  $C_2 = \set{z_2}$,
  $z_1 :_{D_1} T \in \G, x :_D P, \Delta$,
  and $z_2 \in D_1$.
  Now we do a case analysis on where $z_1$ is bound.
  \begin{itemize}
  \item \emph{When $z_1 = x$}.
    Then we have $\theta C_1 = x^\prime$ and $z_2 \in D$.
    By the well-formedness of the environment,
    $x \notin D$,
    therefore $z_2 \ne x$
    and $\theta z_2 = z_2$.
    The goal becomes $\ninter{\G, \theta\Delta}{x^\prime}{z_2}$.
    Note that we have $\ninter{\G}{D}{x^\prime}$,
    on which we can invoke Lemma \ref{lemma:ninter-to-elem}
    to show that $\ninter{\G}{z_2}{x^\prime}$.
    Now we conclude by \ruleref{ni-symm} and weakening.

  \item If $z_1 \in \dom{\G}$.
    We have $\theta D_1 = D_1$,
    and $x \notin D_1$ by the well-formedness.
    Therefore, we have $z_2 \ne x$, which implies that
    $\theta\set{z_2} = \set{z_2}$.
    Now we can conclude this case by \ruleref{ni-degree}

  \item If $z_1 \in \dom{\Delta}$.
    Then $z_1 :_{\theta D_1} \theta T \in \theta\Delta$.
    We can show that
    $\theta z_2  \in \theta D_1$.
    This case is therefore concluded by IH and \ruleref{ni-degree}.
  \end{itemize}

  \emph{Case \ruleref{ni-var}}.
  Then $C_1 = \set{y}$,
  $y :_{D_1} T \in \G, x :_D P, \Delta$,
  and $\ninter{\G, x :_D P, \Delta}{\cs{T}}{C_2}$.
  By the IH we can show that
  $\ninter{\G, \theta\Delta}{\theta\cs{T}}{\theta C_2}$.
  Now we inspect where $y$ is bound.
  \begin{itemize}
  \item If $x = y$,
    then $D = D_1$ and $T = P$.
    The goal becomes $\ninter{\G, \theta\Delta}{x'}{\theta C_2}$.
    By inspecting the derivation of $\typ{\G}{x'}{P}$,
    we can show that
    $\sub{\G}{\set{x'}}{\cs{P}}$.
    By the well-formedness we can show that $\theta\cs{P} = \cs{P}$.
    Now we can conclude this case by Lemma \ref{lemma:subcapt-pres-ninter}.

  \item If $y \in \dom{\G} \cup \dom{\Delta}$,
    we can show that in both cases
    $y :_{\theta D_1} \theta T \in \G, \theta D$.
    We can conclude immediately by using the IH and the \ruleref{ni-var} rule.
  \end{itemize}

  \emph{Case \ruleref{ni-reader}}.
  Then $C_1 = \set{z_1}$, $C_2 = \set{z_2}$,
  and $\sub{\G}{\set{z_i}}{\set{\rdroot}}$ for $i = 1, 2$.
  We conclude by using Lemma \ref{lemma:term-subst-pres-subcapt}.
\end{proof}

\begin{lemma}[Term substitution preserves typing]
  \label{lemma:term-subst-pres-typing}
	If
  (i) $\typ{\G, x :_D P, \Delta}{t}{T}$,
  (ii) $\typ{\G}{x^\prime}{P}$
  and (iii) $\ninter{\G}{D}{x'}$,
  then
  $\typ{\G, \theta\Delta}{\theta t}{\theta T}$
  where $\theta = \fSubst{x}{x^\prime}$.
\end{lemma}

\begin{proof}
	By induction on the typing derivation.

  \emph{Case \ruleref{var}}.
  Then $t = y$
  and $y :_{D_y} T \in \G, x :_D P, \Delta$.
  If $x = y$,
  we have $T = P$
  and the goal is to show that
  $\typ{\G, \theta\Delta}{x^\prime}{\theta P}$.
  By the well-formedness of $\G, x :_D P$,
  $x \notin \theta P$.
  Therefore, $\theta P = P$
  and we conclude this case from the premise.
  Otherwise, if $x \ne y$,
  we have $\theta t = y$,
  and we proceed by a case analysis on where $y$ is bound.
  \begin{itemize}
  \item If $y \in \dom{\G}$,
    then $x \notin T$ by the well-formedness of $\G, x :_D P$,
    and therefore $\theta T = T$.
    This case can be concluded by \ruleref{var}.

  \item If $y \in \dom{\Delta}$,
    then the goal becomes $\typ{\G, \theta\Delta}{y}{\theta T}$.
    We have $y :_{\theta D} \theta T \in \Delta$
    and this case is again concluded by \ruleref{var}.
  \end{itemize}

  \emph{Case \ruleref{sub}}.
  By IH, Lemma \ref{lemma:term-subst-pres-subtyping}, and \ruleref{sub}.

  \emph{Case \ruleref{abs}}.
  Then $t = \lambda(z :_{D_z} U) s$,
  $T = \cv{s}/z \ \forall(z :_{D_z} U) Q$,
  and $\typ{\G, x :_D P, \Delta, z :_{D_z} U}{s}{Q}$.
  By IH we show that
  $\typ{\G, x :_D P, \theta\Delta, z :_{\theta D_z} \theta U}{\theta s}{\theta Q}$.
  This case is therefore concluded by \ruleref{abs}.

  \emph{Case \ruleref{tabs}}. As above.

  \emph{Case \ruleref{app}}.
  Then $t = y_1\, y_2$,
  $\typ{\G, x :_D P, \Delta}{y_1}{C\,\forall(z :_{D_f} U) T^\prime}$
  where $T = \fSubst{z}{y_2} T^\prime$,
  $\typ{\G, x :_D P, \Delta}{y_2}{U}$,
  and $\ninter{\G, x :_D P, \Delta}{D_f}{y_2}$.
  By IH,
  $\typ{\G, \theta \Delta}{\theta y_1}{\theta C\,\forall(z :_{\theta D_f} \theta U) \theta T}$
  and $\typM{D_f^{\prime\prime}}{\G, \theta\Delta}{\theta y_2}{\theta U}$.
  By Lemma \ref{lemma:term-subst-pres-ninter},
  we can show that
  $\typ{\G, \theta\Delta}{\theta D_f}{\theta y_2}$.
  Now we can invoke \ruleref{app-r} to show that
  $\typ{\G, \theta \Delta}{\theta(y_1\,y_2)}{\fSubst{z}{\theta y_2}(\theta T^\prime)}$.
  Since $z$ is fresh we have $\theta(\fSubst{z}{y_2} T^\prime)$
  and conclude this case.

  \emph{Case \ruleref{tapp}}.
  As above.

  \emph{Case \ruleref{box}}.
  Then $t = \Box y$,
  and $\typ{\G, x :_D P, \Delta}{y}{\tCap{C}{S}}$
  and $T = \Box\tCap{C}{S}$.
  By IH,
  we have $\typ{\G, \theta\Delta}{\theta y}{\theta(\tCap{C}{S})}$.
  Since $\theta C \subseteq \dom{\G, \theta\Delta}$,
  we can show that $\typ{\G, \theta}{\theta(\Box y)}{\Box\theta(\tCap{C}{S})}$
  by \ruleref{box}
  and conclude this case.

  \emph{Case \ruleref{unbox}}.
  As above.

  \emph{Case \ruleref{let}}.
  Then $t = \tLetMode{m}{z}{s}{u}$,
  $\typ{\G, x :_{\set{}} P, \Delta}{s}{U}$,
  $\typ{\G, x :_D P, \Delta, z :_{\set{}} U}{u}{T}$,
  and $\ninter{\G, x :_D P, \Delta}{s}{u}$.
  By IH, we have
  $\typ{\G, \theta\Delta}{\theta s}{\theta U}$,
  and $\typ{\G, \theta\Delta, z :_{\set{}} \theta U}{\theta u}{\theta T}$.
  By Lemma \ref{lemma:term-subst-pres-ninter},
  we show that
  $\ninter{\G, \theta\Delta}{\theta s}{\theta u}$.
  Now we can conclude this case by \ruleref{let}.

  \emph{Case \ruleref{dvar}}.
  Then $t = \tLetVarM{D_z}{z}{y}{s}$,
  $\typ{\G, x :_D P, \Delta}{y}{S}$,
  $\typ{\G, x :_D P, \Delta, z :_{D_z} \set{\univ{}}\,\tRef{S}}{s}{T}$.
  We can conclude this case by IH and \ruleref{dvar}.

  \emph{Case \ruleref{read} and \ruleref{write}}.
  By IH and the same rule.
\end{proof}

\paragraph{Type Substitution}

\begin{lemma}[Type substitution preserves reader checking]
  \label{lemma:type-subst-pres-reader-checking}
  If
  (i) $\isrdrtype{\G, X <: S, \Delta}{T}$,
  and (ii) $\subDft{R}{S}$,
  then $\isrdr{\G, \theta\Delta}{\theta T}$,
  where $\theta = \fSubst{X}{R}$.
\end{lemma}

\begin{proof}
  By induction on the derivation.

  \emph{Case \ruleref{rd-reader}}.
  We conclude immediately using the same rule.

  \emph{Case \ruleref{rd-tvar}}.
  Then $T = X\capt C$ for some $C, Y$,
  $Y<:S_0 \in \G, X<:S, \Delta$,
  and $\isrdrtype{\G, X<:S, \Delta}{S_0\capt C}$.
  By the IH we can show that
  $\isrdrtype{\G, \theta\Delta}{\theta(S_0\capt C)}$.
  Now, we proceed by a case analysis on where $Y$ is bound.
  \begin{itemize}
  \item If $Y = X$. Then
    $\theta(Y\capt C) = R\capt C$ and $S_0 = S$.
    First, by the wellformedness we can show that $\theta(S_0\capt C) = S_0\capt C$,
    and therefore $\isrdrtype{\G, \theta\Delta}{S_0\capt C}$.
    Now we conclude this case by weakening and Lemma \ref{lemma:subtype-pres-rdr-checking}.

  \item If $Y$ is bound in either $\G$ or $\Delta$ then in both cases
    we have $Y<:\theta S_0 \in \G, \theta\Delta$.
    We conclude by the IH and the \ruleref{rd-tvar} rule.
  \end{itemize}
\end{proof}

\begin{lemma}[Type substitution preserves \textsf{is-reader}]
  \label{lemma:type-subst-pres-isrdr}
  If
  (i) $\isrdr{\G, X <: S, \Delta}{z}$,
  and (ii) $\subDft{R}{S}$,
  then $\isrdr{\G, \theta\Delta}{\theta z}$,
  where $\theta = \fSubst{X}{R}$.
\end{lemma}

\begin{proof}
  Then we have $z : T \in \G, X <: S, \Delta$,
  and $\sub{\G, X <: S, \Delta}{T}{\tRdr{S_0}\capt C_0}$ for some $C_0, S_0$.
  We first invoke Lemma \ref{lemma:subtype-pres-rdr-checking} to show that
  $\isrdrtype{\G, X<:S, \Delta}{T}$.
  Then, by Lemma \ref{lemma:type-subst-pres-reader-checking} we can show that
  $\isrdrtype{\G, \theta\Delta}{\theta T}$.
  By inspecting where $z$ is bound, we can show that
  $z : \theta T \in \G, \theta\Delta$,
  and finally conclude by Lemma \ref{lemma:is-rdr-eqv} and the \ruleref{rd-tvar} rule.





\end{proof}

\begin{lemma}[Type substitution preserves subcapturing]
  \label{lemma:type-subst-pres-subcapt}
	If
  (i) $\sub{\G, X <: S, \Delta}{C}{D}$,
  and (ii) $\sub{\G}{R}{S}$,
  then
  $\typ{\G, \theta\Delta}{C}{D}$,
  where $\theta = \fSubst{X}{R}$.
\end{lemma}

\begin{proof}
	By induction on the subcapture derivation,
  wherein most cases do not rely on the type bindings in the context
  and thus can be concluded immediatley by the IH and the same rule.
  In the \ruleref{sc-reader} case
  we invoke Lemma \ref{lemma:type-subst-pres-isrdr} to conclude.
\end{proof}

\begin{lemma}[Type substitution preserves subtyping]
  \label{lemma:type-subst-pres-subtyp}
	If
  (i) $\sub{\G, X <: S, \Delta}{T}{U}$,
  and (ii) $\sub{\G}{R}{S}$,
  then
  $\sub{\G, \theta\Delta}{\theta T}{\theta U}$,
  where $\theta = \fSubst{X}{R}$.
\end{lemma}

\begin{proof}
	By induction on the subtype derivation.

  \emph{Case \ruleref{refl} and \ruleref{top}}.
  By the same rule.

  \emph{Case \ruleref{capt}}.
  By Lemma \ref{lemma:type-subst-pres-subcapt},
  and \ruleref{capt}.

  \emph{Case \ruleref{trans}, \ruleref{boxed}, \ruleref{fun} and \ruleref{tfun}}.
  By IH and application of the same rule.

  \emph{Case \ruleref{tvar}}.
  Then $T = Y$,
  $Y <: S^\prime \in \G, X <: S, \Delta$,
  and $U = S^\prime$.
  Our goal is
  $\sub{\G, \theta\Delta}{\theta Y}{\theta S^\prime}$.
  Now we inspect where $Y$ is bound.
  \begin{itemize}
  \item \emph{When $X = Y$}.
    Then $S^\prime = S$,
    and the goal becomes
    $\sub{\G, \theta\Delta}{R}{\theta S}$.
    By the well-formedness,
    we can show that $X \notin \fv{S}$
    and therefore $\theta S = S$.
    Now we conclude by weakening the premise.

  \item \emph{When $Y \in \dom{\G}$}.
    By the well-formedness,
    $X \notin \fv{S^\prime}$.
    Therefore, $\theta S^\prime = S^\prime$.
    Also, $\theta Y = Y$,
    and $Y <: S^\prime \in \G, \theta\Delta$.
    This case is therefore concluded by \ruleref{tvar}.

  \item \emph{When $Y \in \dom{\Delta}$}.
    Then $Y <: \theta S^\prime \in \G, \theta\Delta$.
    This case is therefore concluded by \ruleref{tvar}.
  \end{itemize}
\end{proof}

\begin{lemma}[Type substitution preserves separation]
  \label{lemma:type-subst-pres-ninter}
  If
  (i) $\ninter{\G, X <: S, \Delta}{C_1}{C_2}$,
  and (ii) $\sub{\G}{R}{S}$,
  then
  $\ninter{\G, \theta\Delta}{C_1}{C_2}$,
  where $\theta = \fSubst X R$.
\end{lemma}

\begin{proof}
	By induction on the separation derivation.

  \emph{Case \ruleref{ni-symm}}.
  Then $\ninter{\G, X <: S, \Delta}{C_2}{C_1}$.
  We conclude this case by the IH and \ruleref{ni-symm}.

  \emph{Case \ruleref{ni-set}}.
  Then $\overline{\ninter{\G, X <: S, \Delta}{y}{C_2}}^{y \in C_1}$.
  By applying the IH repeated, we can show that
  $\overline{\ninter{\G, \theta\Delta}{y}{C_2}}^{y \in C_1}$.
  This case is therefore concluded by \ruleref{ni-set}.

  \emph{Case \ruleref{ni-degree}}.
  Then $C_1 = \set{x}$, $C_2 = \set{y}$,
  $x :_D T \in \G, X <: S, \Delta$,
  and $y \in D$.
  Note that $X \notin D$,
  therefore we have $\theta D = D$.
  Therefore, no matter where $x$ is bound,
  we have $x :_D T^\prime \in \G, \theta\Delta$.
  This case can therefore be concluded by \ruleref{ni-degree}.

  \emph{Case \ruleref{ni-var}}.
  Then $C_1 = \set{x}$, $C_2 = \set{y}$,
  $x :_D T \in \G, X <: S, \Delta$,
  and $\ninter{\G, X <: S, \Delta}{\cs{T}}{y}$.
  Depending on where $x$ is bound,
  either $x :_D T \in \G, \theta\Delta$,
  or $x :_{\theta D} \theta T \in \theta\Delta$.
  Since $X \notin \cs{T}$,
  we can show that $\cs{\theta T} = \theta \cs{T} = \cs{T}$.
  By IH, we can show that
  $\ninter{\G, \theta\Delta}{\cs{T}}{y}$.
  We can therefore conclude this case by \ruleref{ni-var}.

  \emph{Case \ruleref{ni-reader}}.
  By applying the IH, Lemma \ref{lemma:type-subst-pres-subcapt}
  and the same rule.
\end{proof}

\begin{lemma}[Type substitution preserves typing]
  \label{lemma:type-subst-pres-typing}
	If
  (i) $\typ{\G, X <: S, \Delta}{t}{T}$,
  and (ii) $\sub{\G}{R}{S}$,
  then
  $\typ{\G, \theta\Delta}{\theta t}{\theta T}$
  where $\theta = \fSubst{X}{R}$.
\end{lemma}

\begin{proof}
	By induction on the typing derivation.

  \emph{Case \ruleref{var}}.
  Then $t = x$,
  $x : \tCap{C}{S^\prime} \in \G, X <: S, \Delta$,
  and the goal becomes
  $\typ{\G, \theta\Delta}{\theta x}{\theta(\tCap{\set{x}}{S^\prime})}$.
  Since $x \ne X$,
  we have $\theta x = x$.
  Now we inspect where $x$ is bound.
  \begin{itemize}
  \item \emph{When $x \in \dom{\G}$}.
    By the well-formedness we can show that
    $X \notin \fv{\tCap{C}{S^\prime}}$.
    Therefore, $\theta S^\prime = S^\prime$.
    We have $x : \tCap{C}{S^\prime} \in \G, \theta\Delta$.
    This case is concluded by \ruleref{var}.

  \item \emph{When $x \in \dom{\Delta}$}.
    We have $x : \theta(\tCap{C}{S^\prime}) \in \G, \theta\Delta$.
    This case is therefore concluded by \ruleref{var}.
  \end{itemize}

  \emph{Case \ruleref{abs} and \ruleref{tabs}}.
  In both cases we conclude by IH and the same rule.

  \emph{Case \ruleref{app-r}}.
  Then $t = x\,y$,
  $\typ{\G, X <: S, \Delta}{x}{C\,\forall(z :_D U) T^\prime}$,
  $\typ{\G, X <: S, \Delta}{y}{U}$,
  $\ninter{\G, X<:S, \Delta}{D}{y}$,
  and $T = \fSubst z y T^\prime$.
  We can conclude this case
  by IH, Lemma \ref{lemma:type-subst-pres-ninter},
  the fact that $\theta\fSubst z y T^\prime = \fSubst z y \theta T^\prime$
  and the \ruleref{app-r} rule.

  \emph{Case \ruleref{tapp}}.
  Then $t = x[S^\prime]$,
  $\typ{\G, X <: S, \Delta}{x}{C\,\forall[Z <: S^\prime] T^\prime}$
  and $T = \fSubst Z {S^\prime} T^\prime$.
  The goal becomes
  ${\G, \theta\Delta}\vdash {\theta(x[S^\prime])}: {\theta \fSubst Z {S^\prime} T^\prime}$.
  By IH we have
  $\typ{\G, \theta\Delta}{\theta x}{\theta C\, \forall[Z <: \theta S^\prime] \theta T^\prime}$.
  By \ruleref{tapp} we can show that
  $\typ{\G, \theta\Delta}{\theta(x[S^\prime])}{\fSubst Z {\theta S^\prime} \theta T^\prime}$.
  We observe that
  $\fSubst Z {\theta S^\prime} \theta T^\prime = \theta \fSubst Z {S^\prime} T^\prime$
  and conclude this case.

  \emph{Case \ruleref{box}}.
  Then $t = \Box\ y$
  $\typ{\G, X <: S, \Delta}{y}{\tCap{C}{S^\prime}}$,
  and $C \subseteq \dom{\G, X <: S, \Delta}$.
  The goal becomes
  $\typ{\G, \theta\Delta}{\theta(\Box\ y)}{\theta(\Box\ \tCap{C}{S^\prime})}$.
  Note that since $y \ne X$,
  we have $\theta y = y$.
  Proof proceeds by inspecting the location of $y$ in the context.
  \begin{itemize}
  \item \emph{When $y \in \dom{\G}$}.
    By the well-formedness we can show that
    $X \notin \fv{\tCap{C}{S^\prime}}$.
    Therefore, $\theta(\tCap{C}{S^\prime}) = \tCap{C}{S^\prime}$.
    We observe that $y : \tCap{C}{S^\prime} \in \G, \theta\Delta$,
    and show that
    $\typ{\G, \theta\Delta}{y}{\tCap{C}{S^\prime}}$ by \ruleref{var}.
    Now this case can be concluded by \ruleref{box}.

  \item \emph{When $y \in \dom{\Delta}$}.
    Then we observe that
    $y : \theta(\tCap{C}{S^\prime}) \in \G, \theta\Delta$,
    and derive that
    $\typ{\G, \theta\Delta}{y}{\theta(\tCap{C}{S^\prime})}$ by \ruleref{var}.
    This case is concluded by \ruleref{box}.
  \end{itemize}

  \emph{Case \ruleref{unbox}}.
  Analogous to the previous case.

  \emph{Case \ruleref{sub}}.
  By IH, Lemma \ref{lemma:type-subst-pres-subtyp} and \ruleref{sub}.

  \emph{Case \ruleref{let}}.
  Then $t = \tLetMode{m}{x}{s}{u}$,
  $\typ{\G, X <: S, \Delta}{s}{T_0}$,
  $\typ{\G, X <: S, \Delta, x :_{\set{}} T_0}{u}{T}$,
  and $\ninter{\G, X <: S, \Delta}{s}{u}$.
  By IH we can show that
  $\typ{\G, \theta\Delta}{\theta s}{\theta T_0}$,
  and $\typ{\G, \theta\Delta, x :_{\set{}} \theta T_0}{\theta u}{\theta T}$.
  Since $X \notin \cv{s}, \cv{u}, \cs{T_0}$,
  we can show that
  $\theta \cv{s} = \cv{s}$,
  and $\theta \cv{u} = \cv{u}$.
  By Lemma \ref{lemma:type-subst-pres-ninter},
  we can show that
  $\ninter{\G, \theta\Delta}{\theta s}{\theta u}$.
  We can therefore conclude this case by \ruleref{let}.

  \emph{Case \ruleref{dvar}}.
  Then $t = \tLetVar{x}{y}{s}$,
  $\typ{\G, X <: S, \Delta}{y}{S}$,
  $\typ{\G, X <: S, \Delta, x :_D \tCap{\set{\univ{}}}{\tRef{S}}}{s}{T}$,
  and $\wfTyp{\G, X <: S, \Delta}{D}$.
  By IH we can show that
  $\typ{\G, \theta\Delta}{\theta y}{\theta S}$,
  and $\typ{\G, \theta\Delta, x :_{\theta D} \theta(\tCap{\set{\univ{}}}{\tRef{S}})}{\theta s}{\theta T}$.
  Also, we can show that $\wfTyp{\G, \theta\Delta}{D}$.
  This case is therefore concluded by \ruleref{dvar}.

  \emph{Case \ruleref{read} and \ruleref{write}}.
  By IH and the same rule.
\end{proof}

\subsubsection{Soundness Theorems}

\begin{lemma}[Canonical forms: term abstraction]
  \label{lemma:cf-fun}
	$\typ{\G}{v}{C\,\forall(x :_D U) T}$
  implies
  $v = \lambda(x :_D U^\prime) t$ for some $U^\prime$ and $t$,
  such that
  $\subDft{U}{U^\prime}$ and $\typ{\G, x :_D U^\prime}{t}{T}$.
\end{lemma}

\begin{proof}
	By induction on the typing derivation.

  \emph{Case \ruleref{abs}}.
  We conclude immediately from the premise.

  \emph{Case \ruleref{sub}}.
  Then $\typ{\G}{v}{T_0}$,
  and $\sub{\G}{T_0}{C\,\forall(x :_D U) T}$.
  By Lemma \ref{lemma:subtype-inversion-fun} and Lemma \ref{lemma:val-typing},
  we can show that
  $T_0 = C_0,\forall(x :_D U_0) T_0$,
  $\subDft{U}{U_0}$,
  and $\sub{\G, x :_D U_0}{T_0}{T}$.
  Invoking IH,
  we can show that
  $v = \lambda(x :_D U^\prime) t$ such that
  $\sub{\G}{U_0}{U^\prime}$,
  and $\typ{\G, x :_D U^\prime}{t}{T_0}$.
  We conclude by the \ruleref{trans}
  and the \ruleref{sub} rule.
\end{proof}

\begin{lemma}[Canonical forms: type abstraction]
  \label{lemma:cf-tfun}
	$\typ{\G}{v}{C\,\forall[X <: S] T}$
  implies
  $v = \lambda[X <: S^\prime] t$ for some $S^\prime$ and $t$,
  such that
  $\subDft{S}{S^\prime}$ and $\typ{\G, X <: S^\prime}{t}{T}$.
\end{lemma}

\begin{proof}
	Analogous to the proof of Lemma \ref{lemma:cf-fun}.
\end{proof}

\begin{lemma}[Canonical forms: boxed term]
  \label{lemma:cf-boxed}
	$\typ{\G}{v}{C\,\Box\ T}$
  implies
  $v = \Box\ x$ for some $x$,
  such that
  $\typ{\G}{x}{T}$
  and ${\cs{T}} \subseteq {\dom{\G}}$.
\end{lemma}

\begin{proof}
	Analogous to the proof of Lemma \ref{lemma:cf-fun}.
\end{proof}

\begin{lemma}[Canonical forms: reader]
  \label{lemma:cf-reader}
  $\typ{\G}{v}{C\,\tRdr{S}}$
  implies
  $v = \reader x$ for some $x$,
  such that
  $\typ{\G}{x}{C\,\tRef{S}}$.
\end{lemma}

\begin{proof}
  Analogous to the proof of Lemma \ref{lemma:cf-fun}.
\end{proof}

\begin{lemma}[Store lookup inversion: typing]
  \label{lemma:store-lookup-inversion-typing}
  If (i) $\match{}{\sta}{\G}$
  (ii) $\typ{\G, \Delta}{x}{\tCap{C}{S}}$,
  (iii) $\sta(\tVal{x}) = v$,
  then $\typ{\G, \Delta}{v}{\tCap{C^\prime}{S}}$.
\end{lemma}

\begin{proof}
  By induction on the derivation of $\match{}{\sta}{\G}$.

  \emph{Case \ruleref{st-empty}}.
  Contradictory.

  \emph{Case \ruleref{st-val}}.
  Then $\sta = \sta_0, \tVal{y} \mapsto v$,
  $\match{}{\sta_0}{\G_0}$,
  $\typ{\G_0}{v}{T}$,
  and $\G = \G_0, y : T$.
  If $x = y$,
  then $x : T \in \G, \Delta$.
  We have $\sta(\tVal{x}) = v$.
  By inspecting the typing derivation $\typ{\G, \Delta}{x}{\tCap{C}{S}}$,
  we can show that
  $\sub{\G, \Delta}{\set{x}}{C}$
  and $\sub{\G, \Delta}{S^\prime}{S}$
  where $T = \tCap{C^\prime}{S^\prime}$ for some $C^\prime$.
  By weakening we show that
  $\typ{\G, \Delta}{v}{\tCap{C^\prime}{S^\prime}}$.
  Then, we invoke the \ruleref{sub} rule and conclude.
  Otherwise, we have $x \ne y$.
  Finally by the IH we conclude this case.

  \emph{Case \ruleref{st-var}}.
  Then $\sta = \sta_0, \tVar{y} = v$,
  $\match{}{\sta_0}{\G_0}$,
  $\typ{\G_0}{v}{S}$,
  and $\G = \G_0, y :_{\dom{\G_0}} \tCap{\set{\univ{}}}{\tRef{S}}$.
  If $x = y$,
  it contradicts with $\sta(\tVal{x}) = v$.
  Otherwise, we have $x \ne y$,
  we conclude this case by the IH.

  \emph{Case \ruleref{st-set}}. By IH.
\end{proof}

\begin{lemma}[Store lookup inversion: term abstraction]
  \label{lemma:lookup-inversion-fun}
	If
  (i) $\match{}{\sta; e}{\G; \Delta}$;
  (ii) $\typ{\G, \Delta}{x}{C\,\forall(z:_D T)U}$
  and (iii) $\sta(\tVal{x}) = \lambda(z:_D T^\prime) t$,
  then
  $\sub{\G, \Delta}{T}{T^\prime}$
  and $\typ{\G, \Delta, z : T^\prime}{t}{U}$.
\end{lemma}

\begin{proof}
  We first invoke Lemma \ref{lemma:store-lookup-inversion-typing} to show that
  $\typ{\G, \Delta}{\lambda(z :_D T^\prime) t}{\tCap{C^\prime}{\forall(z :_D T) U}}$.
  By Lemma \ref{lemma:cf-fun},
  we can show that
  $\sub{\G, \Delta}{T}{T^\prime}$,
  and $\typ{\G, \Delta, z :_D T^\prime}{t}{U}$.
  We therefore conclude.
\end{proof}

\begin{lemma}[Store lookup inversion: type abstraction]
  \label{lemma:lookup-inversion-tfun}
	If
  (i) $\match{}{\sta; e}{\G; \Delta}$;
  (ii) $\typ{\G, \Delta}{x}{C\,\forall[X <: S] U}$
  and (iii) $\sta(\tVal{x}) = \lambda[X <: S^\prime] t$,
  then
  $\sub{\G, \Delta}{S}{S^\prime}$
  and $\typ{\G, \Delta, X <: S^\prime}{t}{U}$.
\end{lemma}

\begin{proof}
	Analogous to the proof of Lemma \ref{lemma:lookup-inversion-fun}.
\end{proof}

\begin{lemma}[Store lookup inversion: boxed term]
  \label{lemma:lookup-inversion-tfun}
	If
  (i) $\match{}{\sta; e}{\G; \Delta}$;
  (ii) $\typ{\G, \Delta}{x}{C^\prime\,\Box\ \tCap{C}{S}}$
  and (iii) $\sta(\tVal{x}) = \Box\ y$,
  then
  $\typ{\G, \Delta}{y}{\tCap{C}{S}}$.
\end{lemma}

\begin{proof}
	Analogous to the proof of Lemma \ref{lemma:lookup-inversion-fun}.
\end{proof}

\begin{lemma}[Store lookup inversion: reader]
  \label{lemma:lookup-inversion-reader}
  If
  (i) $\match{}{\sta; e}{\G; \Delta}$;
  (ii) $\typ{\G, \Delta}{x}{C\,\tRdr{S}}$
  and (iii) $\sta(\tVal{x}) = \reader y$,
  then
  $\typ{\G, \Delta}{y}{C'\,\tRef{S}}$.
\end{lemma}

\begin{proof}
	Analogous to the proof of Lemma \ref{lemma:lookup-inversion-fun}.
\end{proof}

\begin{lemma}[Value typing with strict capture set]
  \label{lemma:value-typing-strict-cs}
	$\typ{\G}{v}{T}$
  implies $\exists T^\prime$
  such that $\typ{\G}{v}{T^\prime}$ and $\cs{T^\prime} = \cv{v}$.
\end{lemma}

\begin{proof}
	By induction on the typing derivation.
  Only the following cases are applicable.

  \emph{Case \ruleref{abs}}.
  Then $v = \lambda(x :_D U) t$,
  $\typ{\G, x :_D U}{t}{T^\prime}$,
  and $T = \cv{t}/x\ \forall(z :_D U) T^\prime$.
  Then $\cs{T} = \cv{v} = \cv{t}/x$.
  This case is therefore concluded.

  \emph{Case \ruleref{tabs} and \ruleref{box}}.
  Analogous to \ruleref{abs}.

  \emph{Case \ruleref{sub}}.
  By IH.
\end{proof}

\begin{lemma}[Value typing inversion: capture set]
  \label{lemma:value-typing-inversion-cs}
	$\typ{\G}{v}{T}$ implies $\sub{\G}{\cv{v}}{\cs{T}}$.
\end{lemma}

\begin{proof}
	Analogous to the proof of Lemma \ref{lemma:value-typing-strict-cs}.
\end{proof}

\preservation*

\begin{proof}
  Proceed by case analysis on the reduction derivation.

  \emph{Case \ruleref{apply}}.
  Then $t = \evctx{x\,y}$,
  $\sta(\tVal{x}) = \lambda(z:_D U)s$,
  $t^\prime = \fSubst{z}{y}s$,
  and $\sta^\prime = \sta$.
  By Lemma \ref{lemma:evctx-typing-inversion},
  $\exists \Delta, Q$ such that
  $\typec{\G}{e}{\Delta}{Q}{x\,y}{T}$
  and $\typ{\G, \Delta}{x\,y}{Q^\prime}$,
  where $Q = \fSubst{z}{y} Q^\prime$.
  Now, by induction on this typing derivation,
  we prove that
  (i) $\typ{\G, \Delta}{x}{C\,\forall(z :_D U^\prime)}{Q^{\prime}}$
  (ii) $\typM{D^\prime}{\G, \Delta}{y}{U^\prime}$,
  and (iii) $\ninter{\G, \Delta}{D}{y}$.
  For the \ruleref{app} case, it is conclude immediately from the preconditions.
  For the \ruleref{sub} case, it can be concluded from IH and \ruleref{sub}.
  Other cases are not applicable.
  By Lemma \ref{lemma:lookup-inversion-fun},
  we can show that
  $\sub{\G, \Delta}{U^\prime}{U}$,
  and $\typ{\G, \Delta, z :_D U}{s}{Q^\prime}$.
  We invoke Lemma \ref{lemma:term-subst-pres-typing}
  to show that $\typ{\G, \Delta}{\fSubst{z}{y} s}{\fSubst{z}{y} Q^\prime}$.
  Since $\sta(\tVal{x}) = \lambda(z :_D U) s$
  and $\match{}{\sta}{\G}$,
  we have $x : T_x \in \G$ for some $T_x$
  where $\cs{T_x} = \cv{s}/z$.
  We can show that
  $\cv{\fSubst{z}{y} s} \subseteq \cv{x\,y}\setminus \set{x} \cup \cs{T_x}$.
  Now we invoke Lemma \ref{lemma:widening-to-subi}
  to show that
  $\subi{\G; \Delta}{\fSubst{z}{y} s}{x\,y}$.
  By evaluation context reification,
  we show that
  $\typ{\G}{\evctx{\fSubst{z}{y} s}}{T}$
  and conclude this case.

  \emph{Case \ruleref{tapply}}.
  As above, but making use of Corollary \ref{coro:subset-to-subi}.

  \emph{Case \ruleref{open}}.
  As above, but making use of Lemma \ref{lemma:subcapt-to-subi}.

  \emph{Case \ruleref{get}}.
  As above.

  \emph{Case \ruleref{lift-let}}.
  Then $t = \evctx{\tLetMode{m}{x}{v}{s}}$,
  $\sta^\prime = \sta, \tValM{m}{x} \mapsto v$,
  $t^\prime = \evctx{s}$,
  and $\fv{v} \cup D \subseteq \dom{\sta}$.
  Now we invoke Lemma \ref{lemma:evctx-typing-inversion}
  to show that
  $\exists \Delta, Q$ such that
  $\typec{\G}{e}{\Delta}{Q}{\tLetMode{m}{x}{v}{s}}{T}$,
  and $\typ{\G, \Delta}{\tLetMode{m}{x}{v}{s}}{Q}$.
  By inspecting this typing derivation,
  we can show that
  $\typ{\G, \Delta}{v}{U}$,
  $\typ{\G, \Delta, x :_{\set{}} U}{s}{Q}$,
  and $\ninter{\G, \Delta}{v}{s}$ if $m = \parmode$.
  By Lemma \ref{lemma:value-typing-strict-cs}, narrowing,
  and Lemma \ref{lemma:subcapt-pres-ninter},
  we can show that $\exists U^\prime$ such that
  $\cs{U^\prime} = \cv{v}$,
  $\typ{\G, \Delta}{v}{U^\prime}$,
  and $\typ{\G, \Delta, x :_D U^\prime}{s}{Q}$.
  Set $\G^\prime = \G, x :_D U^\prime$.
  We can show that $\match{}{\sta^\prime}{\G^\prime}$ by \ruleref{st-val}.
  In the next, we show that
  $\sub{\G, \Delta, x :_D U^\prime}{\cv{s}}{\cv{\tLetMode{D}{x}{v}{s}}}$.
  If $x \notin \cv{s}$,
  we have $\cv{\tLetMode{D}{x}{v}{s}} = \cv{s} = \cv{s}\setminus \set{x}$.
  And the sub-goal is proven by the reflexivity of subcapturing.
  Otherwise, if $x \in \cv{s}$,
  we have $\cv{\tLetMode{D}{x}{v}{s}} = \cv{s}\setminus \set{x} \cup \cv{v}$,
  and $\cv{s} \subseteq \cv{s}\setminus \set{x} \cup \set{x}$.
  We can show that
  $\sub{\G, \Delta}{\set{x}}{\cv{v}}$.
  By the reflexivity of subcapturing and Corollary \ref{coro:subcapt-join-both},
  we can prove the goal.
  Since $\cv{v} \subseteq \dom{\sta}$,
  we can show that $\G, x :_{\set{}} U^\prime$ is well-formed.
  Therefore by permutation
  we have $\sub{\G, x :_{\set{}} U^\prime, \Delta}{\cv{s}}{\cv{\tLetMode{D}{x}{v}{s}}}$.
  Now we invoke Lemma \ref{lemma:subcapt-to-subi}
  to show that
  $\subi{\G, x :_{\set{}} U^\prime; \Delta}{s}{\tLetMode{D}{x}{v}{s}}$.
  By Lemma \ref{lemma:evctx-typing-weakening},
  we can show that $\typec{\G, x :_{\set{}} U^\prime}{e}{\Delta}{Q}{\tLetMode{D}{x}{v}{s}}{T}$.
  Again by permutation, we show that
  $\typ{\G, x :_{\set{}} U^\prime, \Delta}{s}{Q}$.
  We can therefore invoke the reification to show that
  $\typ{\G^\prime}{\evctx{s}}{T}$
  and conclude this case.

  \emph{Case \ruleref{rename}}.
  Then $t = \evctx{\tLetMode{m}{x}{y}{s}}$,
  $t^\prime = \evctx{\fSubst{x}{y} s}$,
  and $\sta^\prime = \sta$.
  We start by invoking Lemma \ref{lemma:evctx-typing-inversion}
  to show that $\exists \Delta, Q$
  such that
  $\typec{\G}{e}{\Delta}{Q}{\tLetMode{m}{x}{y}{s}}{T}$
  and $\typ{\G, \Delta}{\tLetMode{m}{x}{y}{s}}{Q}$.
  By inspecting this typing derivation,
  we can show that
  $\typ{\G, \Delta}{y}{U}$ for some $U$,
  $\typ{\G, \Delta, x :_m U}{s}{Q}$,
  and $\ninter{\G, \Delta}{y}{s}$.
  Now we invoke Lemma \ref{lemma:term-subst-pres-typing} to show that
  $\typ{\G, \Delta}{\fSubst{x}{y} s}{\fSubst{x}{y} Q}$.
  Since $x \notin \fv{Q}$, we have $\fSubst{x}{y} Q = Q$.
  Note that we have
  $\cv{\fSubst{x}{y} s} = \cv{\tLetMode{D}{x}{y}{s}}$.
  By Corollary \ref{coro:subset-to-subi},
  we can show that $\subi{\G; \Delta}{\fSubst{x}{y} s}{\tLetMode{D}{x}{y}{s}}$.
  Now we invoke the reification of evaluation context to show that
  $\typ{\G}{\evctx{\fSubst{x}{y} s}}{T}$
  and conclude this case.

  \emph{Case \ruleref{lift-var}}.
  Then $t = \evctx{\tLetVar{x}{y}{s}}$,
  $\sta(\tVal{y}) = v$,
  $\sta^\prime = \sta, \tVar{x} = v$,
  and $t^\prime = \evctx{s}$.
  By Lemma \ref{lemma:evctx-typing-inversion-mut},
  we show that
  $\typ{\G, x :_{\dom{\G}} \tRef{S}\capt\set{\univ{}}}{\evctx{s}}{T}$.
  We can show that
  $\match{\G}{\sta^\prime}{\G, x :_{\dom{\G}} \tRef{S}\capt\set{\univ{}}}$
  by Lemma \ruleref{st-var}.
  We can therefore conclude this case.

  \emph{Case \ruleref{lift-set}}.
  Then $t = \evctx{\tWrite{x}{y}}$,
  $\sta(\tVal{y}) = v$,
  $t^\prime = \evctx{v}$,
  and $\sta^\prime = \sta, \tSet{x} = v$.
  Now we invoke Lemma \ref{lemma:evctx-typing-inversion}
  to show that $\exists \Delta, Q$ such that
  $\typec{\G}{e}{\Delta}{Q}{\tWrite{x}{y}}{T}$,
  and $\typ{\G, \Delta}{\tWrite{x}{y}}{Q}$.
  By induction on this typing derivation,
  we can show that $\exists S$
  such that
  $\typ{\G, \Delta}{x}{C\,\tRef{S}}$,
  $\typ{\G, \Delta}{y}{S}$,
  and $\sub{\G, \Delta}{S}{Q}$.
  By Lemma \ref{lemma:store-lookup-inversion-typing},
  we can show that $\typ{\G, \Delta}{v}{S}$.
  By Lemma \ref{lemma:value-typing-inversion-cs},
  we have $\sub{\G, \Delta}{\cv{v}}{\varemptyset}$.
  Therefore we have $\sub{\G, \Delta}{\cv{v}}{\cv{\tWrite{x}{y}}}$.
  Now we invoke Lemma \ref{lemma:subcapt-to-subi} to show that
  $\subi{\G; \Delta}{v}{\tWrite{x}{y}}$.
  Now we reify the evaluation context to show that
  $\typ{\G}{\evctx{v}}{T}$.
  The remaining to be shown is that
  $\match{}{\sta^\prime}{\G}$.
  By inverting the typing judgment
  $\typ{\G, \Delta}{x}{\tRef{S}\capt C}$
  and considering that $x \in \dom{\G}$,
  we can show that
  $x :_D \tRef{S}\capt C' \in \G$.
  By induction on $\match{}{\sta}{\G}$,
  we can show that
  $\typ{\G}{v}{S}$ by the fact that
  $\sta(\tVal{y}) = v$ and weakening.
  Now $\match{}{\sta^\prime}{\G}$ can be derived by \ruleref{st-set}.
  This case is therefore concluded.
\end{proof}

\begin{lemma}[Evaluation context trampolining]
  \label{lemma:evctx-trampoline}
  Consider a reduction derivation
  $\reduction{\stDft{t}}{\st{\sta^\prime}{t^\prime}}$.
  We have $\reduction{\stDft{\evctx{t}}}{\st{\sta^\prime}{\evctx{t^\prime}}}$ for any $e$.
\end{lemma}

\begin{proof}
	By straightforward case analysis on the reduction derivation.
  No rule makes use of the evaluation context.
  All cases are analogous, thus we present only the \ruleref{apply} case.

  \emph{Case \ruleref{apply}}.
  Then $t = \evpctx{x\,y}$,
  $\sta(\tVal{x}) = \lambda(z :_D T) s$,
  $\sta^\prime = \sta$,
  and $t^\prime = \evpctx{\fSubst{z}{y} s}$.
  None of these preconditions depends on the evaluation context $e$.
  We therefore conclude this case by the same rule.
\end{proof}

\begin{theorem}[Progress]
  If
  (i) $\match{}{\sta}{\G}$,
  (ii) $\typ{\G}{t}{T}$,
  then
  either $t$ is an answer,
  or $\exists \sta^\prime, t^\prime$ such that
  $\reduction{\stDft{t}}{\st{\sta^\prime}{t^\prime}}$.
\end{theorem}

\begin{proof}
  By induction on the typing derivation.

  \emph{Case \ruleref{var}}.
  Then $t = x$.
  We can conclude this case because $x$ is an answer.

  \emph{Case \ruleref{sub}}.
  We directly apply the IH and conclude.

  \emph{Case \ruleref{abs}, \ruleref{tabs}, \ruleref{box} and \ruleref{reader}}.
  We conclude these cases because $t$ is a value.

  \emph{Case \ruleref{app}}.
  Then $t = x\,y$,
  $\typ{\G}{x}{C\,\forall(z :_D U) T^\prime}$
  where $T = \fSubst z y T^\prime$,
  and $\typ{\G}{y}{U}$.
  By Lemma \ref{lemma:store-lookup-fun},
  we can show that there exists $v$ such that
  $\sta(\tVal{x}) = v$.
  By Lemma \ref{lemma:store-lookup-inversion-typing},
  we have $\typ{\G}{v}{C\,\forall(z :_D U) T^\prime}$.
  By Lemma \ref{lemma:cf-fun},
  we have
  $v = \lambda(z :_D U^\prime) s$ for some $U^\prime$ and $s$.
  Now we conclude by \ruleref{apply}.

  \emph{Case \ruleref{tapp}}.
  As above.

  \emph{Case \ruleref{unbox}}.
  Then $t = \tUnbox{C}{x}$,
  $\typ{\G}{x}{\Box\ (S\capt C)}$,
  and $C \subseteq \dom{\G}$.
  By Lemma \ref{lemma:store-lookup-pure-val},
  we can show that $\exists v$ such that
  $\sta(\tVal{x}) = v$.
  By Lemma \ref{lemma:store-lookup-inversion-typing},
  we can show that $\typ{\G}{v}{\Box\ (S\capt C)}$.
  By Lemma \ref{lemma:cf-boxed},
  we have
  $z = \Box\ y$ for some $y$.
  Now we conclude by \ruleref{open}.

  \emph{Case \ruleref{let}}.
  Then $t = \tLetMode{m}{x}{s}{u}$,
  $\typ{\G}{s}{T_0}$ for some $T_0$,
  and $\typ{\G, x :_{\set{}} T_0}{s}{T}$.
  Proceed by a case analysis on the kind of $s$.
  \begin{itemize}
  \item
    If $s$ is a value,
    we can show that $\fv{s} \subseteq \dom{\G} = \dom{\sta}$.
    This allows for the application of the \ruleref{lift-let} rule.

  \item
    If $s$ is a variable,
    we can conclude by \ruleref{rename}.

  \item
    Otherwise $s$ is a term.
    By IH we can show that
    $\reduction{\stDft{s}}{\st{\sta^\prime}{{s^\prime}}}$
    for some $\sta^\prime$ and $s^\prime$.
    By Lemma \ref{lemma:evctx-trampoline},
    we can derive that
    $\reduction{\stDft{\tLetMode{D}{x}{s}{u}}}{\st{\sta^\prime}{\tLetMode{D}{x}{s^\prime}{u}}}$
    and conclude this case.
  \end{itemize}

  \emph{Case \ruleref{dvar}}.
  Then $t = \tLetVar{x}{y}{s}$,
  $\typ{\G}{y}{S}$,
  and $\typ{\G, x :_D \tRef{S}\capt\set{\univ}}{s}{T}$.
  We can show that $\sta(\tVal{y}) = v$ for some $v$.
  This case is therefore concluded by \ruleref{lift-var}.

  \emph{Case \ruleref{read}}.
  Then $t = \tRead{x}$,
  and $\typ{\G}{x}{\tRdr{S}\capt {C}}$.
  By Lemma \ref{lemma:store-lookup-reader},
  we can show that
  $\exists y. \sta(\tVal{x}) = \reader y$.
  Now we invoke Lemma \ref{lemma:cf-reader} to show that
  $\typ{\G}{y}{\tRef{S}\capt C}$.
  By using Lemma \ref{lemma:store-lookup-mut} we can show that
  $\exists v$ such that
  $\sta(\tVar{y}) = v$.
  Now we conclude this case by \ruleref{get}.

  \emph{Case \ruleref{write}}.
  Then $t = \tWrite{x}{y}$,
  $\typ{\G}{x}{\tRef{S}\capt C}$,
  and $\typ{\G}{y}{S}$.
  By Lemma \ref{lemma:store-lookup-pure-val},
  we can show that
  $\exists v$ such that $\sta(\tVal{y}) = v$.
  Now we conclude this case by \ruleref{lift-set}.
\end{proof}


\subsection{Confluence}
\label{sec:confluence-proof}

\subsubsection{Properties of Evaluation Context Focuses}

A focus is one way to split a term $t$ into an evaluation context $e$ and the focused term $s$
such that $t = \evctx{s}$.
When evaluating a let-binding,
either the bindee or the continuation get reduced,
resulting in multiple possible focuses for the same term.
Each focus is determined by the part we choose to reduce for each let expression.

\begin{definition}[Subfocus]
	We say $t = \ctx{e_1}{t_1}$ is a \emph{subfocus} of $t = \ctx{e_2}{t_2}$,
  written $\ctx{e_1}{t_1} \subfocus \ctx{e_2}{t_2}$, iff
  $\exists e$ such that $t_1 = \ctx{e}{t_2}$ and $e_2 = \ctx{e_1}{e}$.
\end{definition}

\begin{definition}[Branched focus]
	We say $t = \ctx{e_1}{t_1}$ and $t = \ctx{e_2}{t_2}$ are a pair of branched focuses,
  written $\ctx{e_1}{t_1} \branch \ctx{e_2}{t_2}$,
  iff $\exists e, D, e'_1, e'_2$ such that
  $t = \evctx{\tLetMode{D}{x}{\ctx{e'_1}{t_1}}{\ctx{e'_2}{t_2}}}$,
  $e_1 = \evctx{\tLetMode{D}{x}{\ctx{e'_1}{}}{\ctx{e'_2}{t_2}}}$,
  and $e_2 = \evctx{\tLetMode{D}{x}{\ctx{e'_1}{t_1}}{\ctx{e'_2}{}}}$.
\end{definition}

\begin{lemma}[Extending subfocus]
  \label{lemma:extending-subfocus}
	Given two focuses of the same term $t = \evctxN{1} = \evctxN{2}$,
  we can show that for any $e'$,
  $\evctxN{1} \subfocus \evctxN{2}$ implies
  $\evpctx{\evctxN{1}} \subfocus \evpctx{\evctx{2}}$.
\end{lemma}

\begin{proof}
  Then for some $e$,
  we have $s_1 = \evctx{s_2}$,
  and $e_2 = \ctx{e_1}{e}$
  First of all,
  we can show that
  $\evctx{t} = \evctx{\evctxN{1}} = \evctx{\evctxN{2}}$,
  which means that the extended focuses are still for the same term.
  $s_1 = \evctx{s_2}$ is unaffected by the extension.
  And $\ctx{e'}{e_2} = \ctx{e'}{\ctx{e_1}{e}}$ is trivial.
\end{proof}

\begin{lemma}[Extending branched focus]
  \label{lemma:extending-branched-focus}
  Given two focuses of the same term $t = \evctxN{1} = \evctxN{2}$,
  for any $e'$ we can show that
  $\evctxN{1} \branch \evctxN{2}$ implies
  $\evpctx{\evctxN{1}} \branch \evpctx{\evctxN{2}}$.
\end{lemma}

\begin{proof}
  Analogous to proof of Lemma \ref{lemma:extending-subfocus}.
\end{proof}

\begin{lemma}[Different focuses of the same term]
  \label{lemma:analyze-focus}
	Given two focuses $\ctx{e_1}{s_1}$ and $\ctx{e_2}{s_2}$ of the same term $t$,
  we can show that one of the followings holds:
  \begin{enumerate}[(i)]
  \item $\ctx{e_1}{s_1} \subfocus \ctx{e_2}{s_2}$;
  \item $\ctx{e_2}{s_2} \subfocus \ctx{e_1}{s_1}$;
  \item $\ctx{e_1}{s_1} \branch \ctx{e_2}{s_2}$;
  \item $\ctx{e_2}{s_2} \branch \ctx{e_1}{s_1}$.
  \end{enumerate}
\end{lemma}

\begin{proof}
  We begin by induction on the first evaluation context $e_1$.

  \emph{Case $e_1 = []$}.
  Then we can demonstrate that $\ctx{e_1}{s_1} \subfocus \ctx{e_2}{c_2}$
  by setting $e = e_2$.

  \emph{Case $e_1 = \tLetMode{D}{x}{e'_1}{u}$}.
  We proceed the proof by induction on $e_2$.
  \begin{itemize}
  \item \emph{Case $e_2 = []$}.
    Then we conclude by showing that $\evctxN{2} \subfocus \evctxN{1}$.

  \item \emph{Case $e_2 = \tLetMode{m}{x}{e'_2}{u}$}.
    We can invoke IH and use Lemma \ref{lemma:extending-subfocus}
    and Lemma \ref{lemma:extending-branched-focus} to conclude this case.

  \item \emph{Case $e_2 = \tLetMode{m}{x}{u'}{e'_2}$}.
    Then we can show that $u' = \ctx{e'_1}{s_1}$
    and $u = \ctx{e'_2}{s_2}$.
    We can therefore conclude this case immediately
    by showing that $\evctxN{1} \branch \evctxN{2}$.
  \end{itemize}

  \emph{Case $e_1 = \tLetMode{m}{x}{u}{e'_1}$}.
  Analogous to the previous case. They are symmetric.
\end{proof}

\subsubsection{Confluence Theorems}

Now we demonstrate that the reduction of the calculus is \emph{confluent}.
This means that though we can arbitrarily interleave the reduction of both sides of the let bindings,
the result of the program will always be the same,
which implies the absence of data races and thus achieves the safe concurrency guarantee.

First of all,
in different reduction paths,
the store bindings can be lifted to the store in different orders.
For example,
we may introduce variables in different orders,
or have the set-bindings arranged in different ways.
Regardless of these permutations,
the respective stores should possess the same interpretations.
Thus, to carry out the proof
we have consider the equivalence between the stores up to the permutations.
To this end,
we define the \emph{equivalence} between two evaluation stores up to permutation.
\begin{definition}[Equivalent stores]
  We say two stores $\sta_1$
  and $\sta_2$ are equivalent,
  written $\sta_1 \sequiv \sta_2$,
  iff
  \begin{enumerate}[(1)]
  \item $\sta_2$ is permuted from $\sta_1$;
  \item $\forall x \in \bvar{\sta_1}$,
    $\sta_1(\tVar{x}) = \sta_2(\tVar{x})$.
  \end{enumerate}
\end{definition}

\begin{definition}[Equivalent configurations]
	We say two configurations
  $\st{\sta_1}{t_1}$ and $\st{\sta_2}{t_2}$ are equivalent
  iff $\sta_1 \sequiv \sta_2$ and $t_1 = t_2$.
\end{definition}

\begin{fact}[Store equivalence is an equivalence relation]
	The equivalence between stores $\sta_1 \sequiv \sta_2$
  is an equivalence relation:
  it is reflexive, symmetric and transitive.
\end{fact}

\begin{lemma}[Value lookup in equivalent stores]
  \label{lemma:val-lookup-eqv-store}
	If $\sta_1 \sequiv \sta_2$
  then $\sta_1(\tVal{x}) = \sta_2(\tVal{x})$.
\end{lemma}

\begin{proof}
  By the definition of store equivalence and the lookup function.
  The result of the lookup function does not rely on the order of bindings.
\end{proof}

\begin{lemma}[Variable lookup in equivalent stores]
  \label{lemma:var-lookup-eqv-store}
	If $\sta_1 \sequiv \sta_2$,
  then $\sta_1(\tVar{x}) = \sta_2(\tVar{x})$.
\end{lemma}

\begin{proof}
  This follows directly from the definition of store equivalence.
\end{proof}

\begin{lemma}[Typing implies separation]
  \label{lemma:typing-implies-ninter}
  Given non-value terms $s_1$ and $s_2$,
  $\typDft{\tLetMode{\parmode}{x}{\ctx{e_1}{s_1}}{\ctx{e_2}{s_2}}}{T}$
  implies $\ninter{\G}{s_1}{s_2}$.
\end{lemma}

\begin{proof}
  By induction on the typing derivation.

  \emph{Case \ruleref{let}}.
  Then we have $\ninter{\G}{\ctx{e_1}{s_1}}{\ctx{e_2}{s_2}}$.
  Since neither $s_1$ nor $s_2$ is a value,
  by straightforward induction on $e_1$ and $e_2$ we can show that
  $\cv{s_1} \cap \dom{\G} \subseteq \cv{\ctx{e_1}{s_1}} \cap \dom{\G}$
  and
  $\cv{s_2} \cap \dom{\G} \subseteq \cv{\ctx{e_2}{s_2}} \cap \dom{\G}$.
  Now we can conclude by applying Corollary \ref{coro:subset-pres-ninter}.

  \emph{Case \ruleref{sub}}.
  By the IH.
\end{proof}



\begin{lemma}[Exclusiveness of root capability]
  \label{lemma:ninter-root}
	If $\ninter{\G}{C_1}{C_2}$,
  then
  (i) $\univ{} \in C_1$ implies $\sub{\G}{C_2}{\set{}}$,
  and (ii) $\univ{} \in C_2$ implies $\sub{\G}{C_1}{\set{}}$.
\end{lemma}

\begin{proof}
	By induction on the derivation.

  \emph{Case \ruleref{ni-symm}}.
  By swapping the two conclusions in the IH.

  \emph{Case \ruleref{ni-set}}.
  Then $\overline{\ninter{\G}{x}{C_2}}^{x \in C_1}$.
  By applying the IH repeatedly,
  we can show that
  for any $x \in C_1$, we have
  (1) $\univ\in \set{x}$ implies $\sub{\G}{C_2}{\set{}}$,
  and (2) $\univ\in C_2$ implies $\sub{\G}{\set{x}}{\set{}}$.
  First, if $\univ\in C_1$, then exists $x \in C_1$ such that $\univ \in \set{x}$,
  which implies that $\sub{\G}{C_2}{\set{}}$ by the IH.
  Besides, if $x \in C_2$, we can show that $\sub{\G}{\set{x}}{\set{}}$ for any $x \in C_1$.
  By invoking Lemma \ref{lemma:subcapt-join-left} repeatedly we can show that
  $\sub{\G}{C_1}{\set{}}$.
  We can therefore conclude this case.

  \emph{Case \ruleref{ni-degree}}.
  Then $C_1 = \set{x}$ and $C_2 = \set{y}$.
  We can show that $x \ne \univ$ and $y \ne \univ$ by definition,
  and thus conclude this case.

  \emph{Case \ruleref{ni-var}}.
  Then $C_1 = \set{x}$, $C_2 = \set{y}$,
  $x : T \in \G$,
  and $\ninter{\G}{\cs{T}}{y}$.
  Firstly, we have $x \ne \univ$ and therefore $\univ \notin \set{x}$.
  Secondly, if $\univ \in \set{y}$ then by the IH we can show that
  $\sub{\G}{\cs{T}}{\set{}}$.
  Now we conclude this case using the \ruleref{sc-var} rule.

  \emph{Case \ruleref{ni-reader}}.
  Then $C_1 = \set{x}$, $C_2 = \set{y}$,
  $\subDft{\set{x}}{\set{\rdroot}}$,
  and $\subDft{\set{y}}{\set{\rdroot}}$.
  By straightforward induction on $\subDft{\set{x}}{\set{\rdroot}}$
  we can show that $\univ \notin \set{x}$
  and similarly for $\univ \notin \set{y}$.
  We can therefore conclude this case.
\end{proof}

\begin{lemma}
  \label{lemma:ninter-self}
	Given any $x :_D \tCap{\set{\univ}}{S} \in \G$,
  $\ninter{\G}{x}{x}$ is impossible.
\end{lemma}

\begin{proof}
	By induction on the derivation of $\ninter{\G}{x}{x}$,
  wherein we derive contradiction for each case.

  \emph{Case \ruleref{ni-symm} and \ruleref{ni-set}}. By the IH.

  \emph{Case \ruleref{ni-degree}}.
  Then $x \in D$.
  This is contradictory with the well-formedness.

  \emph{Case \ruleref{ni-var}}.
  Then $\ninter{\G}{\set{\univ}}{x}$.
  By Lemma \ref{lemma:ninter-root}
  we can show that $\sub{\G}{\set{x}}{\set{}}$.
  By straightforward induction on it we can derive the contradiction.

  \emph{Case \ruleref{ni-reader}}.
  Then $\sub{\G}{\set{x}}{\set{\rdroot}}$.
  By straightforward induction on this derivation,
  we can derive the contradiction too.
\end{proof}

\begin{lemma}
  In an inert environment $\G$,
	given any $x : \tCap{\set{y}}{\tRdr{S}} \in \G$,
  and $y : {S'}\capt \set{\univ} \in \G$,
  $\ninter{\G}{x}{y}$ is impossible.
\end{lemma}

\begin{proof}
  By induction on the derivation of
  $\ninter{\G}{x}{y}$.

  \emph{Case \ruleref{ni-symm} and \ruleref{ni-set}}.
  By the IH.

  \emph{Case \ruleref{ni-degree}}.
  Then $x \in D_y$ or $y \in D_x$,
  where $D_x$ and $D_y$ denotes the separation degree of $x$ and $y$
  in the environment respectively.
  We can show that $y \notin D_x$ by the well-formedness.
  If $y \in D_x$, then by the inertness we can show that
  $\ninter{\G}{y}{y}$.
  Now we derive the contradiction via Lemma \ref{lemma:ninter-self}.

  \emph{Case \ruleref{ni-var}}.
  Then either $\ninter{\G}{\univ}{x}$,
  or $\ninter{\G}{y}{y}$.
  In the first case we can show that $\sub{\G}{\set{x}}{\set{}}$ by Lemma \ref{lemma:ninter-root},
  from which we can derive the contradiction.
  In the second case,
  we invoke Lemma \ref{lemma:ninter-self} to derive the contradiction.

  \emph{Case \ruleref{ni-reader}}.
  Then we have $\sub{\G}{\set{y}}{\set{\rdroot}}$.
  By induction on the derivation of it, we can show contradiction in each case.
\end{proof}

\begin{corollary}
  \label{coro:ninter-star-neq}
	If (i) $\ninter{\G}{x}{y}$,
  (ii) $x : S_1\capt\set{\univ} \in \G$
  and (iii) $y : S_2\capt\set{\univ} \in \G$,
  then we have $x \ne y$.
\end{corollary}

\begin{corollary}
  \label{coro:ninter-rw-neq}
	If (i) $\ninter{\G}{x}{y}$,
  (ii) $x : \tRdr{S_1}\capt\set{z} \in \G$
  and (iii) $y : S_2\capt\set{\univ} \in \G$,
  then $z \ne y$.
\end{corollary}

\begin{theorem}[Diamond property of reduction]
  Given two equivalent configurations
  $\st{\sta_1}{t} \equiv \st{\sta_2}{t}$,
  if
  (1) $\sta_1\vdash t$ and $\sta_2 \vdash t$;
  (2) $\st{\sta_1}{t} \red \st{\sta'_1}{t_1}$;
  and (3) $\st{\sta_2}{t} \red \st{\sta'_2}{t_2}$,
  then
  either $t_1 = t_2$,
  or there exists $\sta''_1, t', \sta''_2$ such that
  (1) $\st{\sta'_1}{t_1} \red \st{\sta''_1}{t'}$,
  (2) $\st{\sta'_2}{t_2} \red \st{\sta''_2}{t'}$,
  and (3) $\st{\sta''_1}{t'} \equiv \st{\sta''_2}{t'}$.
\end{theorem}

\begin{proof}
  Begin with a case analysis on the derivation of $\st{\sta_1}{t} \red \st{\sta'_1}{t_1}$.

  \emph{Case \ruleref{apply}, \ruleref{tapply} and \ruleref{open}}.
    Proceed by a case analysis on the derivation of $\st{\sta_2}{t} \red \st{\sta'_2}{t_2}$.
    \begin{itemize}
    \item \emph{Case \ruleref{apply}, \ruleref{tapply}, \ruleref{open} and \ruleref{get}}.
      We only present the proof when both derivations are derived by the \ruleref{apply} case,
      and all other possibilities follow analogously.
      Firstly, from the preconditions we know that
      $t = \ctx{e_1}{x_1\,y_1}$ for some $e_1, x_1, y_1$,
      $\sta_1(\tVal{x_1}) = \lambda(z: T_1) s_1$,
      and $t'_1 = \ctx{e_1}{\fSubst z {y_1} {s_1}}$.
      Also, we have
      $t = \ctx{e_2}{x_2\,y_2}$ for some $e_2, x_2, y_2$,
      $\sta_2(\tVal{x_2}) = \lambda(z: T_2) s_2$,
      and $t'_2 = \ctx{e_2}{\fSubst z {y_2} {s_2}}$.
      By invoking Lemma \ref{lemma:analyze-focus},
      we can show that either $e_1 = e_2$ and $x_1 \, y_1 = x_2 \, y_2$,
      or $\ctx{e_1}{x_1\,y_1} \branch \ctx{e_2}{x_2\,y_2}$,
      or $\ctx{e_2}{x_2\,y_2} \branch \ctx{e_1}{x_1\,y_1}$.
      \begin{itemize}
      \item When $e_1 = e_2$, $x_1 = x_2$ and $y_1 = y_2$.
      \item When $\ctx{e_1}{x_1\,y_1} \branch \ctx{e_2}{x_2\,y_2}$.
        Then $t = e[\tLetMode{D}{x}{\ctx{e_1}{x_1\,y_1}}{\ctx{e_2}{x_2\,y_2}}]$.
        Therefore, we have $t_1 = \evctx{\tLetMode{D}{x}{\ctx{e_1}{\fSubst z {y_1} s_1}}{\ctx{e_2}{x_2\,y_2}}}$.
        We let
        $t' = e[\ \tLetMode{D}{x}{\ctx{e_1}{\fSubst z {y_1} s_1}}{\ctx{e_2}{\ctx{e_2}{\fSubst z {y_2} s_2}}}\ ]$,
        and can derive that
        $\st{\sta_1}{t_1}\red \st{\sta_1}{t'}$.
        Similarly, we can show that
        $\st{\sta_2}{t_2}\red \st{\sta_2}{t'}$,
        and thus conclude by the fact that $\sta_1 \sequiv \sta_2$.
      \item When $\ctx{e_2}{x_2\,y_2} \branch \ctx{e_1}{x_1\,y_1}$.
        Analogous to the previous case.
      \end{itemize}

    \item \emph{Case \ruleref{lift-let}}.
      We only present the proof of the \ruleref{apply} case,
      and other cases follows analogously.
      Then $t = \ctx{e_1}{x_1\,y_1}$ for some $e_1, x_1, y_1$,
      $\sta_1(\tVal{x_1}) = \lambda(z: T_1) s_1$,
      and $t_1 = \ctx{e_1}{\fSubst z {y_1} s_1}$.
      Also, $t = \ctx{e_2}{\tLetMode{D}{x}{v}{u}}$,
      $\sta'_2 = \sta_2, \tVal{x} \mapsto v$,
      and $t_2 = \ctx{e_2}{u}$.
      Now we proceed the proof by invoking Lemma \ref{lemma:analyze-focus}
      to analyze the relationship of two focuses $t = \ctx{e_1}{x_1\,y_1} = \ctx{e_2}{\tLetMode{D}{x}{v}{u}}$.
      \begin{itemize}
      \item \emph{When $\ctx{e_1}{x_1\,y_1} \subfocus \ctx{e_2}{\tLetMode{m}{x}{v}{u}}$.}
        This implies that for some $e$
        we have $x_1\,y_1 = \evctx{\tLetMode{m}{x}{v}{u}}$,
        which is impossible.

      \item \emph{When $\ctx{e_2}{\tLetMode{m}{x}{v}{u}} \subfocus \ctx{e_1}{x_1\,y_1}$.}
        Then for some $e$ we have
        $\boxed{\tLetMode{m}{x}{v}{u}} = \ctx{e}{x_1\,y_1}$
        and $e_1 = \ctx{e_2}{e}$.
        By analyzing the equality,
        we can show that
        $t = \ctx{e_2}{\tLetMode{m}{x}{v}{\evpctx{x_1\,y_1}}}$
        and $e_1 = \ctx{e_2}{\tLetMode{m}{x}{v}{e'}}$.
        Set $t' = \ctx{e_2}{\evpctx{\fSubst z {y_1} {s_1}}}$.
        First, we can derive that
        $\st{\sta_1}{\ctx{e_1}{\fSubst z {y_1} {s_1}}} \red \st{\sta_1, \tVal{x}\mapsto v}{\ctx{e_2}{\evpctx{\fSubst z {y_1} {s_1}}}}$
        by showing that $\fv{v} \subseteq \dom{\sta_1}$ and then apply the \ruleref{lift-let} rule.
        Then, we can derive that
        $\st{\sta_2, \tVal{x}\mapsto v}{\ctx{e_2}{\evpctx{x_1\,y_1}}} \red \st{\sta_2, \tValM{D}{x}\mapsto v}{\ctx{e_2}{\evpctx{\fSubst z {y_1} s_1}}}$
        by applying Lemma \ref{lemma:val-lookup-eqv-store} and the \ruleref{apply} rule.
        Now we can show that $(\sta_1, \tVal{x}\mapsto v) \sequiv (\sta_2, \tVal{x}\mapsto v)$
        and conclude this case.

      \item \emph{When $\ctx{e_1}{x_1\,y_1} \branch \ctx{e_2}{\tLetMode{m}{x}{v}{u}}$.}
        Then $t = \ctx{e}{\tLetMode{n}{y}{\ctx{e'_1}{x_1\,y_1}}{\ctx{e'_2}{\tLetMode{m}{x}{v}{u}}}}$
        for some $e, n, y, e'_1$ and $e'_2$.
        Let
        $\sta''_1 = \sta_1, \tValM{m}{x}\mapsto v$,
        $\sta''_2 = \sta_2, \tValM{m}{x}\mapsto v$,
        and $t' = e[\textsf{let}_{E} {y} = {e'_1[\fSubst z {y_1} s_1]}$
        $\,\textsf{in}\,{\ctx{e'_2}{u}}]$.
        For both $t_1, t_2$, we can invoke the other rule to reduce to $t'$.
        We can therefore conclude this case.

      \item \emph{When $\ctx{e_2}{\tLetMode{m}{x}{v}{u}} \branch \ctx{e_1}{x_1\,y_1}$.}
        Analogous to the above case.
        These two cases are symmetric.
      \end{itemize}

    \item \emph{Case \ruleref{rename}, \ruleref{lift-var} and \ruleref{lift-set}}.
      Analogous to the previous case.
    \end{itemize}

  \emph{Case \ruleref{get}}.
    Then $t = \ctx{e_1}{\tRead{x}}$,
    $t_1 = \ctx{e_1}{v}$,
    $\sta(\tVal{x}) = \reader x'$
    and $\sta(\tVar{x'}) = v$.
    We proceed the proof by case analysis on the second derivation,
    wherein the \ruleref{apply}, \ruleref{tapply} and the \ruleref{open} cases are symmetric to the above proven cases.
    The proof of \ruleref{lift-let}, \ruleref{lift-var} and \ruleref{rename} is analogous to the previous case.
    Notably, in the \ruleref{lift-var}
    we are sure that the lifted variable is not $x$ since it is a freshly created local variable.
    \begin{itemize}
    \item \emph{Case \ruleref{get}}.
      To prove this case,
      we begin by invoking Lemma \ref{lemma:analyze-focus} to analyze the relation between two focuses,
      wherein in each possibility we can invoke Lemma \ref{lemma:var-lookup-eqv-store}
      to demonstrate the equality between the result of variable lookup in two equivalent stores,
      and therefore invoke \ruleref{get} again to reduce both sides to the same term.

    \item \emph{Case \ruleref{lift-set}}.
      Then $t = \ctx{e_2}{\tWrite{y_1}{y_2}}$,
      $\sta_2(\tVal{y_2}) = w$,
      $\sta'_2 = \sta_2, \tSet{y_1} = w$,
      and $t_2 = \ctx{e_2}{w}$.
      Now we have to show that $y_1 \ne x'$ so that the read and the write can be swapped.
      We first invoke Lemma \ref{lemma:analyze-focus} to analyze the relation between two focuses.
      The two cases where $\ctx{e_1}{\tRead{x}} \subfocus \ctx{e_2}{\tWrite{y_1}{y_2}}$
      or $\ctx{e_2}{\tWrite{y_1}{y_2}}\subfocus \ctx{e_1}{\tRead{x}}$
      are impossible.
      Now we show the proof when
      $\ctx{e_1}{\tRead{x}}\branch \ctx{e_2}{\tWrite{y_1}{y_2}}$,
      whereas the other case is analogous.
      We have $t = e[\textsf{let}_{m}{x} = {\ctx{e'_1}{\tRead{x}}}\,$
      $\textsf{in}\,\ctx{e'_2}{\tWrite{y_1}{y_2}}]$.
      Let $t' = \ctx{e}{\tLetMode{m}{x}{\ctx{e'_1}{v}}{\ctx{e'_2}{w}}}$.
      By invoking Lemma \ref{lemma:val-lookup-eqv-store}
      we can show that
      $\sta_1(\tVal{y_2}) = w$,
      and therefore derive that
      $\st{\sta_1}{t_1} \red \st{\sta_1, \tSet{y_1} = w}{t'}$ by the \ruleref{lift-set} rule.
      Since $\sta_2 \vdash t$ we know that $\typ{\G_2}{t}{T_2}$ for some $\G_2$ and $T_2$.
      By Lemma \ref{lemma:evctx-typing-inversion} we can show that
      $\typ{\G_2, \Delta_2}{\tLetMode{m}{x}{\ctx{e'_1}{\tRead{x}}}{\ctx{e'_2}{\tWrite{y_1}{y_2}}}}{U}$
      for some $\Delta_2$ and $U$.
      Now we invoke Lemma \ref{lemma:typing-implies-ninter} and Lemma \ref{lemma:ninter-to-elem}
      to show that $\ninter{\G_2, \Delta_2}{x}{y_1}$.
      By induction on $\match{}{\sta}{\G_2}$, we can show that
      $x : \set{x'}\,\tRdr{S_0} \in \G_2$,
      and $y_1 : \set{\univ}\,\tRef{S_1} \in \G_2$ for some $S_0, S_1$.
      Now we invoke Corollary \ref{coro:ninter-rw-neq} to show that
      $x' \ne y_1$.
      Therefore, we can show that
      $\sta'_2(\tVar{x'}) = \sta(\tVar{x'}) = v$
      where $\sta'_2 = \sta_2, \tSet{y_1} = w$.
      We can thus derive that
      $\st{\sta'_2}{t_2}\red \st{\sta'_2}{t'}$.
      Finally, we can show that $\sta_1, \tSet{y_1} = w$ and $\sta_2, \tSet{y_1} = w$ are still equivalent
      and conclude this case.
    \end{itemize}

  \emph{Case \ruleref{rename}}.
    Then $t = \ctx{e_1}{\tLetMode{D_1}{x_1}{y_1}{s_1}}$ for some $e_1$,
    and $t_1 = \ctx{e_1}{\fSubst x y u}$.
    We proceed by case analysis on the other reduction.
    The proof of \ruleref{apply}, \ruleref{tapply}, \ruleref{open} and \ruleref{get} cases
    are symmetric to the proof of the previous cases.
    \begin{itemize}
    \item \emph{Case \ruleref{rename}}.
      Then $t = \ctx{e_2}{\tLetMode{D_2}{x_2}{y_2}{s_2}}$ for some $e_2$.
      Now we use Lemma \ref{lemma:analyze-focus} to analyze the relation between two focuses.
      \begin{itemize}
      \item \emph{When $\ctx{e_1}{\tLetMode{D_1}{x_1}{y_1}{s_1}} \subfocus \ctx{e_2}{\tLetMode{D_2}{x_2}{y_2}{s_2}}$}.
        Then we have $\tLetMode{D_1}{x_1}{y_1}{s_1} = \evctx{\tLetMode{D_2}{x_2}{y_2}{s_2}}$ for some $e$,
        and $e_2 = \ctx{e_1}{e}$.
        If $e = []$, then we have $t_1 = t_2$.
        Then we have $t = e[\textsf{let}_{D_1} {x_1} = {y_1}\,\textsf{in}\,{\ctx{e'}{\tLetMode{D}{x_2}{y_2}{s_2}}}]$ for some $e'$.
        Let $t' = e_1[{e'}[\fSubst{x_1}{y_1} \fSubst{x_2}{\fSubst{x_1}{y_1} y_2} s_2]]$.
        Now we inspect whether $x_1 = y_2$,
        and in both cases we can derive that
        $\st{\sta_1}{t_1}\red \st{\sta_1}{t'}$,
        and $\st{\sta_2}{t_2}\red \st{\sta_1}{t'}$,
        which allow us to conclude this case.
      \item \emph{When $\ctx{e_2}{\tLetMode{D_2}{x_2}{y_2}{s_2}} \subfocus \ctx{e_1}{\tLetMode{D_1}{x_1}{y_1}{s_1}}$}.
        Analogous to the previous case.
      \item \emph{When $\ctx{e_1}{\tLetMode{D_1}{x_1}{y_1}{s_1}} \branch \ctx{e_2}{\tLetMode{D_2}{x_2}{y_2}{s_2}}$}.
        Then for both directions, we can apply \ruleref{rename} to rename the variable in the other branch,
        wherein the renaming in one branch is independent from the renaming in the other branch.
        We can therefore conclude this case.
      \item \emph{When $\ctx{e_2}{\tLetMode{D_2}{x_2}{y_2}{s_2}} \branch \ctx{e_1}{\tLetMode{D_1}{x_1}{y_1}{s_1}}$}.
        Analogous to the previous case.
      \end{itemize}

    \item \emph{Case \ruleref{lift-let}}.
      Then $t = \ctx{e_2}{\tLetMode{D_2}{x_2}{v_2}{s_2}}$,
      $\sta'_2 = \sta_2, \tValM{D}{x_2}\mapsto v_2$,
      and $t_2 = \ctx{e_2}{s_2}$.
      Now we apply Lemma \ref{lemma:analyze-focus} to analyze the relation between two focuses.
      \begin{enumerate}[(i)]
      \item \emph{When $\ctx{e_1}{\tLetMode{D_1}{x_1}{y_1}{s_1}}\subfocus \ctx{e_2}{\tLetMode{D_2}{x_2}{v_2}{s_2}}$}.
        Then, by inspecting the equality we can show that
        $t = \ctx{e_1}{\tLetMode{D_1}{x_1}{y_1}{\ctx{e}{\tLetMode{D_2}{x_2}{v_2}{s_2}}}}$.
        We set $t' = \ctx{e_1}{\evctx{\fSubst {x_1} {y_1} s_2}}$.
        Importantly, since $\cv{v_2} \cap D_2 \subseteq \dom{\sta_2}$,
        we can show that
        $\fSubst{x_1}{y_1} {v_2} = v_2$
        and $\fSubst {x_1} {y_1} D_2 = D_2$.
        We can derive that
        $\st{\sta_1}{t_1}\red \st{\sta_1, \tValM{D_2}{x_2}\mapsto v_2}{t'}$
        and $\st{\sta_2, \tValM{D_2}{x_2}\mapsto v_2}{t_2}\red \st{\sta_2, \tValM{D_2}{x_2}\mapsto v_2}{t'}$.
        This case can thus be concluded
        as $\sta_1, \tValM{D_2}{x_2}\mapsto v_2$
        and $\sta_2, \tValM{D_2}{x_2}\mapsto v_2$
        are still equivalent.

      \item \emph{Other cases}.
        In these cases, the renaming and the binding lifting
        do not influence each other.
        So in each case, for both directions we apply the corresponding rule in the other direction to conclude.
      \end{enumerate}

    \item \emph{Case \ruleref{lift-var} and \ruleref{lift-set}}.
      In this two cases,
      we can show that the variable that is looked up is already in the store,
      and by the well-formedness their value does not mention the renamed variable,
      thus staying unaffected by the renaming.
      We can apply the corresponding rule to conclude each case.
    \end{itemize}

  \emph{Case \ruleref{lift-let} and \ruleref{lift-var}}.
    Proceed by case analysis on the other reduction derivation.
    The proof of the
    \ruleref{apply}, \ruleref{tapply}, \ruleref{open}, \ruleref{get},
    and \ruleref{rename} cases are again symmetric to the previous proof.
    In the remaining cases,
    we can always swap the order the two lifted store bindings
    while preserving store equivalence.

  \emph{Case \ruleref{lift-set}}.
    We do a case analysis on the other reduction derivation,
    wherein all but one cases can be proven symmetrically to the previous ones.
    The only unproven case is when both reductions are derived by the \ruleref{lift-set} rule.
    Then $t = \ctx{e_1}{\tWrite{x_1}{y_1}}$,
    $\sta_1(\tVal{y_1}) = v_1$,
    $\sta'_1 = \sta_1, \tSet{x_1} = v_1$,
    and $t_1 = \ctx{e_1}{v_1}$.
    Also,
    $t = \ctx{e_2}{\tWrite{x_2}{y_2}}$,
    $\sta_2(\tVal{y_2}) = v_2$,
    $\sta'_2 = \sta_2, \tSet{x_2} = v_2$,
    and $t_2 = \ctx{e_2}{v_2}$.
    By invoking Lemma \ref{lemma:analyze-focus}
    we can show that
    either $e_1 = e_2$,
    or $\ctx{e_1}{\tWrite{x_1}{y_1}}\branch \ctx{e_2}{\tWrite{x_2}{y_2}}$
    or $\ctx{e_2}{\tWrite{x_2}{y_2}}\branch \ctx{e_1}{\tWrite{x_1}{y_1}}$.
    In the first case we can conclude immediately
    since this implies that $t_1 = t_2$.
    Otherwise,
    we invoke Lemma \ref{lemma:evctx-typing-inversion},
    Lemma \ref{lemma:typing-implies-ninter},
    and Corollary \ref{coro:ninter-star-neq}
    to show that $x_1 \ne x_2$.
    Therefore, swapping the two set-bindings
    preserves the store equivalence.
    For both directions,
    we can apply the \ruleref{lift-set} rule
    to reduce to the same term.
    This case is therefore concluded.
\end{proof}

\begin{definition}[Reduction closures]
	We define $\st{\sta}{t}\redt \st{\sta'}{t'}$
  as the reflexive and transitive closure of
  $\st{\sta}{t}\red \st{\sta'}{t'}$.
  $\st{\sta}{t}\redm \st{\sta'}{t'}$ denotes
  the union of
  the reduction relation
  $\st{\sta}{t}\red \st{\sta'}{t'}$
  and the reflexive relation.
  In other words,
  $\st{\sta}{t}\redm \st{\sta'}{t'}$ means
  a reduction of zero or one step.
\end{definition}

\begin{corollary}[Diamond property of $\redm{}$]
  \label{coro:diamond-property-redm}
  Given two equivalent configurations
  $\st{\sta_1}{t} \sequiv \st{\sta_2}{t}$,
  if
  (1) $\sta_1\vdash t$ and $\sta_2 \vdash t$;
  (2) $\st{\sta_1}{t}\redm \st{\sta'_1}{t_1}$;
  and (3) $\st{\sta_2}{t}\redm \st{\sta'_2}{t_2}$,
  then
  there exists $\sta''_1, t', \sta''_2$ such that
  (1) $\st{\sta'_1}{t_1}\redm \st{\sta''_1}{t'}$,
  (2) $\st{\sta'_2}{t_2}\redm \st{\sta''_2}{t'}$,
  and (3) $\st{\sta''_1}{t'} \sequiv \st{\sta''_2}{t'}$.
\end{corollary}

\begin{fact}
	The transitive and reflexive closure of
  $\st{\sta}{t}\redm \st{\sta'}{t'}$
  equals $\st{\sta}{t}\redt \st{\sta'}{t'}$.
\end{fact}

\begin{lemma}[Store equivalence preserves reduction]
  \label{lemma:store-equiv-pres-red}
  If
  (i) $\st{\sta_1}{t}\red \st{\sta_1'}{t'}$
  and (ii) $\sta_1 \sequiv \sta_2$
  then there exists $\sta_2'$
  such that
  (i) $\st{\sta_2}{t}\red \st{\sta_2'}{t'}$
  and (ii) $\sta_1' \sequiv \sta_2'$.
\end{lemma}

\begin{proof}
	By straightforward case analysis on the reduction derivation,
  wherein in each case we apply the same typing rule.
  In the cases where store lookup is involved,
  we use Lemma \ref{lemma:val-lookup-eqv-store}
  and Lemma \ref{lemma:var-lookup-eqv-store}
  to show that the result is the same under the two equivalent stores.
  In the cases where the store is extended in the reduction step,
  we can straightforwardly show that the resulted stores are still equivalent to each other.
\end{proof}

\begin{corollary}
  \label{coro:store-equiv-pres-redt}
  If
  (i) $\st{\sta_1}{t}\redt \st{\sta_1'}{t'}$
  and (ii) $\sta_1 \sequiv \sta_2$
  then there exists $\sta_2'$
  such that
  (i) $\st{\sta_2}{t}\redt \st{\sta_2'}{t'}$
  and (ii) $\sta_1' \sequiv \sta_2'$.
\end{corollary}

\begin{lemma}[Asymmetric diamond property of reduction closure]
  \label{lemma:asym-diamond-property-redt}
  Given two equivalent configurations
  $\st{\sta_1}{t} \sequiv \st{\sta_2}{t}$,
  if
  (1) $\sta_1\vdash t$ and $\sta_2 \vdash t$;
  (2) $\st{\sta_1}{t}\redm \st{\sta'_1}{t_1}$;
  and (3) $\st{\sta_2}{t}\redt \st{\sta'_2}{t_2}$,
  then
  there exists $\sta''_1, t', \sta''_2$ such that
  (1) $\st{\sta'_1}{t_1}\redt \st{\sta''_1}{t'}$,
  (2) $\st{\sta'_2}{t_2}\redm \st{\sta''_2}{t'}$,
  and (3) $\st{\sta''_1}{t'} \sequiv \st{\sta''_2}{t'}$.
\end{lemma}

\begin{proof}
	By induction on the length of the reduction
  $\st{\sta_2}{t}\redt \st{\sta'_2}{t_2}$.

  \emph{When there is zero step}.
  Then $\sta'_2 = \sta_2$ and $t_2 = t$.
  We can set $t' = t_1$,
  then reduce $t_2$ one step to $t_1$,
  and reduce $t_1$ zero step.

  \emph{When $\st{\sta_2}{t}\redt \st{\sta'_2}{t_2} = \st{\sta_2}{t}\redm \st{\sta_0}{t_0}\redt \st{\sta'_2}{t_2}$}.
  By Corollary \ref{coro:diamond-property-redm}
  we can show that there exists $t'$ such that
  $\st{\sta'_1}{t_1}\redm \st{\sta''_1}{t'}$,
  $\st{\sta_0}{t_0}\redm \st{\sta'_0}{t'}$,
  and $\sta''_1 \sequiv \sta'_0$
  Now we use the preservation theorem
  so that we can invoke IH to show that there exists $t''$ such that
  $\st{\sta'_0}{t'}\redt \st{\sta''_0}{t''}$,
  $\st{\sta'_2}{t_2}\redm \st{\sta''_2}{t''}$,
  and $\sta''_0 \sequiv \sta''_2$.
  Now we invoke Corollary \ref{coro:store-equiv-pres-redt} to show that
  there exists $\sta'''_1$ such that
  $\st{\sta''_1}{t'}\redt \st{\sta'''_1}{t''}$,
  and $\sta'''_1 \sequiv \sta''_0$.
  Therefore, we can show that
  $\st{\sta'_1}{t_1}\redt \st{\sta'''_1}{t''}$.
  Also, we can show that $\sta'''_1 \sequiv \sta''_2$,
  and thus conclude this case.
\end{proof}


\begin{proof}
	By induction on the length of the first reduction.

  \emph{When the length is zero}.
  This case can be trivially concluded.

  \emph{When $\st{\sta_1}{t}\redm \st{\sta_0}{t_0}\redt \st{\sta'_1}{t_1}$}.
  We first invoke Lemma \ref{lemma:asym-diamond-property-redt} to show that
  $\exists t'$ such that
  $\st{\sta_0}{t_0}\redt \st{\sta'_0}{t'}$,
  $\st{\sta'_2}{t_2}\redm \st{\sta''_2}{t'}$,
  and $\sta'_0 \sequiv \sta''_2$.
  Then, we use the preservation lemma so that we can invoke IH,
  showing that
  $\exists t''$ such that
  $\st{\sta'_1}{t_1}\redt \st{\sta''_1}{t''}$,
  $\st{\sta'_0}{t_0}\redt \st{\sta''_0}{t''}$,
  and $\sta''_1 \sequiv \sta''_0$.
  Now by Corollary \ref{coro:store-equiv-pres-redt} we can show that
  $\exists \sta'''_2$ such that
  $\st{\sta''_2}{t'}\redt \st{\sta'''_2}{t''}$
  and $\sta'''_2 \sequiv \sta''_0$.
  Therefore, we have
  $\st{\sta'_2}{t_2}\redt \st{\sta'''_2}{t''}$
  and can show that $\sta'''_2 \sequiv \sta''_1$,
  thus concluding this case.
\end{proof}

\begin{lemma}[Answers do not reduce]
	Given any store $\sta$ and an answer $a$,
  $\st{\sta}{a}\red \st{\sta'}{t'}$ is impossible.
\end{lemma}

\begin{proof}
	By straightforward case analysis on the reduction derivation,
  wherein none of the rules reduces an answer.
\end{proof}

\begin{corollary}
	\label{coro:ans-red-self}
  Given any store $\sta$ and an answer $a$,
  $\st{\sta}{a}\redt \st{\sta'}{t}$ implies that
  $\sta' = \sta$ and $t = a$.
\end{corollary}

\uniquenessofanswer*

\begin{proof}
	We first use Theorem \ref{thm:confluence} to show that
  $\exists \sta_1', \sta_2', t'$ such that
  $\st{\sta_1}{a_1}\redt \st{\sta'_1}{t'}$
  and $\st{\sta_2}{a_1}\redt \st{\sta'_2}{t'}$
  where $\sta'_1 \sequiv \sta'_2$.
  Now we use Corollary \ref{coro:ans-red-self} to show that
  $t' = a_1 = a_2$,
  $\sta_1 = \sta'_1$ and $\sta_2 = \sta'_2$,
  which conclude our goal.
\end{proof}


\end{document}